\documentclass[a4paper, cleveref, autoref, thm-restate]{lipics-v2021}

\pdfoutput=1 %
\hideLIPIcs  %

\title{Parameterized Complexity of Directed Traveling Salesman Problem}

\author{Václav Blažej}{Institute of Informatics, University of Warsaw, Warsaw, Poland \and \url{https://vaclavblazej.github.io/}}{v.blazej@uw.edu.pl}{https://orcid.org/0000-0001-9165-6280}{}
\author{Andreas Emil Feldmann}{University of Sheffield, Sheffield, United Kingdom \and \url{https://sites.google.com/site/aefeldmann/home}}{feldmann.a.e@gmail.com}{https://orcid.org/0000-0001-6229-5332}{}
\author{Foivos Fioravantes}{Faculty of Information Technology, Czech Technical University in Prague, Prague, Czech Republic}{foivos.fioravantes@fit.cvut.cz}{https://orcid.org/0000-0001-8217-030X}{}
\author{Paweł Rzążewski}{Warsaw University of Technology, Warsaw, Poland \and Institute of Informatics, University of Warsaw, Warsaw, Poland}{pawel.rzazewski@pw.edu.pl}{https://orcid.org/0000-0001-7696-3848}{}
\author{Ondřej Suchý}{Faculty of Information Technology, Czech Technical University in Prague, Prague, Czech Republic}{ondrej.suchy@fit.cvut.cz}{https://orcid.org/0000-0002-7236-8336}{Co-funded by the European Union under the project Robotics and advanced industrial production (reg. no. CZ.02.01.01/00/22\_008/0004590).}

\authorrunning{V. Blažej, A.\,E. Feldmann, F. Fioravantes, P. Rzążewski, and O. Suchý} %
\Copyright{Václav Blažej, Andreas Emil Feldmann, Foivos Fioravantes, Paweł Rzążewski, and Ondřej Suchý} %
\category{} %
\relatedversiondetails[cite=BlazejFFRS2025]{Full Version}{https://arxiv.org/abs/2506.22127} %

\funding{\textit{Václav Blažej and Paweł Rzążewski}: The work on this article is a part of project BOBR that has received funding from the European Research Council (ERC) under the European Union’s Horizon 2020 research and innovation programme (grant agreement No. 948057).
}%
\acknowledgements{
The project was initiated at the workshop Homonolo 2023. We are grateful to the organizers and other participants for a nice and productive atmosphere.
}%

\nolinenumbers %

\EventEditors{Ho-Lin Chen, Wing-Kai Hon, and Meng-Tsung Tsai}
\EventNoEds{3}
\EventLongTitle{36th International Symposium on Algorithms and Computation (ISAAC 2025)}
\EventShortTitle{ISAAC 2025}
\EventAcronym{ISAAC}
\EventYear{2025}
\EventDate{December 7--10, 2025}
\EventLocation{Tainan, Taiwan}
\EventLogo{}
\SeriesVolume{359}
\ArticleNo{11}
\usepackage{xspace} %

\usepackage{tikz}
\usetikzlibrary{calc, shapes, backgrounds, patterns, arrows, decorations, decorations.pathreplacing, decorations.pathmorphing, fit, arrows.meta,patterns,patterns.meta,decorations.pathmorphing, bending}
\pgfdeclarelayer{bg}
\pgfdeclarelayer{bbg}
\pgfsetlayers{bbg,bg,main}
\graphicspath{{images/}}
\tikzstyle{defbox} = [every node/.style={fill=gray!30,draw,very thin,font=\small,minimum height=0.6cm,,minimum width=0.7cm}]
\tikzstyle{NPh} = [draw=red!80,very thick,fill=red!30]
\tikzstyle{Wh} = [draw=orange!95,very thick]
\tikzstyle{XP} = [fill=orange!25]
\tikzstyle{FPT} = [fill=green!25]
\tikzstyle{def} = [-{stealth'},every node/.style={draw,circle,inner sep=2pt}]
\tikzstyle{hide} = [draw=none,fill=none,rectangle]
\tikzstyle{vertex} = [draw,fill=gray,circle,inner sep=2pt] %

\usepackage{bm}
\usepackage{tikz}
\usetikzlibrary{calc,scopes,backgrounds,decorations.pathmorphing,fit,math}
\usetikzlibrary{decorations.markings}
\def\resultbox<#1>(#2)(#3)(#4;#5)(#6;#7)(#8;#9)%
{
    \tikzstyle{cpxity} = [minimum height=0.4cm,inner sep=0pt,align=center,node distance=0];
    \def\totalwidth{3.90cm};
    \node[inner ysep=3pt,inner xsep=0pt,text width=\totalwidth,align=center, fill=black!7] (#2) at (#3) {\small #1};
    \tikzmath{\boxwidth = \totalwidth/3;\nboxwidth = -\totalwidth/3;}
    \node[text width=\boxwidth,cpxity,#4,anchor=north west] (#2_TSP) at (#2.south west) {\small\ttfamily #5};
    \node[text width=\boxwidth,cpxity,below=0 of #2,#6] (#2_sTSP) {\small\ttfamily #7};
    \node[text width=\boxwidth,cpxity,#8,anchor=north east] (#2_WRP) at (#2.south east) {\small\ttfamily #9};
    \node[fit=(#2)(#2_TSP)(#2_sTSP)(#2_WRP),draw,inner sep=-0.01] {};
    \draw (#2_TSP.north west) -- (#2_WRP.north east);
    \draw (#2_sTSP.north west) -- (#2_sTSP.south west);
    \draw (#2_sTSP.north east) -- (#2_sTSP.south east);
    \coordinate (#2_top) at (#2.north);
    \coordinate (#2_bot) at (#2_sTSP.south);
}

\usepackage{hyperref}
\usepackage{cleveref}
\usepackage{todonotes}
\newcommand{\itemstyle}[1]{\textcolor{lipicsGray}{\sffamily\bfseries\upshape\mathversion{bold}#1}}

\newcommand{\N}{\ensuremath{\mathbb{N}}}
\newcommand{\Z}{\ensuremath{\mathbb{Z}}}

\newcommand{\NP}{\textsf{NP}\xspace}
\newcommand{\NPh}{\mbox{\NP-hard}\xspace}

\newcommand{\FPT}{\textsf{FPT}\xspace}
\newcommand{\XP}{\textsf{XP}\xspace}
\newcommand{\Wh}[1][1]{\mbox{\textsf{W[#1]}-hard}\xspace}
\newcommand{\Whness}[1][1]{\mbox{\textsf{W[#1]}-hardness}\xspace}

\newcommand{\TSP}{\textsc{Traveling Salesman Problem}\xspace}
\newcommand{\TSPshort}{\textsc{TSP}\xspace}
\newcommand{\DMTSP}{\textsc{Directed \TSP}\xspace}
\newcommand{\DMTSPshort}{\textsc{DTSP}\xspace}

\newcommand{\DTSPshort}{\textsc{DTSP}\xspace}
\newcommand{\DMsTSP}{\textsc{Directed Subset \TSPshort}\xspace}
\newcommand{\DMsTSPshort}{{\textsc{Ds\TSPshort}}\xspace}

\newcommand{\DsTSPshort}{{\textsc{Ds\TSPshort}}\xspace}
\newcommand{\DWRP}{\textsc{Directed Waypoint Routing Problem}\xspace}
\newcommand{\DWRPshort}{\textsc{DWRP}\xspace}
\newcommand{\CDSshort}{\textsc{CDS}\xspace}

\newcommand{\budget}{\ensuremath{b}} %
\newcommand{\wFn}{\ensuremath{\omega}} %
\newcommand{\WP}{\ensuremath{W}} %
\newcommand{\cFn}{\ensuremath{\kappa}} %

\newcommand{\tw}{\operatorname{tw}}
\newcommand{\pw}{\operatorname{pw}}

\newtheorem{rrule}{Reduction Rule}
\crefname{rrule}{Reduction Rule}{Reduction Rules}
\Crefname{rrule}{Reduction Rule}{Reduction Rules}
\Crefname{claim}{Claim}{Claims}

\usepackage{framed}
\usepackage{tabularx}

\newlength{\RoundedBoxWidth}
\newsavebox{\GrayRoundedBox}
\newenvironment{GrayBox}[1]%
   {\setlength{\RoundedBoxWidth}{.93\columnwidth}
    \def\boxheading{#1}
    \begin{lrbox}{\GrayRoundedBox}
       \begin{minipage}{\RoundedBoxWidth}}%
   {   \end{minipage}%
    \end{lrbox}%
    \begin{center}%
    \begin{tikzpicture}%
       \node(Text)[draw=black!20,fill=white,rounded corners,inner sep=2ex,text width=\RoundedBoxWidth]%
             {\usebox{\GrayRoundedBox}};%
        \coordinate(x) at (current bounding box.north west);%
        \node [draw=white,rectangle,inner sep=3pt,anchor=north west,fill=white]%
        at ($(x)+(6pt,.75em)$) {\boxheading};%
    \end{tikzpicture}%
    \end{center}}

\newenvironment{defproblemx}[1]{\noindent\ignorespaces%
                                \FrameSep=6pt%
                                \parindent=0pt%
                \begin{GrayBox}{#1}%
                \begin{tabular*}{\columnwidth}{!{\extracolsep{\fill}}@{\hspace{.1em}} >{\itshape} p{1.5cm} p{0.86\columnwidth} @{}}%
            }{
                \end{tabular*}%
                \end{GrayBox}%
                \ignorespacesafterend%
            }

\newcommand{\problemQuestion}[3]{%
  \begin{defproblemx}{#1}
    Input: & #2 \\
    Question: & #3
  \end{defproblemx}%
}

\usepackage{etoolbox}
\def\ShortVersion{1} %
\ifdefined\ShortVersion
\newcommand{\sv}[1]{#1}
\newcommand{\lv}[1]{}
\newcommand{\appendixText}{}
\newcommand{\toappendix}[1]{\gappto{\appendixText}{{#1}}}
\else
\newcommand{\sv}[1]{}
\newcommand{\lv}[1]{#1}
\newcommand{\appendixText}{}
\newcommand{\toappendix}[1]{#1}
\fi
\newcommand{\appmark}{$\star$}

\begin{document}

\begin{titlepage}

\maketitle

\begin{abstract}
The \DMTSP (\DMTSPshort) is a variant of the classical \TSP in which the edges in the graph are directed and a vertex and edge can be visited multiple times.
The goal is to find a directed closed walk of minimum length (or total weight) that visits every vertex of the given graph at least once.
In a yet more general version, \DWRP (\DWRPshort), some vertices are marked as terminals and we are only required to visit all terminals. Furthermore, each edge has its capacity bounding the number of times this edge can be used by a solution.

While both problems (and many other variants of \TSPshort) were extensively investigated, mostly from the approximation point of view, there are surprisingly few results concerning the parameterized complexity. Our starting point is the result of Marx et al. [APPROX/RANDOM 2016] who proved that \DMTSPshort is \Wh[1] parameterized by distance to pathwidth 3.
In this paper we aim to initiate the systematic complexity study of variants of \DMTSP with respect to various, mostly structural, parameters.

We show that \DWRPshort is \FPT parameterized by the solution size, the \emph{feedback edge number}  and the \emph{vertex integrity} of the underlying undirected graph.
Furthermore, the problem is \XP parameterized by treewidth.

On the complexity side, we show that the problem is \Wh[1] parameterized by the distance to constant treedepth.
\end{abstract}

\end{titlepage}

\section{Introduction}\label{sec:intro}

The \TSP(\TSPshort) is possibly the most famous combinatorial optimization problem. It is one of the few problems having entire monographs devoted solely to its aspects~\cite{ApplegateBCC2006,Cook2012,GutinP2007,ShmoysLKL1985}.
In this paper we focus on the directed variant, which can be formulated as follows.
The input instance is a directed graph $G$ with edges assigned positive integers (\emph{weights}) and an integer $\budget$, and we are to decide whether there is a closed walk in $G$ visiting every vertex \emph{at least} once and having total weight at most $\budget$.
We refer to this problem as \DMTSP (\DMTSPshort).
We also consider two natural generalizations of the problem.
In the first of them, \DMsTSP (\DMsTSPshort), we are given alongside the graph also a subset of its vertices $\WP$ called \emph{waypoints} or \emph{terminals} and the walk is only required to visit these vertices.
Obviously, an instance of \DMTSPshort can be interpreted as the instance of \DMsTSPshort by letting $\WP = V(G)$.
In some applications it is also preferred not to traverse any (or some) of the edges of the graph too often.
To this end, in \DWRP (\DWRPshort) (introduced in~\cite{AmiriFJS18}), we augment the input with a positive integer for each edge specifying its \emph{capacity}.
The walk is then required to traverse each edge at most as many times as its capacity.
A simple argument (provided later) shows that each edge is used at most $n:=|V(G)|$ times in any solution.
Hence, setting the capacities to $n$ is equivalent to dropping the capacity constraints completely.

Our focus is on exact parameterized algorithms for the problem.
Let us first summarize the known results for the undirected variant of the problem.
Since \textsc{Hamiltonian Cycle (HamC)} is \NPh{} even on planar graphs of maximum degree three~\cite{GareyJT76}, \TSPshort{} is para-\NPh{} with respect to (w.r.t.) the maximum degree of the input graph as well as w.r.t.\ the maximum number of visits to any vertex.
As the solution size is equal to the number of vertices, the dynamic programming of Bellman~\cite{Bellman62} and Held and Karp~\cite{HeldK62} which runs in $O(2^n \cdot n^2)$ time can be considered as an FPT algorithm w.r.t.\ solution size.
It can be also used to solve \textsc{Subset \TSPshort} in $O(2^{|W|} \cdot |W|^2)$ after an initial computation of distances between all pairs of waypoints.
This can be improved to $O(2^{\sqrt{|W|}\log |W|} n^{O(1)})$ if the graph is planar and the weights are bounded by a polynomial~\cite{KleinM14}.

On the one hand, a single-visit variant of \TSPshort{} is known to be \FPT{} w.r.t.\ treewidth (see also~\cite{CyganKN18}) %
and this result also extends to our multi-visit variant and to \textsc{Waypoint Routing (WPR)}~\cite{SchierreichS22}.
On the other hand, \TSPshort{} is hard on cliques with edge weights 1 and 2, shown by a simple reduction from \textsc{HamC}.
Therefore, it is para-\NPh w.r.t.\ any parameter which is constant on cliques.
Moreover, already \textsc{HamC} is \Wh{} w.r.t.\ cliquewidth~\cite{FGLSZ2018} (while \FPT w.r.t.\ the neighborhood diversity~\cite{Lampis12}).

Recently, the (undirected) \TSPshort{} was studied from the perspective of kernelization~\cite{VanBevernS2024,BlazejCKSSV22}.
It was shown to admit a polynomial kernel w.r.t.\ the feedback edge number and w.r.t.\ distance to constant size components.
Here, for a graph class $\mathcal{C}$, \emph{distance to $\mathcal{C}$} refers to the minimum number of vertices that must be deleted from the input graph to obtain a graph that belongs to~$\mathcal{C}$.
In the above case, the class contains all graphs with every connected component of size at most~$c$, for any fixed~$c$.
\TSPshort was shown not to admit polynomial kernel w.r.t.\ vertex integrity, unless \NP $\subseteq$ co\NP/poly, whereas \textsc{Subset TSP} does not admit polynomial kernel even w.r.t.\ distance to disjoint cycles, under the same assumption~\cite{BlazejCKSSV22}.
Here \emph{vertex integrity} of a graph is the minimum $k$ such that we can delete at most $k$ vertices from the graph in such a way that each connected component of the resulting graph has at most $k$ vertices.

Much less is known about the parameterized complexity of the directed case of the problem.
The mentioned hardness results translate to the directed case by symmetrically orienting the input graph.
While the algorithmic results do not transfer automatically to the directed setting, it is often possible to adapt them.
For example, the dynamic programming w.r.t.\ number of vertices~\cite{Bellman62,HeldK62} applies also to this case.
This can be again improved to~$O(2^{\sqrt{|W|}\log |W|} n^{O(1)})$ for \DsTSPshort{} in planar graphs~\cite{MarxPP22}; also lifting the weight-restriction from the undirected case.
In contrast, \DWRPshort is \NPh already for two waypoints~\cite{AmiriFJS18}. %

Similarly, the single visit variant is \FPT w.r.t. the treewidth.
However, Marx et al.~\cite{MarxSS16} showed, that \DMTSPshort is \Wh{} w.r.t.\ the pathwidth $\pw$ of the input graph  and there is no algorithm for \DMTSPshort running in $f(\pw) \cdot n^{o(\pw)}$ time for any computable function~$f$, unless the Exponential Time Hypothesis (ETH)~\cite{ImpagliazzoP2001,ImpagliazzoPZ01} fails.
More precisely, the proof shows that this is the case even w.r.t.\ distance to graphs of pathwidth~$3$. %
Furthermore, a more careful analysis of the proof (see \Cref{thm:d_to_pw} in the appendix) shows that the \Whness{} holds true even w.r.t.\ distance to pathwidth $2$. %
Thus, the fpt-algorithm for \textsc{WRP} w.r.t.\ treewidth cannot be transferred to the directed case.

Apparently, these are the only results concerning the parameterized complexity of \DMTSPshort, and, in particular, no fpt-algorithms w.r.t.\ structural parameters are known for the problem. 

\newcommand{\cpxbox}[6]{
    \node[#1]      at ($(#5)-(2*.08,2*.1)$) {\phantom{#6}};
    \node[#2]      at ($(#5)-(1*.08,1*.1)$) {\phantom{#6}};
    \node[#3] (#4) at (#5) {#6};
}
{%
\begin{figure}[tbh]
    \centering
    \begin{tikzpicture}
        \tikzstyle{fpt} = [fill=green!30]
        \tikzstyle{whard} = [fill=orange!40]
        \tikzstyle{unknown} = [fill=yellow!10]
        \def\ystep{-1.4}

        \resultbox<\textit{legend}>(leg)(-5,4.9*\ystep)%
        (unknown;\TSPshort)%
        (unknown;\DMsTSPshort)%
        (unknown;\DWRPshort);

        \resultbox<Vertex Cover Number>(VC)(0,0.9*\ystep)%
        (fpt;der)%
        (fpt;der)%
        (fpt;der);

        \resultbox<Dist. to Const. td>(DTD)(0,1.9*\ystep)%
        (unknown;?)%
        (unknown;?)%
        (whard;Thm~\ref{thm:dwp_td_hardness});

        \resultbox<Dist. to Const. pw>(DPW)(0,2.9*\ystep)%
        (whard;Thm~\ref{thm:d_to_pw})%
        (whard;der)%
        (whard;der);

        \resultbox<Feedback Edge Set No.>(FES)(5,1*\ystep)%
        (fpt;der)%
        (fpt;der)%
        (fpt;Thm~\ref{thm:fes});

        \resultbox<Feedback Vertex Set No.>(FVS)(5,2.3*\ystep)%
        (unknown;?)%
        (unknown;?)%
        (unknown;?);

        \resultbox<Treewidth + Max. Degree>(DTW)(5,3.8*\ystep)%
        (unknown;?)%
        (unknown;?)%
        (unknown;?);

        \resultbox<Vertex Integrity>(VI)(-4.5,2*\ystep)%
        (fpt;der)%
        (fpt;der)%
        (fpt;Thm~\ref{thm:vi});

        \resultbox<Treedepth>(TD)(-4.5,3*\ystep)%
        (unknown;?)%
        (unknown;?)%
        (whard;der);

        \resultbox<Pathwidth>(PW)(-4.5,4*\ystep)%
        (whard;der)%
        (whard;der)%
        (whard;der);

        \resultbox<Treewidth>(TW)(0,4.9*\ystep)%
        (whard;der)%
        (whard;der)%
        (whard;Thm~\ref{thm:tw});

        \begin{scope}[very thick]
            \draw ($(FES_bot)-(0.0,0)$) -- ($(FVS_top)+(0.0,0)$);
            \draw ($(FVS_bot)-(0.8,0)$) -- ($(TW_top)+(0.0,0)$);
            \draw ($(VC_bot)+(-0.8,0)$) -- ($(VI_top)+(0.8,0)$);
            \draw ($(VC_bot)+(0.0,0)$) -- ($(DTD_top)+(0.0,0)$);
            \draw ($(VC_bot)+(0.8,0)$) -- ($(FVS_top)+(-0.8,0)$);
            \draw ($(DTD_bot)+(-0.8,0)$) -- ($(TD_top)+(0.8,0)$);
            \draw ($(DTD_bot)+(0.0,0)$) -- ($(DPW_top)+(0.0,0)$);
            \draw ($(DPW_bot)+(-0.8,0)$) -- ($(PW_top)+(0.8,0)$);
            \draw ($(VI_bot)+(0.0,0)$) -- ($(TD_top)+(0.0,0)$);
            \draw ($(TD_bot)+(0.0,0)$) -- ($(PW_top)+(0.0,0)$);
            \draw ($(PW_bot)+(0.8,0)$) -- ($(TW_top)+(-0.8,0)$);
            \draw ($(DTW_bot)-(0.8,0)$) -- ($(TW_top)+(0.8,0)$);
        \end{scope}
    \end{tikzpicture}
    \caption{
        Complexity picture for the studied problems.
        Box fill shows tractability: green stands for \FPT, orange for \Wh and in \XP, and yellow for in \XP, but open whether \FPT or \Wh.
        The boxes either contain references to the corresponding theorem, or \texttt{der} if the result can be derived from the other results. 
        All the hardness results apply even for unit weights, whereas the algorithms do not assume any restriction on the weights.
    }%
    \label{fig:parameters}
\end{figure}
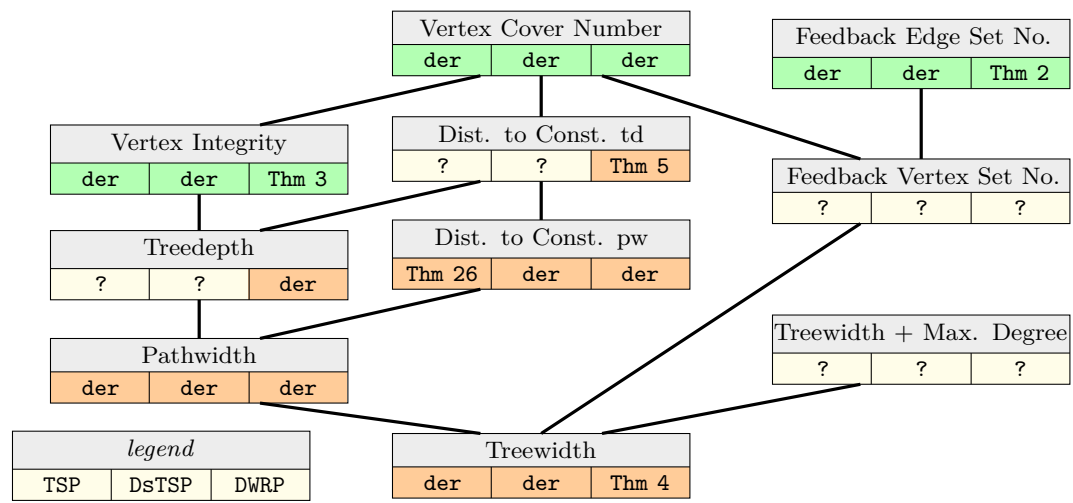
}

\subparagraph*{Our Contribution.}
We initiate the systematic study of the parameterized complexity of \DMTSPshort and its generalizations focusing on structural parameters.
In particular, we concentrate on parameters measuring the structure of the underlying undirected graph of the input.
We aim to fill the apparent gap in knowledge and address the lack of fpt-algorithms.

First, we revisit the standard parameterization, i.e., the parameterization by the solution size.
The budget bounds the number of edges in the solution, hence also the number of vertices visited, and in particular the size of $\WP$ (for any yes-instance of \DTSPshort and \DsTSPshort)
Thus for these problems, the known single exponential algorithms~\cite{Bellman62,HeldK62} and subexponential algorithm for planar graphs~\cite{MarxPP22} show that the problem is FPT w.r.t.\ the budget $\budget$.
For \DWRPshort{} the situation is not as straightforward.
Nevertheless, we show the following result.\footnote{Proofs of statements marked with (\appmark)~can be found in the appendix.}

\begin{restatable}[\appmark]{theorem}{thmsolsize}
\label{thm:solution_size}
 In $e^k4^kk^{O(\log k)}\cdot m \log m$ time we can decide whether the given $m$-edge instance of \DWRPshort has a solution consisting of at most $k$ edge occurrences.
\end{restatable}

Next we turn to the structural parameters (see \Cref{fig:parameters} for an overview).
The result of Marx et al.~\cite{MarxSS16} shows that \DTSPshort{} is already \Wh{} w.r.t.\ distance to graphs of pathwidth $3$, i.e., for graphs with rather restricted structure. In fact, as we show in \Cref{thm:d_to_pw} (appendix), this holds also for unweighted graphs.
Therefore, in search for fpt-algorithms for the problems, the structure of the graph after removal of parameter many vertices or edges should be more restricted than having pathwidth~$3$.
Our first positive result is w.r.t. the feedback edge number, i.e., the minimum number of edges that needs to be removed to obtain a forest.

\begin{restatable}{theorem}{thmfes}
\label{thm:fes}
\DWRPshort can be solved in $k^{O(k)} + n^{O(1)}$ time, where $k$ is the feedback edge number of the input graph.
\end{restatable}

Then, we present our main algorithmic result, showing that \DWRPshort is \FPT w.r.t.\ vertex integrity.
Recall that vertex integrity allows, roughly speaking, to delete parameter many vertices so as to obtain a graph with connected components of parameter size.

\begin{restatable}{theorem}{thmvi}
\label{thm:vi}
 \DWRPshort can be solved in $k^{O(k^6)} \cdot n \cdot \log^2 n (\log n + \log U)$ time, where $k$ is the vertex integrity of the input graph and $U$ is the maximum weight of an edge in $G$.
\end{restatable}

The algorithm first guesses a part of the solution ensuring its connectivity and then employs a well structured ILP (\emph{$N$-fold IP}~\cite{KouteckyLO18}) to find the cheapest way to complete the initial guess into a solution.

We also show that, while the problems are probably not FPT w.r.t.\ treewidth alone, they are FPT w.r.t.\ treewidth combined with the bound on the number of visits to every city.

\begin{theorem}[\appmark]
\label{thm:tw}
\DWRPshort on $n$-vertex graphs of treewidth $t$ can be solved in $n^{O(t)}$ time.
\end{theorem}

This means that all considered problems are in XP w.r.t.\ treewidth and, hence, w.r.t.\ all structural parameters considered in this paper.

We complement the positive results by showing that \DWRPshort is \Wh w.r.t.\ distance to constant treedepth, i.e., if the input graph becomes of constant treedepth after deleting parameter many vertices.

\begin{restatable}{theorem}{thmtd}
\label{thm:dwp_td_hardness}
    \DWRPshort is \Wh with respect to modulator to constant treedepth, even if all weights are 1.
\end{restatable}

Note that this result is incomparable to that of Marx et al.~\cite{MarxSS16} in that it holds for a more restricted graph class, while it only applies to a more general problem.

\subparagraph*{Further Related Work.}
Another popular formulation of \textsc{(Directed) \TSPshort} is that the input graph is complete and, thus, we are given a matrix of weights for all ordered pairs of vertices.
Since this matrix is asymmetric for the case of directed graphs, the directed variant of \TSPshort{} is often called \textsc{Asymmetric TSP (ATSP)}.
It is easy to switch from graph $G$ to a complete graph on the same set of vertices by performing the so called \emph{metric closure}, i.e., assigning the edge $(u,v)$ of the complete graph the weight given by distance from vertex $u$ to vertex $v$ in the graph $G$.
The weights obtained this way satisfy \emph{triangle inequality}.
Whenever this is the case, there is no reason to revisit any vertex, as it cannot make the tour shorter.
However, if the matrix is not required to satisfy the triangle inequality, it makes a difference whether visiting a vertex more than once is allowed or not (single visit version).

For the single visit version with general matrices a folklore result shows that it cannot be approximated within any polynomial factor, unless P=\NP \cite{DBLP:books/daglib/0023376}.
For undirected case with triangle inequality it was know that there is $\frac32$-approximation since 1976~\cite{Christofides76}, which was later improved~\cite{GharanSS11,MomkeS11, Mucha12} up to the currently best known factor $1.4$~\cite{SeboV14} for the graphic case, where the weights originate from the distances in an unweighted graph.
However, for the directed case, it was long open, whether a constant factor approximation exists for matrices satisfying the triangle inequality~\cite{AnariG15,AsadpourGMGS17}.
Constant-factor approximation algorithms appeared only recently, first for unweighted directed graphs~\cite{Svensson15}, later for general distances~\cite{SvenssonTV20}, with the currently best approximation factor being $22+\varepsilon$ for any $\varepsilon > 0$~\cite{TraubV22}. The best known lower bound on the approximability is $\frac{75}{74}$~\cite{KarpinskiLS15}.

It is surprising, given the lack of parameterized analysis for the problem, that there are several studies concerning parameterized approximation algorithm for \DMTSPshort.
First, B\"{o}ckenhauer~\cite{BockenhauerHKK07} considered (undirected) TSP with deadlines on the latest time to visit some of the vertices and showed a $2.5$-approximation algorithm in $k! \cdot n^{O(1)}$ time, where $k$ is the number of deadlines, and complemented that by several innaproximability results.
The main result of the already mentioned paper of Marx et al.~\cite{MarxSS16} is a polynomial-time constant-factor approximation for \DMTSPshort in nearly embeddable graphs, where both the factor and the time depend on the actual parameters of the class of graphs considered.
Bonnet et al.~\cite{BonnetLP18} provided a $(\log r)$-approximation for ATSP in $2^{O(\frac{n}{r})}$ time for any $r$.
Behrendt et al.~\cite{BehrendtC0LLW23} considered the case that there are few edges with (significantly) asymmetric costs and presented a $2.5$-approximation parameterized by the vertex cover number of the graph formed by such edges and a $3$-approximation parameterized by the minimum number of asymmetric edges in a minimum arborescence of the graph.

Further results are known for a variant of \DWRPshort where the waypoints have a prescribed order~\cite{AmiriFJPS18} and for the undirected case~\cite{AmiriFS20}.

Another studied variant of \DTSPshort is the \textsc{Many-visit \TSPshort}, where each vertex has a prescribed number of visits and the task is to find a closed walk which visits each vertex \emph{exactly} the given number of times. After a significant effort, an algorithm was found for the problem that is exponential only in the number of vertices, i.e., independent of the demands which can be huge~\cite{BergerKMV20,KowalikLNSW22}. The problem was recently studied from parameterized perspective~\cite{MannensNSS21} under the name \textsc{Connected Flow}. Here only a subset of vertices has a prescribed number of visits, the other vertices can be visited arbitrary number of times (or not visited at all) and edges have capacities (as in \DWRPshort). They showed that the problem is para-\NPh w.r.t. number of vertices with demand even in unweighted graph, while FPT w.r.t.\ this parameter, if there are no capacities. Further, they gave a polynomial kernel w.r.t. vertex cover number, an $n^{O(\tw)}$-time algorithm and a matching $n^{o(\tw)}$-time ETH lower bound.

Several parameterized and fine-grained complexity studies considered the local search for the undirected single-visit \TSPshort parameterized the size of the neighborhood to be searched~\cite{deBergBJW21,BonnetIJK19,CyganKS19,GuoHNS13,LanciaD24,Marx08}. None of these studies seems to extend to directed graphs.

\section{Notations, Preliminaries, and Basic Observations}

\toappendix{
\section{Omitted Material from Preliminaries} %

\begin{definition}[Tree decomposition]\label{def:treeDecomposition}
	A \emph{tree decomposition} of a graph~$G$ is a triple $(T, \beta, r)$, where~$T$ is a tree rooted at node~$r$ and $\beta \colon V(T) \to 2^{V(G)}$ is a mapping that assigns to each node $x$ of the tree its \emph{bag} $\beta(x)$ and satisfies:
	\begin{enumerate}
		\item $\bigcup_{x \in V(T)} \beta(x) = V(G)$;
		\item For every $\{u, v\} \in E(G)$ there exists a node $x \in V(T)$, such that $u, v \in \beta(x)$;
		\item For every $u \in V(G)$ the nodes $\{x \in V(T) \mid u \in \beta(x)\}$ form a connected subtree of~$T$.
	\end{enumerate}
\end{definition}

To distinguish between the vertices of a tree decomposition and the vertices of the underlying graph, we use the term \emph{node} for the vertices of a given tree decomposition.

The \emph{width} of a tree decomposition $(T, \beta, r)$ is $\max_{x \in V(T)} |\beta(x)|-1$.
The \emph{treewidth} of a graph~$G$ (denoted $\tw(G)$) is the minimum width of a~tree decomposition of~$G$ over all such decompositions.

For a tree decomposition $(T, \beta, r)$ and a node $x \in V(T)$, we denote by $\gamma(x)$ the union of vertices in $\beta(x)$ and in~$\beta(y)$ for all descendants~$y$ of~$x$ in~$T$. By $E_x$, we denote the set of edges introduced in the subtree of $T$ rooted in $x$. Altogether, we denote by $G_x$ the graph $(\gamma(x), E_x)$. %

\begin{definition}[Nice tree decomposition~{\cite[p.~168]{CyganFKLMPPS15}}]\label{def:niceTreeDecomposition}
	A tree decomposition of a graph~$G$ is \emph{nice} if $\beta(r) = \emptyset$, and each node $x \in V(T)$ is of one of the following five types:
	\begin{itemize}
		\item \emph{Leaf node}---$x$ has no children and $\beta(x) = \emptyset$;
		\item \emph{Introduce vertex node}---$x$ has exactly one child~$y$ and $\beta(x) = \beta(y) \cup \{u\}$ for some $u \in V(G) \setminus \beta(y)$;
		\item \emph{Introduce edge node}---$x$ has exactly one child~$y$, $\beta(x) = \beta(y)$, and an edge $\{u, v\} \in E(G)$ for $u, v \in \beta(x)$ is introduced (i.e., $E_x = E_y \cup \{\{u,v\}\}$);
		\item \emph{Forget node}---$x$ has exactly one child~$y$ and $\beta(x) = \beta(y) \setminus \{u\}$ for some $u \in \beta(y)$;
		\item \emph{Join node}---$x$ has exactly two children~$y, z$ and $\beta(x) = \beta(y) = \beta(z)$,
	\end{itemize}
	and each edge $e \in E(G)$ is introduced exactly once. %
\end{definition}

We note that any given tree decomposition of width $t$ can be transformed into a nice one of the same width in $O(t^2\cdot n)$ time~\cite[Lemma 7.4, see also the discussion on page 168]{CyganFKLMPPS15}.
}

We follow standard graph-theoretic notation~\cite{Diestel2025}. Also, we refer the reader to~\cite{CyganFKLMPPS15} for a textbook on parameterized complexity.

\subparagraph*{(Generalized) $N$-fold integer programming ($N$-fold IP)} is the problem of minimizing a separable convex objective (for us it suffices to minimize a linear objective) over a set of constraints of a specific form.
Let $r \in \N$ and $s_i, t_i  \in \N$ for every $i \in [N]$.
The IP has $d = \sum_{i \in [N]} t_i$ variables partitioned into $N$ so-called \emph{bricks}.
The $i$-th brick is denoted $x^{(i)}$ and contains $t_i$ variables.
The constraints have the following form:
\begin{align*}
D_1 x^{(1)} + D_2 x^{(2)} + \cdots + D_N x^{(N)} &= \mathbf{b}_0    &&\text{ \emph{(linking (global) constraints)}}\\%
A_i x^{(i)}                        &= \mathbf{b}_i    & \forall i \in [N] &\text{ \emph{(local constraints)}}\\ %
\mathbf{l}_i \le x^{(i)}                 &\le \mathbf{u}_i  & \forall i \in [N] &\text{ \emph{(box constraints)}}%
\end{align*}
where we have $D_i \in \Z^{r \times t_i}$, $A_i \in \Z^{s_i \times t_i}$, $\mathbf{b}_0 \in \Z^{r}$, $\mathbf{b}_i \in \Z^{s_i}$ and $\mathbf{l}_i,\mathbf{u}_i \in \Z^{t_i}$ for every $i \in [N]$.
Let us denote $s = \max_{i \in [N]} s_i$ and recall that the dimension is $d = \sum_{i \in [N]} t_i$.
There are $r$ global constraints and, apart from these, each variable is only involved in at most $s$ local constraints and one box constraint.
We call variables %
not appearing in linking constraints \emph{local}.

The current best algorithm solving the $N$-fold IP in \( (rs\Delta)^{O(r^2s+rs^2)} d\log d \log D \log val_{\max}\) time is by Eisenbrand et al.~\cite[Cor. 97]{EisenbrandHKKLO19}; see also \cite{EisenbrandHK18,KouteckyLO18}, where $D$ is the maximum range of a variable, \mbox{$\Delta = \max_{i \in [N]} \Big( \max \big( \|D_i\|_\infty, \|A_i\|_\infty \big) \Big)$}, and $val_{\max}$ is the maximum feasible value of the objective.

\subparagraph*{Variants of the Problem and Basic Observations.}
Let us formally introduce the considered problems and establish notation used in the next sections.
\problemQuestion{\DMTSP (\DMTSPshort)}
{A directed graph $G=(V,E)$, edge weights $\wFn \colon E \to \N$, and a budget $\budget \in \N$.}
{Is there a closed directed walk in $G$ of total weight at most $\budget$ that traverses each vertex of $G$ at least once?}%
\problemQuestion{\DWRP (\DWRPshort)}
	{A directed graph $G=(V,E)$, set of waypoints $\WP \subseteq V$, edge weights $\wFn \colon E \to \N$, edge capacities $\cFn \colon E \to \N$, budget $\budget \in \N$.}
	{Is there a closed directed walk $C$ in $G$ of total weight at most $\budget$, that traverses each vertex in $\WP$ (at least once) and such that for each edge $e$ the number of times $C$ traverses $e$ is at most $\cFn(e)$?}%
To avoid degenerate cases, we assume that $|W|$ is always at least 2 and that the input graph is strongly connected.
We use undirected parameters, and, when speaking about a width of a directed graph, we mean the width of the underlying undirected graph.

\begin{lemma}[{Folklore, see, e.g., Bang-Jansen and Gutin~\cite[Exercise 1.12]{JensenG2009}}]\label{lem:cycle_decomposition}
 Any closed directed walk (in particular a solution to the problem) can be decomposed into simple cycles, such that the number of times an edge is traversed by the walk equals the number of the cycles it appears in.
\end{lemma}

\begin{observation}[\appmark]\label{obs:n_traversals}
 In each optimal solution to \DWRPshort, each vertex $v$ is visited at most $|W \setminus \{v\}|$ times.
\end{observation}

Hence, an instance of \DMsTSPshort can be interpreted as an equivalent instance of \DWRPshort, by letting $\cFn(e)=|V|$ for every arc.

\toappendix{
\begin{proof}[Proof of \Cref{obs:n_traversals}]
 Consider the decomposition of the optimal solution $S$ to cycles as guaranteed by \Cref{lem:cycle_decomposition}, denote it $\mathcal{C}=\{C_1, \ldots, C_q\}$.
 Consider an undirected bipartite graph $B$ with $V(B) = \mathcal{C} \cup V(G)$ and an edge connecting $u \in V(G)$ with $C \in \mathcal{C}$ if and only if $u \in V(C)$.
 Obviously $\deg_B v$ equals the number of cycles $v$ appears in, which in turn equals the number of times $v$ is visited by~$S$. 
 Let $T$ be a BFS tree for the BFS started from the vertex~$v$ in graph~$B$. 
 It might not contain all vertices of $V(G)$, it only contains vertices visited by $S$, and in particular it contains all vertices in $W$.
 Obviously $\deg_T v =\deg_B v$.
 
 If $v$ is a leaf of $T$, then we are done.
 Otherwise, let $T'$ be a tree obtained from $T$ by iteratively removing leaves that are not vertices of $W$.
 Let $\mathcal{C}'$ be the set of cycles that appear in $T'$ and $H$ be a multigraph obtained as the edge disjoint union of cycles in $\mathcal{C}'$.
 As $H$ is an edge disjoint union of cycles, it is balanced.
 We claim that it is also connected.
 In particular, there is a path from any vertex $u \in V(H)$ to $v$ in~$H$.
 Indeed, $u$ appears on some cycle $C \in \mathcal{C}'$ and there is a path $P_{T'}$ from $C$ to $v$ in~$T'$.
 By replacing each cycle on $P_{T'}$ by a part of the cycle between the appropriate vertices, we obtain a path from $u$ to $v$ in $H$.
 Thus there is an Eulerian trail $S'$ in $H$.
 Also, since~$T'$ contains all vertices in $W$, $H$ contains all vertices in~$W$, and $S'$ visits all vertices in~$W$.
 Moreover it traverses each edge at most as many times as~$S$.
 Thus it is a solution.
 Therefore, as $S$ is an optimal solution, we have $\mathcal{C}'=\mathcal{C}$.
 
 This, in particular, implies that $v \in V(T')$ and $\deg_{T'} v = \deg_T v$.
 However, as each leaf of~$T'$ is a vertex of $W\setminus\{v\}$, $T'$ has at most $|W\setminus\{v\}|$ leaves, which implies $\deg_{T'} v \le |W\setminus\{v\}|$.
\end{proof}
}%

\toappendix{
\section{FPT Algorithm with Respect to Solution Size}

\thmsolsize*

Our proof is based on the color coding technique introduced by Alon et al.~\cite{AlonYZ1995}, hence we need the following proposition.

\begin{proposition}[Naor et al.~\cite{NaorSS95}]\label{prop:hashing}
    Let $m,k$ be any positive integers.
    There exists a family $\mathcal{F}$ of colorings $f \colon [m]\to [k]$ of size $e^kk^{O(\log k)}\log m$ constructible in $e^kk^{O(\log k)}m\log m$ time such that for any $F\subseteq [m]$ with $|F|\leq k$ there exists $f \in \mathcal{F}$ such that $f|_F$ is injective.
\end{proposition}

\begin{proof}[Proof of \cref{thm:solution_size}]
 Since a solution which consist of at most $k$ edge occurrences can use each edge at most $k$ times, we assume that each capacity is at most $k$.
 Now we turn the input graph into a multigraph by replacing each edge $e$ with $\cFn(e)$ edges, each with the same tail and head as $e$, weight $\wFn(e)$, and capacity $1$.
 It is easy to see that there is a solution of weight $\budget$ in this new multigraph if and only if there is a solution of weight $\budget$ in the original graph.
 For simplicity, we denote this new multigraph $G=(V,E)$, the new weight function $\wFn$ and $\cFn$ is the capacity function assigning $1$ to every edge.

 Since there must be an arc entering each of the waypoints, if $|\WP| > k$, then there is no solution consisting of at most $k$ edge occurrences and we can answer NO.
 Hence, let us assume that $|\WP| \le k$ for the rest of the proof.

 Now we construct the family $\mathcal{F}$ of colorings as guaranteed by \Cref{prop:hashing} to color $E$ with $k$ colors.
 For each coloring $f \in \mathcal{F}$ we search for a \emph{colorful} solution to our problem, that is, a solution that uses at most one edge of each color.
 Obviously, any colorful solution is a solution for the original problem.
 In particular, it cannot use any edge more than once, as this would mean that it used more than one edge of the corresponding color.
 Thus, it obeys the capacities.
 Moreover, by the properties of $f$, any solution consisting of at most $k$ edges becomes colorful under at least one coloring from $\mathcal{F}$.
 Hence, the minimum weight colorful solution over all colorings from $\mathcal{F}$ is also a minimum weight solution using at most $k$ edges, and it is enough to compare the weight of this solution with the budget $\budget$.

 To find a minimum weight colorful solution under a fixed coloring $f \in \mathcal{F}$ we use dynamic programming.
 First fix a waypoint $w_0 \in \WP$ and let $\widehat{\WP} = \WP \setminus \{w_0\}$.
 Now for every $W' \subseteq \widehat{\WP}$, every $C \subseteq [k]$ and every $v \in V$ we create a table cell $T[W',C,v]$ with the following semantics.
 Cell $T[W',C,v]$ stores the minimum weight of a walk that starts in $w_0$, visits each waypoint in $\WP'$, uses exactly one edge of each color in $C$, and ends in~$v$.
 If there is no such walk, then it stores $\infty$.

 We fill the table in order of increasing~$|C|$.
 For $C = \emptyset$ we let $T[\emptyset, \emptyset, w_0]=0$ and $T[\WP',\emptyset, v]=\infty$ for any $\WP' \neq \emptyset$ or $v \neq w_0$.
 For $|C| \ge 1$ we let
 \[T[\WP', C, v] = \min\big\{T[\WP' \setminus\{v\}, C\setminus\{f(u,v)\}, u] + \wFn(u,v)\,\big|\, (u,v) \in E, f(u,v) \in C\big\}.\]
 Once the table was filled for all $C$, the minimum weight of a colorful solution under $f$ can be obtained as $\min\big\{T[\widehat{\WP},C,w_0] \,\big|\, C \subseteq [k]\big\}$.

 As for the running time, for a fixed coloring $f$ and a fixed vertex $v \in V$, the are $2^{|\widehat{\WP}|} \cdot 2^k$ table cells of the form $T[\WP',C, v]$ and each such cell can be computed in $O(\deg^{in}v)$ time
 Summing over all vertices $v \in V$ we obtain that the tables can be filled in $2^{|\widehat{\WP}|} \cdot 2^k \cdot m'$ time, where $m'=|E|$.
 As $|\widehat{\WP}| \le k-1$, this gives $O(4^k \cdot m')$ time to fill the cells.
 The answer can be found in the same running time.
 Summing over all $f \in \mathcal{F}$, this gives $e^k4^kk^{O(\log k)}\cdot m' \log m'$ time in total, including the time to construct~$f$.
 Since the capacity of each edge in the original graph was at most $k$, the multigraph has at most $k$ times more edges than the original graph, i.e., $m' \le k \cdot m$.
 Therefore we get the claimed $e^k4^kk^{O(\log k)}\cdot m \log m$ running time.

 Towards correctness, we first prove that the table is filled according to the desired semantics.
 We proceed by induction on $|C|$.
 If the walk is not allowed to use edge of any color, then it has to be trivial, and, thus, must end in $w_0$ and cannot visit any waypoint in $\widehat{\WP}$.
 As the trivial walk is of weight $0$, for $C =\emptyset$ the table is filled correctly.

 Let us now assume that $|C| \ge 1$ and the table is filled correctly for all smaller sets of colors.
 On the one hand, for every $(u,v) \in E$ with $f(u,v) \in C$, by induction hypothesis, if $T[\WP' \setminus\{v\}, C\setminus\{f(u,v)\}, u] \neq \infty$, then there is a walk $P$ of weight $T[\WP' \setminus\{v\}, C\setminus\{f(u,v)\}, u]$ that starts in $w_0$, visits each waypoint in $\WP' \setminus\{v\}$, uses exactly one edge of each color in $C \setminus\{f(u,v)\}$, and ends in~$u$.
 Attaching the edge $(u,v)$ to the end of walk $P$, we obtain a walk of weight $T[\WP' \setminus\{v\}, C\setminus\{f(u,v)\}, u] + \wFn(u,v)$ that starts in $w_0$, visits each waypoint in $\WP'$, uses exactly one edge of each color in $C$, and ends in~$v$.
 Therefore, the value stored in $T[W',C,v]$ is at least the minimum weight of such a walk.
 On the other hand, if $P$ is a walk of weight $h$ that starts in $w_0$, visits each waypoint in $\WP'$, uses exactly one edge of each color in $C$, and ends in~$v$, then let $(u,v)$ be the last edge on that walk (there is at least one edge since $|C| \ge 1$).
 Let $P'$ be obtained from $P$ by removing the last edge.
 Then $P'$ is a walk of weight $h -\wFn(u,v)$ that starts in $w_0$, visits each waypoint in $\WP' \setminus\{v\}$, uses exactly one edge of each color in $C \setminus\{f(u,v)\}$, and ends in~$u$.
 Thus $T[\WP' \setminus\{v\}, C\setminus\{f(u,v)\}, u] \le h -\wFn(u,v)$ and
 \[T[W',C,v] \le T[\WP' \setminus\{v\}, C\setminus\{f(u,v)\}, u] + \wFn(u,v) \le h -\wFn(u,v) + \wFn(u,v) = h.\]
 I.e., as $P$ was arbitrary, $T[W',C,v]$ is at most the weight of any such walk.
 Together with the previous one, it shows that the value $T[W',C,v]$ is computed correctly.

 Any walk that starts and ends in $w_0$, visits each waypoint in $\widehat{\WP}$ and uses exactly one edge of each color in $C$ for some $C \subseteq [k]$ gives a colorful solution to the problem, as it visits each waypoint in $\widehat{\WP} \cup \{w_0\}=\WP$.
 Thus any cell $T[\widehat{\WP},C,w_0]$ corresponds to a colorful solution.
 Conversely, any colorful solution must visit $w_0$.
 Therefore, it can be interpreted as a walk that starts and ends in $w_0$, visits each waypoint in $\widehat{\WP}$ and uses exactly one edge of each color in $C$ for some $C \subseteq [k]$.
 The minimum weight of such a walk is stored in $T[\widehat{\WP},C,w_0]$.
 Thus, our answer is correct.
\end{proof}

\begin{corollary}
 \DWRPshort is FPT with respect to budget $\budget$.
\end{corollary}

\begin{proof}
 As the weight of each edge is at least one, a solution of weight $\budget$ cannot use more than $\budget$ edges.
\end{proof}

\begin{corollary}
In $2^{O(\widehat{k}^2)}\cdot m \log m$ time we can decide whether the given instance of \DWRPshort has a solution using at most $\widehat{k}$ distinct edges (ignoring multiplicities).
\end{corollary}

\begin{proof}
 Consider a solution using at most $\widehat{k}$ distinct edges. 
 By \Cref{obs:n_traversals}, each of the edges is used at most $|W| \le \widehat{k}$ times.
  This implies that the solution uses at most $\widehat{k}^2$ edge occurrences in total and the result follows from \Cref{thm:solution_size}.
\end{proof}
}

\section{\DWRPshort is FPT with Respect to the Feedback Edge Number}\label{sec:fesn}
\toappendix{
\section{Ommited Material from \Cref{sec:fesn}}
}
In this section we prove \cref{thm:fes} that we restate here.

\thmfes*

\begin{proof}
Let $(G=(V,E), \WP, \wFn, \cFn, \budget)$ be an instance of \DWRPshort, where $n = |V|$.
Let $G'$ be the underlying undirected graph of $G$.
Let $F \subseteq E(G')$ be a feedback edge set of  $G'$ of size at most $k$, i.e.,  $G' \setminus F$ is a forest.
Note that it can be computed in polynomial time.

Our algorithm consists of three phases.
First, we exhaustively apply some straightforward reduction rules in order to simplify (or already reject) the input.
Then, we branch into $2^{O(k)}$ of possibilities, guessing the structure of the solution.
Finally, for each such guess, we create an equivalent instance of \DWRPshort with $O(k)$ vertices and edges.
The original instance is a yes-instance if and only if at least one of the created instances is a yes-instance.
Thus, our algorithm can be seen as a compression to OR of $2^{O(k)}$  instances, each of size $O(k)$.
By \cref{obs:n_traversals}, each such instance can be solved in time $k^{O(k)}$ by brute force.
Thus, the overall running time is $k^{O(k)} + n^{O(1)}$.

\proofsubparagraph{Preprocessing.}
The following reduction rule lets us deal with vertices of degree 1 in $G'$.

\begin{rrule}
Suppose there is a vertex $v$ in $G$ which has only one neighbor $u$ in $G'$.
\begin{itemize}
 \item If $v \notin W$, remove $v$ from the graph.
 \item If $v \in W$ and both $(u,v) \in E$ and $(v,u) \in E$, then remove the vertex $v$, include $u$ into $W$, and decrease the budget by $\wFn((u,v)) + \wFn((v,u))$.
 \item If $v \in W$ and one of the arcs $(u,v),(v,u)$ is not present, then reject the instance.
\end{itemize}
\end{rrule}
The correctness of the rule is straightforward.
Note that, as $\{u,v\}$ is not part of any cycle in $G'$, we may assume that $\{u,v\} \notin F$, and thus $F$ is still a feedback edge set of the graph resulting from the reduction rule.

We apply the reduction rule exhaustively, and, for simplicity, we keep denoting the resulting instance by $(G,\WP,\wFn,\cFn,\budget)$ and the underlying undirected graph by~$G'$.

Let $R$ be the set of all vertices that are incident to the edges in $F$.
Note that since the reduction rule was exhaustively applied, each vertex has %
degree at least two in $G'$.
Therefore, each leaf of $G''=G' \setminus F$ is in $R$ and, thus, $G''$ has at most $2k$ leaves.
Let $D$ be the set of vertices that have degree at least $3$ in $G''$.
As $G''$ is a forest with at most $2k$ leaves, the size of $D$ is at most $2k-1$.
We also define $X = R \cup D$, and we have $|X| \leq 4k-1$.

Let $\mathcal{P}'$ contain all paths in $G''$ with both endpoints in $X$ and all internal vertices outside~$X$.
In other words, the paths in $\mathcal{P}'$ correspond to the edges of the forest obtained from $G''$ by contracting all vertices of degree 2.
As there are at most $4k-1$ vertices in $X$ and, as $G''$ is a forest, the number of paths in $\mathcal{P}'$ is at most $4k-2$.
Finally, we define  $\mathcal{P} =  \mathcal{P}' \cup F$; we clearly have $|\mathcal{P}| \leq 5k-2$.

\proofsubparagraph{Types of paths and their costs.}
Consider an optimum solution $C^*$, i.e., a shortest (in terms of the sum of weights) closed walk that visits every vertex from $W$ in $G$. Among all such solutions, pick one with the minimum number of edges.

Consider a path $P \in \mathcal{P}$ with endpoints $u,v \in X$.
A \emph{visit} is an inclusion-wise maximal set of orientations of edges from $P$ that appear consecutively in $C^*$.
A visit is a \emph{pass} if it contains for every edge of $P$ exactly one of its orientations and it contains it exactly once, i.e., it corresponds to traversing $P$ once, either from $u$ to $v$ or from $v$ to $u$.

Notice that a path $P$ might be of one of the following types:
\begin{description}
\item[\boldmath ($u \curvearrowright v$)] there is at least one visit and every visit is a pass from $u$ to $v$,
\item[($u \curvearrowleft v$)]  there is at least one visit and  every visit is a pass from $v$ to $u$,
\item[($u \leftrightarrows v$)] every visit is a pass and there is at least one pass in each direction,
\item[($u \looparrowleft v$)] none of the visits is a pass, there is a visit that starts in $u$ and ends in $u$, but no visit that starts in $v$ and ends in $v$,
\item[($u \looparrowright v$)] none of the visits is a pass, there is a visit that starts in $v$ and ends in $v$, but no visit that starts in $u$ and ends in $u$,
\item[($u \circlearrowleft v$)] none of the visits is a pass, there is a visit that starts in $u$ and ends in $u$, and a visit that that starts in $v$ and ends in $v$,
\item[$(u \cdots v$)] the path $P$ is not visited at all by $C^*$.
\end{description}
\begin{claim}[\appmark]\label{clm:visits}
 The types listed above are mutually exclusive and cover all possibilities.
\end{claim}

\toappendix{ 
\begin{claimproof}[Proof of \Cref{clm:visits}]
First, suppose that there is a visit that starts at, say, $u$, ends in $v$, and traverses some edge more than once (i.e., is not a pass).
In such a case one can easily remove some edges from this visit, trimming it to a pass.
Note that the obtained walk visits the same vertices as $C^*$ and is shorter, contradicting the minimality of $C^*$.
Consequently, every non-pass visit starts and end in the same vertex, either $u$ or $v$.

Second, suppose that there is at least one visit that is a pass and at least one non-pass visit.
Again, we can remove all edges corresponding to the non-pass visit, obtaining a better solution -- it is still a closed walk, as the removed visit starts and end in the same vertex.
This justifies the claim.
\end{claimproof}
}

Let us make some further observations about possible visits in $P$ by $C^*$.

If no internal vertex of $P$ is in $W$, then $P$ is of type $u \curvearrowleft v$, $u \curvearrowright v$, $u \leftrightarrows v$, or $u \cdots v$. Indeed, any non-pass visit can be removed from $C^*$, contradicting the minimality.

Now consider a path $P$ of type $u \looparrowleft v$ (the case $u \looparrowright v$ is symmetric).
We may safely assume that $P$ is visited exactly once by $C^*$, as one of the visits (the one reaching furthest towards $v$) covers all vertices covered by all the visits altogether. Further, we may assume that this single visit starts in $u$, reaches the internal vertex of $P$ that belongs to $W$ and is closest to $v$, and then goes back to $u$.
Thus, if $P$ is of type $u \looparrowleft v$ or $u \looparrowright v$, then the visit in $P$ contributes to the total cost of the solution by a fixed amount that can be computed in advance.
We denote this cost by $\mathsf{cost}(P,u \looparrowleft v)$ or $\mathsf{cost}(P,v \looparrowright u)$, respectively. %

Now consider a path $P$ of type $u \circlearrowleft v$.
Similarly as before, we can assume that there are at least two vertices from $W$ among internal vertices of $P$, as otherwise we can shorten $C^*$ by removing one of the visits.
Furthermore, we can assume that $P$ is visited exactly twice: one visit starts at, $u$ reaches some vertex $x_u \in W$ and goes back to $u$, and the other visit starts at $v$, reaches some vertex $x_v \in W$, and goes back to $v$. In particular, no edge between $x_u$ and $x_v$ is traversed by $C^*$.
As before, in polynomial time we can find an optimal pair of visits that cover all vertices from $W$ that are on $P$, and compute its contribution to the cost of the solution (here note that we need to be careful as some edges might be only available in one direction). We denote this cost by $\mathsf{cost}(P,u \circlearrowleft v)$.

Finally, each pass in a fixed direction has a fixed cost, equal to the sum of weights of edges in $G$ corresponding to the edges of $P$, where we only choose edges in the correct direction.
We denote it by $\mathsf{cost}(P,u \curvearrowright v)$ or $\mathsf{cost}(P,v \curvearrowright u)$, depending on the direction.
We also define $\mathsf{capacity}(P,u \curvearrowright v)$ (and, symmetrically, $\mathsf{capacity}(P,u \curvearrowleft v)$) to the minimum capacity of an edge of $G$ corresponding to an edge of $P$, in the direction from $u$ to $v$ (from $v$ to $u$), resp.

Note that due to the directions and capacities of edges, some types might not be available for $P$.
We will define the cost associated with such an unavailable type to be $\infty$.

\proofsubparagraph{Building compressed instances.}

Now we perform the branching. For each path $P$ in  $\mathcal{P}$ with endvertices $u,v$, we guess its type $u\sim v$, where $\sim \in \{\curvearrowleft, \curvearrowright, \leftrightarrows, \looparrowleft, \circlearrowleft, \looparrowright, \cdots\}$. This results in $7^{|\mathcal{P}|} = 2^{O(k)}$ branches.

Consider one such branch. Now we build an instance $(G', W', \wFn', \cFn', b')$ of \DWRPshort.
Initialize $b'= b$.
We start with including all vertices from $X$ to $G'$, and all vertices from $X \cap W$ to $W'$.
Now let $P$ be a path in $\mathcal{P}$ with endvertices $u,v \in X$. We proceed, depending on $\sim$.
If the type of $P$ is $u \cdots v$, we do nothing.
If the type of $P$ is $u \looparrowleft v$ ($\looparrowright$), we add $u$ ($v$) to $W'$, and decrease $b'$ by $\mathsf{cost}(P,u \looparrowleft v)$ ($\mathsf{cost}(P,v \looparrowright u)$), respectively.
If the type of $P$ is $u \circlearrowleft v$, we add both $u$ and $v$ to $W'$, and decrease $b'$ by $\mathsf{cost}(P,u \circlearrowleft v)$.

Now consider the case that the type of $P$ is $u \curvearrowright v$ (the case $u \curvearrowleft v$ is symmetric).
We add a new vertex $x^P_{u \curvearrowright v}$ and edges $(u,x^P_{u \curvearrowright v})$ and $(x^P_{u \curvearrowright v},v)$. Vertices $u, v$, and $x^P_{u \curvearrowright v}$ are all included into $W'$.
We define $\wFn'((u,x^P_{u \curvearrowright v})) = \mathsf{cost}(P,u \curvearrowright v)$ and $\wFn'((x^P_{u \curvearrowright v},v))=0$. Furthermore, we set $\cFn'((u,x^P_{u \curvearrowright v})) = \cFn'((x^P_{u \curvearrowright v},v))  =  \mathsf{capacity}(P,u \curvearrowright v)$.
If the type of $P$ is $u \leftrightarrows$, we proceed as if was of both types $u \curvearrowright v$ and $u \curvearrowleft v$.

This concludes the construction of the instance. If $b' < 0$, we reject it immediately.
Clearly, the number of vertices of $G'$ is at most $|X| + 2|\mathcal{P}| \leq 14k$,
and the number of edges is at most $4|\mathcal{P}| \leq 20k$.
The discussion about correctness is postponed to appendix.
\toappendix{

\subparagraph*{Correctness of the Algorithm from \cref{thm:fes}.}
To prove correctness, we need to argue that $(G, \WP, \wFn, \cFn, \budget)$ is a yes-instance of \DWRPshort if and only if there exists a branch for which we created a yes-instance $(G', \WP', \wFn', \cFn', \budget')$ of \DWRPshort.

Suppose first that  $(G, \WP, \wFn, \cFn, \budget)$ is a yes-instance and let $C^*$ be an optimum solution of cost at most $\budget$. As argued before, we may assume that the type of each path $P \in \mathcal{P}$ with respect to $C^*$ is as listed above.
Consider the branch in which all types are guessed correctly, according to $C^*$ and the corresponding instance  $(G', \WP', \wFn', \cFn', \budget')$.
Define a closed walk $C'$ in $G'$ as follows.
Clearly, $C^*$ can be decomposed into visits in paths from $\mathcal{P}$.
Each visit that is not a pass is removed from $C^*$.
Each visit in a path $P$ that is a pass, say from $u$ to $v$, is replaced by two consecutive edges $(u,x^P_{u \curvearrowright v}), (x^P_{u \curvearrowright v}, v)$ in $G'$. Call the obtained closed walk (in $G'$) $C'$.
It is straightforward to verify that $C'$ visits all vertices from $W'$, does not exceed edge capacities, and its total cost is at most $b'$. Thus, $(G',\wFn', \cFn', b')$ is a yes-instance of \DWRPshort.

Now suppose that there is a branch in which the corresponding instance $(G',\wFn', \cFn', b')$ is a yes-instance.
Let $C'$ be an optimum solution for this instance.
We modify is as follows.
For each subpath of the form $u,x^P_{u \curvearrowright v},v$ for some $P \in \mathcal{P}$, we replace it by traversing the path $P$ from $u$ to $v$. Note that as $x^P_{u \curvearrowright v} \in W'$, all edges of types
$u \curvearrowright v$, $u \curvearrowleft v$, $u \leftrightarrows v$ are correctly passed by the constructed walk.

Now consider a path $P$ from $u$ to $v$, of type $u \looparrowleft v$ (the cases $u \looparrowright v$ and $u \circlearrowleft v$ are treated analogously). Recall that $u \in W'$, so it appears somewhere in $C'$.
Immediately after an arbitrary occurrence of $u$, we include a precomputed optimal visit in $P$ that starts and ends at $u$, and visits all vertices of $W$ that appear as internal vertices of $P$.
Note that this increases the cost of the solution by $\mathsf{cost}(P,u \looparrowleft v)$.
Let us denote the walk obtained by all these modifications by $C^*$.
Again, it is straightforward to verify that $C^*$ is a feasible solution to the instance $(G, \WP, \wFn, \cFn, \budget)$ of \DWRPshort.}
\end{proof}

\section{\DWRPshort is FPT Parameterized by Vertex Integrity}\label{sec:vi}

\toappendix{
\section{Omitted Proofs from \cref{sec:vi}}
}
In this section we prove \cref{thm:vi}.

\thmvi*

\begin{proof}
An undirected graph $G$ with vertex integrity $k$ has \emph{a $k$-modulator} $M \subseteq V(G)$ of size at most~$k$ such that each connected component of $G - M$ is of size at most $k$. 
We assume that $G$ is given along with $M$ as it can be found in $k^{O(k)}n$ time~\cite{DrangeDH16,fellows1989immersion}.
Thus we assume that we are given an instance $(G,\WP,\wFn,\cFn,\budget)$ of \DWRPshort, with a set $M$ of size at most $k$, such that each weakly connected component of $G-M$ has size at most $k$. For brevity, we will speak about `components' instead of `weakly connected components.'

In what follows, we start with introducing some notions and establishing some properties of an optimum solution.
Then, after some branching, we build an instance of ILP that corresponds to a solution with the required properties.

\proofsubparagraph{Segments.}
A \emph{segment} is a walk in $G$ with at least two vertices, that starts and ends in a vertex of $M$ and all internal vertices are not in $M$.
We remark that there might be no internal vertices -- such a segment is just an arc contained in $M$. We call such a segment \emph{trivial}.
However, if a segment is non-trivial (i.e., not trivial), all its internal vertices are contained in a single component of $G-M$;
the segment is \emph{associated} with the component.

Note that the endpoints of a non-trivial segment need not to be distinct.
A segment~$s$ from $u$ to $v$ is \emph{minimal} if there is no segment from $u$ to $v$ that uses only the arcs of $s$,
contains all vertices of $s$ that are in $W$, and is shorter than $s$.

\proofsubparagraph{Segments in an Optimum Solution.}
Let $C^*$ be an optimum solution, i.e., a shortest (w.r.t. the sum of weights of edges) closed directed walk in $G$ that visits every vertex of $W$.
We treat $C^*$ as a sequence of vertices.
Clearly $C^*$ can be decomposed into segments (where two consecutive segments overlap on corresponding endpoints);
let $\mathcal{S}^*$ be the collection of these segments. We emphasize that $\mathcal{S}^*$ is a multiset,
as some segments might be used more than once.
We observe that we can assume that each segment in $\mathcal{S}^*$ is minimal,
for otherwise we can replace it with a shorter one with the same endpoints, still covering all of $W$.

Some segments in $\mathcal{S}^*$ might not be needed to cover $W$, but are required to keep the connectivity of the solution.
We call them \emph{connectors}.
We can identify connectors with the following procedure. We initialize $\mathcal{S}^\circ = \emptyset$ and proceed iteratively as follows.
If there is a segment $s$ in $\mathcal{S}^* \setminus \mathcal{S}^\circ$ such that the segments $\mathcal{S}^* \setminus \mathcal{S}^\circ \setminus \{s\}$ cover all vertices of $W$, we include $s$ into $\mathcal{S}^\circ$.
We emphasize that here we work with multisets, so the `$\setminus$' operation respects the multiplicity of elements.
The set $\mathcal{S}^\circ$ contains connectors.
Note that $\mathcal{S}^\circ$ is not defined uniquely and it might depend on the ordering in which segments were selected.

Consider a connector, i.e., a segment $s$ from $u$ to $v$ that is in $\mathcal{S}^\circ$.
We observe that we can safely assume that it is a $u$-$v$-path, as otherwise we can replace it with such a path obtaining a solution of no larger length.
We call such segments \emph{simple}.

\proofsubparagraph{Traversals.}
Now let us consider a component $C$ of $G - M$.
A \emph{traversal} of $C$ is a collection $\mathcal{S}_C$ of segments associated with $C$ that:
\begin{itemize}
	\item together cover every vertex from $W \cap V(C)$,
	\item for every $s \in \mathcal{S}_C$ there is a vertex in $W \cap V(C)$ not covered by segments in $\mathcal{S}_C \setminus \{s\}$,
	\item every segment in $\mathcal{S}_C$ is minimal,
	\item every edge $e$ is used at most $\cFn(e)$ times by $\mathcal{S}_C$.
\end{itemize}
Note that the segments from $\mathcal{S}^* \setminus \mathcal{S}^\circ$ associated with $C$ form a traversal.
Indeed, the first condition follows since $C^*$ covers all vertices from $W$ and we moved segments to $\mathcal{S}^\circ$ if they were not required to cover $W \cap V(C)$.
The second condition again follows from the definition of $\mathcal{S}^\circ$.
Finally, recall that we already established the third condition without loss of generality.

\begin{claim}[\appmark]\label{clm:traversals}
	Every traversal has at most $k(k+2)$ vertices.
	Moreover, for any component~$C$ of $G-M$,
	there are at most $k^{k(k+2)}$ traversals of $C$.
	Furthermore, each traversal uses each edge at most $k$ times.
\end{claim}
\toappendix{
\begin{claimproof}[Proof of \cref{clm:traversals}.]
	Note that each traversal can be decomposed into paths with both endvertices in $M$ and internal vertices in $C$
	and cycles with at most one vertex in $M$ and the remaining ones in~$C$.

	Note that each such path has at most $k+2$ vertices	and each such cycle has at most $k+1$ vertices.
	Moreover, by the second and third condition in the definition of a traversal,
	the total number of paths and cycles does not exceed $|W \cap V(C)| \leq |V(C)| \leq k$.
	Consequently, the total number of vertices of a traversal is at most $k(k+2)$ and each edge is used at most $k$ times.

	The bound on the number of traversals follows from the fact that $C$ and $M$ both have at most $k$ vertices.
\end{claimproof}
}

\proofsubparagraph{Skeleton of the Solution.}
Let us define two directed graphs $H_0$ and $H_1$ as follows.
Both have the same vertex set, $M$.

For distinct $u,v \in M$, the arc $uv$ exists in $H_0$ if such an arc appears in $C^*$ (i.e., $C^*$ contains $u$ immediately followed by $v$).
On the other hand, the arc $uv$ exists in $H_1$ if there is a non-trivial segment in $\mathcal{S}^*$ starting in $u$ and ending in $v$.
We emphasize that neither $H_0$ nor $H_1$ has parallel arcs nor loops.

Let $H$ be the directed graph with vertex set $M$ and the arc set being the union of arc sets of $H_0$ and $H_1$ (again, with no parallel arcs).
Note that $H$ has one strongly connected component $H'$ containing all vertices from $W \cap M$, and vertices that are not in $H'$ are isolated in $H$. These vertices do not appear on $C^*$ at all.

For $u,v \in M$, we say that a segment from $u$ to $v$  is \emph{bad} (for a particular choice of $H_0, H_1$) if $u$ or $v$ are not in $H'$,
or the arc $uv$ does not appear in $H$. A traversal is \emph{bad} if it contains a bad segment.

We say that a collection $\mathcal{S}$ of segments is compatible with $H_0$ and $H_1$ if, for every arc $uv$ of $H_0$,
$\mathcal{S}$ contains at least one copy of the trivial segment $uv$, and for every arc $uv$ of $H_1$, $\mathcal{S}$ contains at least one non-trivial segment starting in $u$ and ending in $v$.

\proofsubparagraph{The Algorithm.} We proceed to the description of the algorithm.
We exhaustively guess the graphs $H_0$ and $H_1$ described above, considering all possible graphs on the vertex set $M$, so that the union of these two graphs has one strongly connected component and the remaining vertices isolated.
This gives us $2^{O(k^2)}$ possible guesses.
We are going to look for a set of segments that are compatible with $H_0$ and $H_1$, and the structure of $H_0,H_1$ will be used to ensure that the solution is connected.

Suppose we are in a branch corresponding to one guess.
We will build an instance of generalized $N$-fold ILP~\cite{EisenbrandHK18,KouteckyLO18}.
In our case, we have a brick corresponding to each component~$C$ of $G-M$ and some further bricks that do not have any local constraints.

For each traversal of $C$ that is not bad, we introduce one binary variable; recall that there are at most $k^{k(k+2)}$ such variables.
The value of this variable indicates whether the vertices from $W \cap V(C)$ are covered by a particular traversal.
Furthermore, we introduce a local constraint to ensure that exactly one traversal is chosen.

For each simple segment associated with $C$ which is a $u$-$v$-path for some $u,v \in M$, we introduce a variable with possible values from $0$ to $n$.
The intended meaning of this variable is the number of times this segment was used as a connector from $u$ to $v$ by the solution.
There are at most $k^{k+2}$ such variables.

For each edge $e$ with at least one endpoint in $C$ we introduce a (local) constraint ensuring that this edge is used by the solution at most as many times as the capacity permits.
Namely, we sum over all simple segments and over all traversals that use edge $e$ (multiplying by how many times they use $e$) and require that this sum is at most $\cFn(e)$.

Now, let us move to variables not related to components.
For each ordered pair $u,v$ of distinct vertices, we introduce a variable with possible values from $0$ to $\cFn((u,v))$.
The intended meaning of this variable is the number of times the trivial segment is used as a connector from $u$ to $v$ by the solution.
This gives us at most $k^2$ variables.
We can put each of this variable to a brick for itself.
As the edges within $M$ are not used by any traversals or non-trivial segments, the box constraints ensure that the capacities are adhered for these arcs, i.e., there will be no local constraints for these bricks.

Next, we add global constraints that ensure that the selected set of segments is compatible with $H_0$ and $H_1$.
More precisely, for each arc $uv$ of $H_0$ ($H_1$), we add a constraint that ensures that we choose the trivial segment $uv$ (at least one non-trivial segment that starts in $u$ and ends in $v$), respectively. We sum it over all selected traversals and connectors. These constraints are responsible for making the solution connected.

Finally, we need to make sure that for each $v \in M$, the number of segments that start in~$v$ is equal to the number of segments that end in $v$. Again, we can force it by summing the number of particular segments, including trivial, connectors, and ones belonging to traversals.
These are also global constraints.

We minimize the total length (i.e., the sum of weights of edges) in all selected segments.
Note that for some guesses we might get an infeasible instance of ILP -- this means that, e.g., the segments that are required by $H_0$ and $H_1$ do not exist in $G$ or all traversals of some component are bad.
If all branches give infeasible instances, we return that there is no solution.
However, every solution found (in any branch) is feasible and we can return the shortest one.
From the description above it follows that if an optimum solution $C^*$ exist, we will find it (or, more precisely, a solution of the same total length) in the branch corresponding to the actual $H_0$ and $H_1$.

\proofsubparagraph{Solving $N$-fold ILP.}
Recall that generalized $N$-fold ILP can be solved (see \cite[Cor. 97]{EisenbrandHKKLO19} and also \cite{EisenbrandHK18,KouteckyLO18}) in 
$(rs\Delta)^{O(r^2s +rs^2)}d \log d \log D\log val_{\max}$ time, where
 \begin{itemize}
   \item $r$ is the number of global constraints (at most $k(k-1)+k=k^2$ in our case),
   \item $s$ is the number of local constraints for each brick ($O(k^2)$ in our case),
   \item $\Delta$ is the largest coefficient in the constraints ($k$ in our case),
   \item $d$ is the total number of variables (at most $k^2 + n \cdot k^{k(k+2)}+n \cdot k^{k+2}$ in our case)%
   \item $D$ is the maximum range of a variable (at most $n$ in our case),
   \item and $val_{\max}$ is the maximum value of the objective function (by \cref{obs:n_traversals}, in our case it is at most $n^2 \cdot U$, where $U$ is the maximum weight of an arc in $G$).
 \end{itemize}
Including the branching to guess $H_0$ and $H_1$, we obtain the final running time bound of
$k^{O(k^6)} \cdot n \cdot \log^2 n (\log n + \log U)$.
This completes the proof.
\end{proof}

\toappendix{
\section{\DWRPshort is XP with Respect to Treewidth}

In this section we present our algorithm parameterized by treewidth. As the algorithm uses standard techniques, it is rather similar to the algorithm of Mannens et al.~\cite{MannensNSS21}.
The following theorem is a more detailed version of \cref{thm:tw}.

\begin{theorem}
\label{thm:FPT_wrt_tw_no_visits}
 There is an algorithm that given an instance $(G =(V,E), \WP, \wFn, \cFn, \budget)$ of \DWRP, a nice tree decomposition $(T, \beta)$ of $G$ of width~$t$ and an integer $\nu$, computes in $O(((t+1)(\nu+1)^2)^{2(t+1)} \cdot t^2 \cdot n)$ (FPT wrt $t+\nu$) time the length of a shortest closed walk~$C$ that traverses each vertex in $\WP$ at least once and each vertex in $V$ at most $\nu$ times and such that for each edge $e$ the number of times $C$ traverses $e$ is at most $\cFn(e)$ if there is any walk satisfying these conditions.
\end{theorem}

\begin{corollary}\label{cor:XP_tw}
 \DWRP is in \XP with respect to treewidth and \FPT with respect to treewidth combined with the number of waypoints.
\end{corollary}

\begin{proof}[Proof of \Cref{cor:XP_tw}]
 Follows directly from the theorem and \cref{obs:n_traversals}.
\end{proof}

\begin{proof}[Proof of \Cref{thm:FPT_wrt_tw_no_visits}]
 We use a dynamic programming over the given decomposition $(T,\beta)$. Let $r \in V(T)$ be the root of the decomposition. Let $w_0 \in \WP$ be an arbitrary waypoint. Remove the node where $w_0$ is forgotten from the decomposition and add $w_0$ to all the ancestors of that bag. With a suitable choice of the root while turning the tree decomposition into a nice one, this could be done without increasing the width of the decomposition.

 We create a table $A$ indexed by nodes $x \in V(T)$, partitions $\mathcal{P}$ of $\beta(x)$, and pairs of functions $in, out \colon \beta(x) \to \{0, \ldots, \nu\}$.

 We call a sequence $\mathcal{C} = C_1, \ldots, C_q$ of walks \emph{a partial solution} at $x$ if
 \begin{enumerate}[i)]
 \item each $C_i$ is a non-trivial walk in $G_x$ that starts and ends in a vertex of $\beta(x)$,
 \item each vertex in $\WP \cap (\gamma(x) \setminus \beta(x))$ is traversed by at least one of the walks,
 \item no vertex is traversed more than $\nu$ times (in particular, no vertex has in-degree or out-degree more than $\nu$ with respect to the walks), summing over all the walks,
 \item and each edge $e$ is traversed at most $\cFn(e)$ times, summing over all the walks.
 \end{enumerate}
 The \emph{length} of a partial solution is the sum of the length of its walks.

 A partial solution $\mathcal{C}$ is \emph{compatible with} $(\mathcal{P}, in, out)$ at $x$, if for each $v \in \beta(x)$ vertex $v$ is entered $in(v)$ times and exited $out(v)$ times summing over all walks in $\mathcal{C}$ and moreover the following holds.
 Let $H$ be the subgraph of $G_x$ formed by the edges used by the walks in $\mathcal{C}$ and $H'$ be the underlying undirected graph of $H$.
 Then for each non-trivial connected component $K$ of $H'$ we have that $K \cap \beta(x) \neq \emptyset$ and $K \cap \beta(x) \in \mathcal{P}$ and for each $v \in \beta(x)$ isolated in $H'$ we have $\{v\} \in \mathcal{P}$.

 For each $x \in V(T)$, partition $\mathcal{P}$ of $\beta(x)$, and a pair of functions $in, out \colon \beta(x) \to \{0, \ldots, \nu\}$ we simply store in $A[x, \mathcal{P}, in, out]$ the length of a shortest partial solution compatible with $(\mathcal{P}, in, out)$ at $x$, if there is any, or $+\infty$ otherwise.

We compute the values of $A[x, \mathcal{P}, in, out]$ in a bottom-up manner distinguishing the type of the node $x$.

\proofsubparagraph{Leaf Node} If $x$ is a leaf node, then, as $\beta(x)=\emptyset$, the only possible partition is $\emptyset$ and the only possible function $\beta(x) \to  \{0, \ldots, \nu\}$ is also $\emptyset$. We let $A[x, \emptyset, \emptyset,\emptyset]=0$.

\proofsubparagraph{Introduce Vertex Node} If $x$ is a node introducing vertex $v$ and $y$ is the only child of $x$, then $v$ is isolated in $G_x$. Hence, there is no way to enter it or leave it. Thus, if $in(v) \neq 0$ or $out(v)\neq 0$, or $\{v\} \notin \mathcal{P}$, then we let $A[x,\mathcal{P}, in, out] = + \infty$. Otherwise, we let $\mathcal{P}' = \mathcal{P} \setminus \{\{v\}\}$, and $in'$ and $out'$ be $in$ and $out$ restricted to $\beta(y)=\beta(x) \setminus \{v\}$, respectively. We let $A[x, \mathcal{P}, in, out]=A[y, \mathcal{P}', in', out']$.

\proofsubparagraph{Introduce Edge Node} If $x$ is a node introducing an edge $\{u,v\}$ and $y$ is the only child of $x$, then we have to guess how many times was the edge $(u,v)$ and how many times the edge $(v,u)$ is used by the solution.
Let $P_u \in \mathcal{P}$ be such that $u \in P_u$.
If $v \notin P_u$, then we let $A[x, \mathcal{P}, in, out]=A[y, \mathcal{P}, in, out]$.

Otherwise, let $a_{uv} =0$ if the edge $(u,v)$ is not in $G$ and let $a_{uv} = \min\{\cFn((u,v)), out(u), in (v)\}$ otherwise. Similarly, let $a_{vu}=0$ if the edge $(v,u)$ is not part of $G$ and let $a_{vu} = \min\{\cFn((v,u)), out(v), in (u)\}$ otherwise.
Let
\begin{multline*}
A[x, \mathcal{P}, in, out] = \min \{\min_{\substack{0 \le b_{uv} \le a_{uv}\\ 0 \le b_{vu} \le a_{vu}}} (A[y, \mathcal{P}, in', out'] + b_{uv} \cdot \wFn((u,v)) + b_{vu} \cdot \wFn((v,u))),\\
 \min_{\substack{0 \le b_{uv} \le a_{uv}\\ 0 \le b_{vu} \le a_{vu}\\ b_{uv}+b_{vu} \ge 1\\  \mathcal{P}'}} (A[y, \mathcal{P}', in', out'] + b_{uv} \cdot \wFn((u,v)) + b_{vu} \cdot \wFn((v,u))) \},
\end{multline*}

 where \[in'(h) = \begin{cases}
                 in(h) & h \notin \{u,v\}\\
                 in(v) - b_{uv} & h=v \\
                 in(u) - b_{vu} & h=u \\
                 \end{cases},\]
                 \[out'(h) = \begin{cases}
                 out(h) & h \notin \{u,v\}\\
                 out(u) - b_{uv} & h=u \\
                 out(v) - b_{vu} & h=v \\
                 \end{cases},\] and $\mathcal{P}'$ contains sets $P'_u$ and $P'_v$ such that $u \in P'_u$, $v \in P'_v$ and $\mathcal{P} = (\mathcal{P'} \setminus \{P'_u,P'_v\}) \cup \{P_u\}$ (possibly $P'_u=P'_v=P_u$).
\proofsubparagraph{Forget Node} If $x$ is a node forgetting vertex $v$ and $y$ is the only child of $x$, then we distinguish cases based on whether $v$ is a waypoint or not.
In either case, let $\mathcal{P}=\{P_1, \ldots, P_s\}$.
To assure that the final solution is connected, if $v$ is traversed, then $v$ must share a connected component with a vertex in $\beta(x)$.

If $v$ is a waypoint, then we need to make sure that it is traversed by the solution.
We let
\[A[x, \mathcal{P}, in, out] = \min_{\substack{1 \le i \le \nu\\ 1 \le j \le s}} A[y, \mathcal{P}_j, in_i, out_i],\]
where $\mathcal{P}_j = (\mathcal{P} \setminus \{P_j\}) \cup \{P_j\cup\{v\}\}$ (i.e., we add $v$ to the set $P_j$ of $\mathcal{P}$),
\[in_i(u)=\begin{cases}
           in(u) & u \in \beta(x)\\
           i & u=v
          \end{cases},
\]
and
\[
 out_i(u)=\begin{cases}
           out(u) & u \in \beta(x)\\
           i & u=v
          \end{cases}.
\]

If $v$ is not a waypoint, then it might not be part of the solution.
We let
\[A[x, \mathcal{P}, in, out] = \min\{A[y, \mathcal{P}\cup \{\{v\}\}, in_0, out_0], \min_{\substack{1 \le i \le \nu\\ 1 \le j \le s}} A[y, \mathcal{P}_j, in_i, out_i]\},
\]
where the meaning of $\mathcal{P}_j$, $in_i$, and $out_i$ is the same as in the previous case.

\proofsubparagraph{Join Node} Finally, let $x$ be a join node with two children $y$ and $z$.
We simply let \[A[x, \mathcal{P}, in, out] = \min_{\substack{\mathcal{P}_y, \mathcal{P}_z\\ in_y, in_z\\out_y,out_z}} A[y, \mathcal{P}_y, in_y, out_y] + A[z, \mathcal{P}_z, in_z, out_z],\]
where the minimum is taken over all pairs of $(\mathcal{P}_y, in_y, out_y)$ and $(\mathcal{P}_z, in_z, out_z)$ such that $\mathcal{P}$ is the finest common coarsening of $\mathcal{P}_y$ and $\mathcal{P}_z$ and for every $v \in \beta(x)=\beta(y)=\beta(z)$ we have $in(v)=in_y(v)+in_z(v)$ and $out(v)=out_y(v)+out_z(v)$.
Note that we can simply iterate over the tables of the children, gradually updating the table of $x$.

Once the computation is finished, the desired length of a closed walk satisfying the conditions is obtained as
$\min_{1 \le i\le \nu} A[r,\{\{w_0\}\}, w_0 \mapsto i, w_0 \mapsto i]$.

\proofsubparagraph{Time Complexity}
Since each bag is of size at most $t+1$, there are at most $(t+1)^{(t+1)}$ partitions of it.\footnote{We could reduce the number of partitions considered to $2^{O(t)}$ by the now standard technique of representative sets~\cite{BodlaenderCKN15}, but we refrain from that here to simplify the presentation.}
Also there are at most $(\nu+1)^{(t+1)}$ functions $\beta(x) \to \{0, \ldots, \nu\}$.
Hence, for each $x$ the table $A[x, \bullet, \bullet, \bullet]$ has at most $((t+1)(\nu+1)^2)^{(t+1)}$ entries.

The computation in leaf nodes takes a constant time.
Similarly, the computation of each entry in an introduce vertex node takes a constant time.
In introduce edge node, for each entry we take a minimum over at most $(\nu+1)^2+(\nu+1)^2\cdot 2^{(t-1)}$ entries of the child.
In forget node, for each entry we take a minimum over at most $1+(\nu)t$ entries of the child.
Finally, in a join node, we make a $O(t)$ time computation for each pair of entries of the children.
Therefore, the computation in each node can be done in $O(((t+1)(\nu+1)^2)^{2(t+1)} \cdot t)$ time.
As we can assume that there are $O(tn)$ nodes in the given decomposition~\cite[Lemma 7.4, see also p. 168]{CyganFKLMPPS15}, the algorithm runs in $O(((t+1)(\nu+1)^2)^{2(t+1)} \cdot t^2 \cdot n)$ time.

\proofsubparagraph{Correctness}
We first prove that the number stored in the table satisfies the desired semantics, that is, $A[x, \mathcal{P}, in, out]$ is the length of a shortest partial solution compatible with $(\mathcal{P}, in, out)$ at $x$, if there is any, or $+\infty$ otherwise.

We start by proving that if $A[x, \mathcal{P}, in, out]$ is finite, then there is a compatible partial solution of length $A[x, \mathcal{P}, in, out]$.
We proceed by induction on the height of the subtree of the tree decomposition rooted at $x$.
For leaf nodes, the empty sequence (of walks) forms such a solution of length $0$, proving the claim.

Next we assume that the claim already holds for the children of node $x$ by induction hypothesis.
For introduce vertex node $x$ with child $y$, it is easy to show that the partial solution compatible with $(\mathcal{P}', in', out')$ at $y$ is also compatible with $(\mathcal{P}, in, out)$ at~$x$.

For introduce edge node $x$ with child $y$ and the edge $\{u,v\}$ being introduced, take the solution $\mathcal{C'}$ compatible with $(\mathcal{P}', in', out')$ at $y$, which exists by the induction hypothesis, and add to it $b_{uv}$ walks from $u$ to $v$ and $b_{vu}$ walks from $v$ to $u$, each consisting of a single arc, to produce a sequence $\mathcal{C}$.
It is straightforward to check that $\mathcal{C}$ is a partial solution compatible with $(\mathcal{P}, in, out)$ at~$x$.

For forget node, it is easy, if $A[x, \mathcal{P}, in, out] = A[y, \mathcal{P}\cup \{\{v\}\}, in_0, out_0]$---simply take the same partial solution.
If $A[x, \mathcal{P}, in, out] = A[y, \mathcal{P}_j, in_i, out_i]$ for some $i \in \{1, \ldots, \nu\}$ and $j \in \{1, \ldots, s\}$, then let $C_1, \ldots, C_q$ be the partial solution for $A[y, \mathcal{P}_j, in_i, out_i]$, which exists by the induction hypothesis.
Furthermore, assume that walks $C_1, \ldots, C_a$ are closed walks that start and end in~$v$, walks $C_{a+1}, \ldots C_b$ start in~$v$, but end in a different vertex of $\beta(x)$ and walks $C_{b+1}, \ldots, C_c$ start in a different vertex of $\beta(x)$ and end in~$v$ (possibly, $a=0$, $b=a$, or $c=b$) and no other walks start or end in~$v$.
Note that $b-a = c-b$.
For $d\in \{1, \ldots, b-a\}$ we connect $C_{b+d}$ with $C_{a+d}$ to make a walk that neither starts nor ends in~$v$.
If $a=0$, then we are done, as no walk starts or ends in $v$ and we can use the solution for $x$.
Otherwise, we connect all walks $C_1, \ldots, C_a$ into a single closed walk $C'$.
Note that $C'$ is now the only walk that starts or ends in~$v$.
As $P_j \neq \emptyset$ and $(P_j\cup\{v\}) \in \mathcal{P}_j$, either $C'$ contains a vertex $u \in \beta(x)$ different from $v$, or it shares a vertex $u$ with another walk $C''$ that starts and ends in a vertex of $\beta(x)$ different from~$v$.
As $C'$ is a closed walk, we can ``rotate'' it, to start and end in $u$.
In the former case we are done, the walk now ends in starts and ends in a vertex of $\beta(x)$ different from~$v$.
In the latter case, we break $C''$ at $u$ and insert $C'$ into it.
Again, after that, no walk starts or ends in $v$ and we can use this solution for $x$.
As we did not change how many times edges are traversed, the other conditions are easy to verify.

Finally, for the join node, it is enough to take the union of the partial solutions for the two children.
Since each edge is introduced exactly once, no edge can appear both in $G_y$ and $G_z$ and the capacities are adhered.
Also each vertex of $\gamma(x) \setminus \beta(x)$ is traversed by at most one of the solutions and hence at most $\nu$ times and each vertex of $\WP \cap (\gamma(x) \setminus \beta(x))$ is traversed by one of the solutions.

For the converse direction, suppose that there is partial solution $\mathcal{C}=\{C_1, \ldots, C_q\}$ compatible with $(\mathcal{P}, in, out)$ at $x$ of length $h$.
We want to show that $A[x, \mathcal{P}, in, out] \le h$.
We again proceed by induction on the height of the subtree of the tree decomposition rooted at $x$.
For leaves this is clear, as the only solution is $\emptyset$ with length $0$.

Next we assume that the claim already holds for the children of the node $x$ by induction hypothesis.
For introduce vertex node, $v$ is isolated in $G_x$ and hence must be also isolated in the solution.
As the solution is compatible with $(\mathcal{P}, in, out)$ at $x$, this implies that $in(v) = 0$, $out(v)= 0$, and $\{v\} \in \mathcal{P}$.
Thus the algorithm set $A[x, \mathcal{P}, in, out]=A[y, \mathcal{P}', in', out']$.
But we can use the same partial solution for the child $y$ and the claim follows from the claim for the child.

Let now $x$ be an introduce edge node, $y$ its child, and $uv$ the edge being introduced.
Let $b^*_{uv}$ be the sum of the number of times the edge $(u,v)$ is traversed over all walks in $\mathcal{C}$ and similarly with $b^*_{vu}$ and the edge $(v,u)$.
Obviously, if $(u,v)$ is not in $G$, then we have $b^*_{uv}=0$ and otherwise we have $b^*_{uv}\le\cFn{(u,v)}$, and by compatibility also $b^*_{uv}\le out(u)$ and $b^*_{uv}\le in(v)$.
In either case $b^*_{uv}\le a_{uv}$ a similarly $b^*_{vu} \le a _{vu}$.

If $b^*_{uv}=b^*_{vu}=0$, then $\mathcal{C}$ is also compatible with $(\mathcal{P}, in, out)$ at $y$ and, by induction hypothesis, we have $A[y, \mathcal{P}, in, out] \le h$.
If $v \notin P_u$, then $A[x, \mathcal{P}, in, out] = A[y, \mathcal{P}, in, out] \le h$.
Otherwise we have $A[x, \mathcal{P}, in, out] \le \min_{\substack{0 \le b_{uv} \le a_{uv}\\ 0 \le b_{vu} \le a_{vu}}} (A[y, \mathcal{P}, in', out'] + b_{uv} \cdot \wFn((u,v)) + b_{vu} \cdot \wFn((v,u))) \le A[y, \mathcal{P}, in, out] \le h$ (note that $in'=in$ and $out'=out$ if $b_{uv}=b_{vu}=0$).

If $b^*_{uv}+b^*_{vu} \ge 1$, then we remove all occurrences of $(u,v)$ and $(v,u)$ from the walks in $\mathcal{C}$, splitting them into more walks or shortening them.
Call the resulting set of walks $\widehat{\mathcal{C}}$.
As both $u$ and $v$ are in $\beta(x)$, each new walk starts and ends in $\beta(x)=\beta(y)$ and each vertex in $\WP \cap (\gamma(x) \setminus \beta(x))$ is traversed exactly as many times by $\widehat{\mathcal{C}}$ as by $\mathcal{C}$, i.e., at least once.
Furthermore, each vertex and each edge is traversed at most as many times by $\widehat{\mathcal{C}}$ as by $\mathcal{C}$.
Hence $\widehat{\mathcal{C}}$ is a partial solution of length $h - b^*_{uv} \cdot \wFn((u,v)) - b^*_{vu} \cdot \wFn((v,u))$.

Let $H$ be the graph constructed for $\mathcal{C}$ and $\widehat{H}$ be obtained from $H$ by removing the edge $\{u,v\}$.
If $u$ and $v$ are in the same connected component of $\widehat{H}$, then $\widehat{\mathcal{C}}$ is compatible with $(\mathcal{P}, in', out')$, where $in'$ and $out'$ are obtained by letting $b_{uv}=b^*_{uv}$ and $b_{vu}=b^*_{vu}$.
Therefore $A[y, \mathcal{P}, in', out'] \le h - b^*_{uv} \cdot \wFn((u,v)) - b^*_{vu} \cdot \wFn((v,u))$ and  $A[x, \mathcal{P}, in, out] \le h$.

If $u$ and $v$ are in different connected components of $\widehat{H}$, then let $K_u$ be the component containing $u$ and $K_v$ be the component containing $v$.
Note that, if $K$ is the component containing $u$ and $v$ in $H$, then $V(K)=V(K_u) \cup V(K_v)$, since $H$ and $\widehat{H}$ only differ by the edge $\{u,v\}$.
Let $P'_u = \beta(x) \cap K_u$, $P'_v = \beta(x) \cap K_v$, and $\mathcal{P}' = (\mathcal{P} \setminus \{P_u\}) \cup \{P'_u,P'_v\}$.
Then $\widehat{\mathcal{C}}$ is compatible with $(\mathcal{P}', in', out')$, where $in'$ and $out'$ are again obtained by letting $b_{uv}=b^*_{uv}$ and $b_{vu}=b^*_{vu}$.
Therefore $A[y, \mathcal{P}', in', out'] \le h - b^*_{uv} \cdot \wFn((u,v)) - b^*_{vu} \cdot \wFn((v,u))$ and  $A[x, \mathcal{P}, in, out] \le h$.

Now suppose that $x$ is a forget node, $y$ its child and $v$ the vertex being forgotten.
Let $\mathcal{P} =\{P_1, \ldots, P_s\}$.
Let $i$ be the number of times that $v$ is traversed by $\mathcal{C}$.
If $v$ is waypoint, then $i \ge 1$ by the conditions on a partial solution.
If $v$ is traversed, then it is not isolated in $H$ and therefore its connected component has a non-empty intersection with $\beta(x)$, which is in $\mathcal{P}$, say $P_j$.
Then the intersection of this component with $\beta(y)$ is $P_j \cup \{v\}$ and $\mathcal{C}$ is compatible with $(\mathcal{P}_j, in_i, out_i)$ at $y$.
Thus $A[y, \mathcal{P}_j, in_i, out_i] \le h$ by induction hypothesis and also $A[x, \mathcal{P}, in, out] \le h$.

If $v$ is not traversed (which can only occur if it is not a waypoint), then it is isolated in $H$ and $\mathcal{C}$ is compatible with $(\mathcal{P} \cup \{\{v\}\}, in_0, out_0)$ at $y$.
Then $A[y, \mathcal{P} \cup\{\{v\}\}, in_0, out_0] \le h$ by induction hypothesis and also $A[x, \mathcal{P}, in, out] \le h$.

Finally, let $x$ be a join node and $y$ and $z$ its children.
Let $\mathcal{C}_y$ and $\mathcal{C}_z$ be initially empty.
For each walk $C \in \mathcal{C}$, take each maximal continuous part of walk $C$ using only arcs of $G_y$ and put it into $\mathcal{C}_y$ and put each maximal continuous part only using arcs of $G_z$ and put it into $\mathcal{C}_z$.
Since every arc is introduced exactly once, graphs $G_y$ and $G_z$ do not share any arcs and therefore, the above maximal parts are not overlapping, apart from their endpoints and the multiset of arcs appearing in $\mathcal{C}$ is a disjoint union of multisets of arcs of $\mathcal{C}_y$ and $\mathcal{C}_z$.
Thus the sum of their lengths $h_y$ and $h_z$ equals the length $h$ of $\mathcal{C}$.
Furthermore, only vertices of $\beta(x)$ might be incident with arcs of both $G_y$ and $G_z$.
Therefore, all the walks in $\mathcal{C}_y$ and $\mathcal{C}_z$ start and end in $\beta(x)$.

Each vertex $u$ in $\gamma(x) \setminus \beta(x)$ is either in $\gamma(y) \setminus \beta(x)$ or $\gamma(z) \setminus \beta(x)$.
Thus each such waypoint is traversed by at least one walk either of $\mathcal{C}_y$ or of $\mathcal{C}_z$.
No vertex and no edge is traversed more times by $\mathcal{C}_y$ or by $\mathcal{C}_z$.
Hence, $\mathcal{C}_y$ and $\mathcal{C}_z$ form a partial solution for $y$ and $z$, respectively.

Let $H_y$ be the subgraph of $G_y$ formed by the edges used by the walks in $\mathcal{C}_y$ and $H'_y$ be the underlying undirected graph of $H_y$.
Let $\mathcal{P}_y=\{K \cap \beta(x) \mid K \text{ is a connected component of }H_y \wedge K \cap \beta(x) \neq \emptyset\}$.
Similarly define $H_z$, $H'_z$, and $\mathcal{P}_z$ based on $\mathcal{C}_z$ and $H$, $H'$, $\mathcal{P}$ based on $\mathcal{C}$.
By assumption $\mathcal{C}$ is compatible with $(\mathcal{P}, in, out)$, where for each $v \in \beta(x)$ vertex $v$ is entered $in(v)$ times and exited $out(v)$ times summing over all walks in $\mathcal{C}$.

Since $E(H')=E(H'_y) \cup E(H'_z)$, partition $\mathcal{P}$ is the finest common coarsening of $\mathcal{P}_y$ and $\mathcal{P}_z$.
Let $in_y,out_y, in_z, out_z \colon \beta(x) \to \{0, \ldots, \nu\}$ be such that for each $v \in \beta(x)$ vertex $v$ is entered $in_y(v)$ times and exited $out_y(v)$ times by the walks in $\mathcal{C}_y$ and entered $in_z(v)$ times and exited $out_z(v)$ times by the walks in $\mathcal{C}_z$.
By definition $\mathcal{C}_y$ is compatible with $(\mathcal{P}_y, in_y, out_y)$ at $y$ and $\mathcal{C}_z$ is compatible with $(\mathcal{P}_z, in_z, out_z)$ at $z$.
Therefore, by induction hypothesis, $A[y,\mathcal{P}_y, in_y, out_y] \le h_y$ and $A[z,\mathcal{P}_z, in_z, out_z] \le h_z$.

Since the walks in $\mathcal{C}_y$ and $\mathcal{C}_z$ form a partition of the walks in $\mathcal{C}$, we have for every $v \in \beta(x)$ that $in(v)=in_y(v)+in_z(v)$ and $out(v)=out_y(v)+out_z(v)$.
Therefore $A[x,\mathcal{P}, in, out] \le A[y,\mathcal{P}_y, in_y, out_y] + A[z,\mathcal{P}_z, in_z, out_z] \le  h_y + h_z =h$, finishing this implication.

Now we know that $A[r,\{\{w_0\}\}, w_0 \mapsto i, w_0 \mapsto i]$ is the shortest length of a sequence $\mathcal{C}$ of walks in $G_r=G$ that start and end in $w_0$, each vertex in $\WP \cap (\gamma(r) \setminus \beta(r)) = \WP \setminus \{w_0\}$ is traversed at least once and not more than $\nu$ times, $w_0$ is traversed exactly $i \le \nu$ times and each arc is traversed at most $\kappa(e)$ times. Note that a closed walk which is a solution to the instance forms such a sequence of one walk for an appropriate $i \in \{1, \ldots, \nu\}$.
Conversely, any such sequence can be turned into a single closed walk by concatenating the walks and the resulting walk represent a solution to the instance.
Therefore, the answer of the algorithm is correct.
\end{proof}
}

\section{W[1]-hardness with Respect to Distance to Constant Treedepth}\label{sec:d-to-td}

\toappendix{
\section{Omitted Details and Proofs from \Cref{sec:d-to-td}}
}
In this section we prove \cref{thm:dwp_td_hardness}.

\thmtd*

We provide a reduction from \textsc{Capacitated Dominating Set} (\CDSshort), which is defined as follows.
We are given an undirected graph $G$ and a \emph{capacity function} $c \colon V(G) \to \mathbb N$, where $1 \le c(u) \le \deg_G(u)$ for every $u \in V(G)$.
A set of vertices $S \subseteq V(G)$ is a \emph{capacitated dominating set} if there exists a \emph{domination mapping} $f \colon (V(G) \setminus S) \to S$ such that for every $v \in (V(G) \setminus S)$ we have $\{f(v),v\} \in E(G)$ and
each vertex $s\in S$ dominates at most $c(s)$ vertices, i.e., $|f^{-1}(s)| \le c(s)$.
Given a $k\in\mathbb{N}$, the problem is to find a capacitated dominating set of size $k$.
\toappendix{
\problemQuestion{\textsc{Capacitated Dominating Set} (CDS)}{
    An (undirected) graph $G$, a capacity function $c$, and positive integer $k$.
}{
    Does there exist a capacitated dominating set $S \subseteq V(G)$ of size $|S| \le k$?
}
}%
CDS was shown to be \Wh with respect to treewidth in~\cite{DomLSV08}. %
In order to prove \Cref{thm:dwp_td_hardness}, we actually need to observe a slight improvement of the construction proposed in~\cite{DomLSV08}: we require that CDS is \Wh with respect to the distance to constant treedepth.
This \Whness is obtained by carefully considering the structure of the graph used in the reduction in~\cite{DomLSV08}.
Formally:

\toappendix{
CDS was shown to be \Wh with respect to treewidth by Dom, Lokshtanov, Saurabh, and Villanger~\cite{DomLSV08}.
They prove that the graph has bounded treewidth by having a parameter-sized modulator to a forest.
To prove that \textsc{Directed Waypoint Routing} is \Wh by treedepth we adapt this reduction of Fomin et al.\ by (in essence) replacing high order bipartite subgraphs with a construction that exploits edges of bounded capacities.
However, to prove that the resulting graph has bounded distance to constant treedepth, we require that CDS is \Wh for the distance to a constant treedepth.
We get this \Whness, in the theorem below, we carefully consider the structure of the modulator to forest from \cite{DomLSV08} to conclude that the forest is a disjoint union of stars.
Hence, the reduction in their proof creates a graph with bounded distance to stars.
This parameter implies not only bounded treewidth but also bounded treedepth, and more importantly for us, bounded distance to treedepth 2.
}
\begin{theorem}[\appmark]\label{thm:cds_whard_stars}
    \CDSshort is \Wh with respect to the distance to stars.
\end{theorem}
\toappendix{
\begin{proof}[Proof of \Cref{thm:cds_whard_stars}]
    \Whness of CDS with respect to the treewidth of the input graph is shown by a reduction~\cite{DomLSV08} from the \textsc{Multicolored-Clique} problem.
    To argue that the result has a bounded distance to stars, let us revisit the constructed graph.
    Say we have a \textsc{Multicolored-Clique} instance $(G,k)$ where the graph has a $k$-partition $V_1,\dots,V_k$ and our task is to find a clique $K_k$ which has exactly one vertex in each of the parts.
    It can be assumed that $|V_i|=|V_j|$ for each $i,j \in [k]$, let $N=|V_i|$ be the size of each part.
    Let $E_{ij}$ denote all edges between $V_i$ and $V_j$.
    We reiterate how the authors reduce from $(G,k)$ to a CDS instance $(H,c,k')$.
    We argue just that the $H$ has bounded distance to stars and omit the proof of correctness which follows from the source paper~\cite{DomLSV08}.

    The reduction uses the following two simplifying notions.
    \emph{Marked vertex} is one which at the end of the reduction gets additional $k'+1$ leaves and $k'+1$ capacity -- this forces the solution to contain such a vertex.
    Second notion is an \emph{$(s,\ell)$-arrow} from $u$ to $v$ which stands for adding $s$ once-subdivided edges between $u$ and $v$, and then adding $\ell$ leaves to $v$.

    Let $k'=7k(k-1)+2k$.
    Let $c_x$ and $c'_x$ be positive integers that depend on $x \in G$, their precise value is irrelevant to this proof.
    Now we list all the vertices and edges of $H$, see \Cref{fig:cds-whard-stars}:
    \begin{itemize}
        \item For each $i \in [k]$ there is a marked vertex $\hat x_i$,
        \item for each $v \in V_i$ there is a vertex $\bar v$ that neighbors $\hat x_i$,
        \item for each $i \in [k]$ there are $k'-1$ vertices $S_i$ where each is adjacent to every vertex $\bar v$ where $v \in V_i$,
        \item for each $i,j \in [k], i \ne j$ there are vertices $\hat y_{ij}$ and $\hat z_{ij}$ and  for every $v \in V_i$ there is an $(c_v,c'_v)$-arrow from $\bar v$ to $\hat y_{ij}$ and an $(c'_v,c_v)$-arrow from $\bar v$ to $\hat z_{ij}$,
        \item for each $i,j \in [k], i < j$ there are five marked vertices $\hat x_{ij}$, $\hat p_{ij}$, $\hat p_{ji}$, $\hat q_{ij}$, $\hat q_{ji}$,
        \item for each $e = \{u,v\} \in E_{ij}$ where $u \in V_i$ and $v \in V_j$ there is a vertex $\bar e$ that neighbors $\hat x_{ij}$, also there is a $(c'_u,c_u)$-arrow from $\bar e$ to $p_{ij}$, a $(c_u,c'_u)$-arrow from $\bar e$ to $q_{ij}$, a $(c'_v,c_v)$-arrow from $\bar e$ to $p_{ji}$, and a $(c_v,c'_v)$-arrow from $\bar e$ to $q_{ji}$, and there are $k'+1$ vertices $S_{ij}$ where each is adjacent to every vertex $\bar e$ where $e \in E_{ij}$,
        \item last, for each $i,j \in [k], i \ne j$ there are marked vertices $\hat r_{ij}$ and $\hat s_{ij}$ with $(2N,0)$-arrows from $\hat y_{ij}$ to $\hat r_{ij}$, the same arrows from $\hat p_{ij}$ to $\hat r_{ij}$, from $\hat z_{ij}$ to $\hat s_{ij}$, and from $\hat q_{ij}$ to $\hat s_{ij}$.
    \end{itemize}
    Each of the vertices is assigned a capacity through the capacity function, which we ignore as it is unnecessary to this proof.
    Let $M_1$ be the union of $S_i \cup \{\hat x_i\}$ for each $i \in [k]$, $M_2$ be the union of $S_{ij} \cup \{\hat x_{ij},\hat y_{ij},\hat z_{ij},\hat p_{ij},\hat q_{ij},\hat r_{ij},\hat s_{ij}\}$ for each $i,j \in [k], i \ne j$.
    Set $M_1$ has size $(k'+1+1) \cdot k$ while $M_2$ has size $(k'+1+7) \cdot \binom k2$.
    Let $M= M_1 \cup M_2$.
    We see that $|M| = 7k^4/2 + 6k^3 - 7k^2/2 -2k$, i.e., it is bounded in terms of $k$, it remains to show that $M$ is a modulator to stars.

    \begin{figure}[tbh]
        \centering
        \begin{tikzpicture}[scale=0.6,inner sep=0.5mm]

    \node[draw, circle, line width=0.5pt, fill=white](xi) at (-1,3.5)[label=left: $\hat x_i$] {};

    \node[draw, circle, line width=0.5pt, fill=black](vi1) at (1,6)[] {};
    \node[draw, circle, line width=0.5pt, fill=black](vi2) at (1,5)[] {};
    \node[draw, circle, line width=0.5pt, fill=black](vij) at (1,3)[label=above left: $\bar u$] {};
    \node[draw, circle, line width=0.5pt, fill=black](vin) at (1,1)[] {};
    \node[fit=(vi1)(vi2)(vij)(vin), draw, black, rounded corners, thick, inner sep=1em, label=-90:$V_i$] {};

    \begin{pgfonlayer}{bg}    
        \foreach \i in {0,...,3}{
            \node[draw, circle, line width=0.5pt, fill=black] (S\i) at ($(-1,2.0-\i*.5)$) {};
            \foreach \j in {vi1,vi2,vij,vin}{
                \draw[-, line width=0.5pt, color=gray] (S\i) -- (\j);
            }
        }
    \end{pgfonlayer}
    \node[fit=(S0)(S1)(S2)(S3), draw, black, rounded corners, inner sep=.2em, label=180:$S_i$] {};

    \path (vi2) -- (vij) node [black, midway, sloped] {$\dots$};
    \path (vij) -- (vin) node [black, midway, sloped] {$\dots$};

    \draw[-, line width=0.5pt, color=black]  (xi) -- (vi1);
    \draw[-, line width=0.5pt, color=black]  (xi) -- (vi2);
    \draw[-, line width=0.5pt, color=black]  (xi) -- (vij);
    \draw[-, line width=0.5pt, color=black]  (xi) -- (vin);

    \node[draw, circle, line width=0.5pt, fill=white](yi1) at (3.5,7)[] {};
    \node[draw, circle, line width=0.5pt, fill=white](zi1) at (3.5,6)[] {};
    \node[draw, circle, line width=0.5pt, fill=white](yij) at (3.5,4)[label=above: $\hat y_{ij}$] {};
    \node[draw, circle, line width=0.5pt, fill=white](zij) at (3.5,3)[label=below: $\hat z_{ij}$] {};
    \node[draw, circle, line width=0.5pt, fill=white](yik) at (3.5,1)[] {};
    \node[draw, circle, line width=0.5pt, fill=white](zik) at (3.5,0)[] {};

    \path (zi1) -- (yij) node [black, midway, sloped] {$\quad\dots$};
    \path (zij) -- (yik) node [black, midway, sloped] {$\dots\quad$};


    \draw[-{stealth'}, line width=1.5pt, color=black]  (vij) -- (yi1);
    \draw[-{stealth'}, line width=1.5pt, color=black]  (vij) -- (zi1);
    \draw[-{stealth'}, line width=1.5pt, color=black]  (vij) -- (yij);
    \draw[-{stealth'}, line width=1.5pt, color=black]  (vij) -- (zij);
    \draw[-{stealth'}, line width=1.5pt, color=black]  (vij) -- (yik);
    \draw[-{stealth'}, line width=1.5pt, color=black]  (vij) -- (zik);

    \begin{scope}[yshift=-9cm]
        \node[draw, circle, line width=0.5pt, fill=white](xi) at (-1,3.5)[label=left: $\hat x_j$] {};

        \node[draw, circle, line width=0.5pt, fill=black](vi1) at (1,6)[] {};
        \node[draw, circle, line width=0.5pt, fill=black](vi2) at (1,5)[] {};
        \node[draw, circle, line width=0.5pt, fill=black](vij) at (1,3)[label=above left: $\bar v$] {};
        \node[draw, circle, line width=0.5pt, fill=black](vin) at (1,1)[] {};
        \node[fit=(vi1)(vi2)(vij)(vin), draw, black, rounded corners, thick, inner sep=1em, label=$V_j$] (box) {};

        \begin{pgfonlayer}{bg}    
            \foreach \i in {0,...,3}{
                \node[draw, circle, line width=0.5pt, fill=black] (S\i) at ($(-1,2.0-\i*.5)$) {};
                \foreach \j in {vi1,vi2,vij,vin}{
                    \draw[-, line width=0.5pt, color=gray] (S\i) -- (\j);
                }
            }
        \end{pgfonlayer}
        \node[fit=(S0)(S1)(S2)(S3), draw, black, rounded corners, inner sep=.2em, label=180:$S_j$] {};

        \path (vi2) -- (vij) node [black, midway, sloped] {$\dots$};
        \path (vij) -- (vin) node [black, midway, sloped] {$\dots$};

        \draw[-, line width=0.5pt, color=black]  (xi) -- (vi1);
        \draw[-, line width=0.5pt, color=black]  (xi) -- (vi2);
        \draw[-, line width=0.5pt, color=black]  (xi) -- (vij);
        \draw[-, line width=0.5pt, color=black]  (xi) -- (vin);

        \node[draw, circle, line width=0.5pt, fill=white](yi1) at (3.5,7)[] {};
        \node[draw, circle, line width=0.5pt, fill=white](zi1) at (3.5,6)[] {};
        \node[draw, circle, line width=0.5pt, fill=white](yji) at (3.5,4)[label=above: $\hat y_{ji}$] {};
        \node[draw, circle, line width=0.5pt, fill=white](zji) at (3.5,3)[label=below: $\hat z_{ji}$] {};
        \node[draw, circle, line width=0.5pt, fill=white](yik) at (3.5,1)[] {};
        \node[draw, circle, line width=0.5pt, fill=white](zik) at (3.5,0)[] {};

        \path (zi1) -- (yji) node [black, midway, sloped] {$\quad\dots$};
        \path (zji) -- (yik) node [black, midway, sloped] {$\dots\quad$};


        \draw[-{stealth'}, line width=1.5pt, color=black]  (vij) -- (yi1);
        \draw[-{stealth'}, line width=1.5pt, color=black]  (vij) -- (zi1);
        \draw[-{stealth'}, line width=1.5pt, color=black]  (vij) -- (yji);
        \draw[-{stealth'}, line width=1.5pt, color=black]  (vij) -- (zji);
        \draw[-{stealth'}, line width=1.5pt, color=black]  (vij) -- (yik);
        \draw[-{stealth'}, line width=1.5pt, color=black]  (vij) -- (zik);

    \end{scope}

    \begin{scope}[xshift=13cm,yshift=-4.5cm]
        \node[draw, circle, line width=0.5pt, fill=white](xij) at (1,3.5)[label=right: $\hat x_{ij}$] {};

        \node[draw, circle, line width=0.5pt, fill=black](ei1) at (-1,6)[] {};
        \node[draw, circle, line width=0.5pt, fill=black](ei2) at (-1,5)[] {};
        \node[draw, circle, line width=0.5pt, fill=black](eij) at (-1,3)[label=above right:$\bar e$] {};
        \node[draw, circle, line width=0.5pt, fill=black](ein) at (-1,1)[] {};
        \node[fit=(ei1)(ei2)(eij)(ein), draw, black, rounded corners, thick, inner sep=1em, label=-90:$E_{ij}$] {};

        \begin{pgfonlayer}{bg}    
            \foreach \i in {0,...,3}{
                \node[draw, circle, line width=0.5pt, fill=black] (S\i) at ($(+1,2.0-\i*.5)$) {};
                \foreach \j in {ei1,ei2,eij,ein}{
                    \draw[-, line width=0.5pt, color=gray] (S\i) -- (\j);
                }
            }
        \end{pgfonlayer}
        \node[fit=(S0)(S1)(S2)(S3), draw, black, rounded corners, inner sep=.2em, label=0:$S_{ij}$] {};

        \path (ei2) -- (eij) node [black, midway, sloped] {$\dots$};
        \path (eij) -- (ein) node [black, midway, sloped] {$\dots$};

        \draw[-, line width=0.5pt, color=black]  (xij) -- (ei1);
        \draw[-, line width=0.5pt, color=black]  (xij) -- (ei2);
        \draw[-, line width=0.5pt, color=black]  (xij) -- (eij);
        \draw[-, line width=0.5pt, color=black]  (xij) -- (ein);

        \node[draw, circle, line width=0.5pt, fill=white](pij) at (-3.5,6)[label=above: $\hat p_{ij}$] {};
        \node[draw, circle, line width=0.5pt, fill=white](qij) at (-3.5,5)[label=below: $\hat q_{ij}$] {};
        \node[draw, circle, line width=0.5pt, fill=white](pji) at (-3.5,2)[label=above: $\hat p_{ji}$] {};
        \node[draw, circle, line width=0.5pt, fill=white](qji) at (-3.5,1)[label=below: $\hat q_{ji}$] {};

        \draw[-{stealth'}, line width=1.5pt, color=black]  (eij) -- (pij);
        \draw[-{stealth'}, line width=1.5pt, color=black]  (eij) -- (qij);
        \draw[-{stealth'}, line width=1.5pt, color=black]  (eij) -- (pji);
        \draw[-{stealth'}, line width=1.5pt, color=black]  (eij) -- (qji);

        \node[draw, circle, line width=0.5pt, fill=white](rij) at (-6.5,7)[label=above: $\hat r_{ij}$] {};
        \node[draw, circle, line width=0.5pt, fill=white](sij) at (-6.5,6)[label=below: $\hat s_{ij}$] {};
        \node[draw, circle, line width=0.5pt, fill=white](rji) at (-6.5,1)[label=above: $\hat r_{ji}$] {};
        \node[draw, circle, line width=0.5pt, fill=white](sji) at (-6.5,0)[label=below: $\hat s_{ji}$] {};

    \end{scope}

    \draw[-{stealth'}, line width=1.5pt, color=black]  (yij) -- (rij);
    \draw[-{stealth'}, line width=1.5pt, color=black]  (zij) -- (sij);
    \draw[-{stealth'}, line width=1.5pt, color=black]  (yji) -- (rji);
    \draw[-{stealth'}, line width=1.5pt, color=black]  (zji) -- (sji);

    \draw[-{stealth'}, line width=1.5pt, color=black]  (pij) -- (rij);
    \draw[-{stealth'}, line width=1.5pt, color=black]  (qij) -- (sij);
    \draw[-{stealth'}, line width=1.5pt, color=black]  (pji) -- (rji);
    \draw[-{stealth'}, line width=1.5pt, color=black]  (qji) -- (sji);

\end{tikzpicture}
        \caption{
            Part of the graph $H$ described in the proof of \Cref{thm:cds_whard_stars} from \cite{DomLSV08} that models an edge $e=uv$.
            Marked vertices are depicted as empty, thick arrows stand for $(c,c')$-arrows for some integers $c,c'$.
        }%
        \label{fig:cds-whard-stars}
    \end{figure}

    Observe that adding an arrow from $u$ and $v$ creates only edges that are incident to either $u$ or $v$.
    Let $C$ be the set that contains $\bar v$ for each $v \in V_i$, $\bar e$ for each $e \in E_{ij}$.
    We claim that $H-M-C$ contains no edges and adding back the independent set $C$ results in the graph $H-M$ being a disjoint union of stars.
    First, note that all arrows have either both endpoints in $M$ or one endpoint in $M$ and the second in $C$, implying that none of the edges added due to arrows are present in $H-M-C$.
    The (non-arrow) edges that were added to $H$ are $\{o_i,\bar v\}$ for each $i \in [k], o_i \in \{\hat x_i\} \cup S_i, v \in V_i$, and similarly $\{o_{ij},\bar e\}$ for each $i,j \in[k], o_{ij} \in \{\hat x_{ij}\} \cup S_{ij}, \bar e \in E_{ij}$, which are edges between $M$ and $C$.
    Finally, all the marked vertices are in $M$, therefore all leaves attached to the marked vertices at the end of the construction are isolated in $H-M$.
\end{proof}

Note that Fomin et al.~\cite[Section 5]{Fomin2010} showed \Whness of \textsc{Hamiltonian Cycle} with respect to clique-width by a reduction from CDS parameterized by treewidth.
The improved analysis of \Cref{thm:cds_whard_stars} does not strengthen this result as the construction still creates high order bipartite subgraphs.
}

We are now ready to proceed with the sketch of the proof of~\Cref{thm:dwp_td_hardness}. A more detailed version can be found in the appendix.

\begin{proof}[Sketch of Proof of \Cref{thm:dwp_td_hardness}.]
    Let $I=(G=(V,E),c,k)$ be an instance of CDS. %
    We describe the construction of an instance $I'=(G',\WP,\wFn,\cFn,\budget)$ of \DWRPshort that is equivalent to $I$ such that $G'$ has a bounded modulator to constant treedepth.

    \proofsubparagraph{The construction.}
    First, we describe the construction of the used gadgets.
    Each gadget has dedicated \emph{connecting vertices} marked with `in' (e.g. $u_{\rm in}$), `out', or `io'.
    Every gadget of the final $G'$ will be connected to other gadgets by edges incident only to its connecting vertices: `in' (`out' resp.) vertices can only be incident to arcs that are coming `into' (going `out' from resp.) the gadget, and `io' can be incident to both types of arcs.
    If this is true then we say that the graph \emph{contains} the gadget.
    We use $u^X$ to denote the vertex $u$ that belongs to a gadget $X$.
    All edges have weight $1$.
    Unless stated otherwise, every edge has capacity $n$; this is always sufficient by \Cref{obs:n_traversals}.
    Crucially, getting from one terminal to another requires the traversal of exactly two non-terminal vertices.
    As the final budget is set to exactly $3 \cdot (\text{number of terminals})$, it follows that every terminal is visited exactly once.

    We begin with the \emph{force-$p$-traversals-gadget}.
    It consists of three non-terminal vertices $u_{\rm in},u_{\rm out},w$ and terminal vertices $v_1,\dots,v_p$ joined with edges $(u_{\rm in},v_i),(v_i,w)$ for every $i \in [p]$ and $(w,u_{\rm out})$ of capacity $p$, see \Cref{fig:force-traversal-b}.

    \begin{figure}[tbh]
        \centering
        \begin{tikzpicture}[def,scale=0.8]
            \node (uin) at (0,1) {};
            \node (w) at (2,1) {};
            \node (uout) at (3,1) {};
            \node[hide] at (0,1.4) {$u_{\rm in}$};
            \node[hide] at (1,2.3) {$v_i$};
            \node[hide] at (2,1.4) {$w$};
            \node[hide] at (3,1.4) {$u_{\rm out}$};
            \foreach \y in {1,...,5}{
                \node[fill] (v\y) at (1,.5*\y-.5) {};
                \draw (uin) -- (v\y);
                \draw (v\y) -- (w);
            }
            \draw (w) --node[hide,below]{$5$} (uout);
            \foreach \x in {1,...,8}{
                \node[hide] (q\x) at (-2.4+.3*\x,-0.4) {};
                \draw (q\x) edge[densely dotted,bend left=5] (uin);
                \node[hide] (e\x) at (2.6+.3*\x,-0.4) {};
                \draw (uout) edge[densely dotted,bend left=5] (e\x);
            }
        \end{tikzpicture}
        \caption{Force-5-traversals-gadget. Black and white vertices represent terminals and non-terminals, respectively.}%
        \label{fig:force-traversal-b}
    \end{figure}

    \begin{claim}\label{claim:traversal_gadget}
        Let $G'$ be a graph that contains a force-$p$-traversals-gadget $X$.
        Then every solution to \DWRPshort in $G'$ traverses from $u^X_{\rm in}$ to $u^X_{\rm out}$ exactly $p$ times.
    \end{claim}

    Next, we need the \emph{cover-$z$-gadget}, which is illustrated in~\Cref{fig:edge-gadget-adapted-b}. Let $R_1$ be the directed path that starts at $s_{\rm in}$, goes to $f$, then to $b$, continues to $x_{\rm in}$, traverses all the $w$s' until $y_{\rm out}$, goes to $c$ and finishes in $t_{\rm out}$. Also, let $R_2$ and $P$ be the red and blue paths depicted in~\Cref{fig:edge-gadget-adapted-b}. The idea here is that this gadget has two possible behaviors: either its terminals are covered by the combination of the $R_2$ and $P$ paths, or all of its terminals except from $z_{\rm io}$ are covered by the $R_1$ path, and $z$ is covered by some external path that visits this gadget. Formally:

    \begin{figure}[tbh]
        \centering
        \begin{tikzpicture}[def,scale=0.8]
            \newcommand{\customarc}[5]{
                \node[fill=white] (t1) at ($(#1)!#5!(#2)$) {};
                \node[fill=white] (t2) at ($(#1)!1-#5!(#2)$) {};
                \draw[-{Stealth[scale=#3]}] (#1) edge[bend right=#4] (t1);
                \draw[-{Stealth[scale=#3]}] (t1) edge[bend right=#4] (t2);
                \draw[-{Stealth[scale=#3]}] (t2) edge[bend right=#4] (#2);
            }
            \newcommand{\customarccolor}[6]{
                \node[fill=white] (t1) at ($(#1)!#5!(#2)$) {};
                \node[fill=white] (t2) at ($(#1)!1-#5!(#2)$) {};
                \draw[-{Stealth[scale=#3]},#6] (#1) edge[bend right=#4] (t1);
                \draw[-{Stealth[scale=#3]},#6] (t1) edge[bend right=#4] (t2);
                \draw[-{Stealth[scale=#3]},#6] (t2) edge[bend right=#4] (#2);
            }
            \newcommand{\customedgecolor}[7]{
                \node[fill=white] (t1) at ($(#1)!#5!(#2)$) {};
                \node[fill=white] (t2) at ($(#1)!1-#5!(#2)$) {};
                \draw[-{Stealth[scale=#3]},#6] (#1) edge[bend right=#4] (t1);
                \draw[-{Stealth[scale=#3]},#6] (t1) edge[bend right=#4] (t2);
                \draw[-{Stealth[scale=#3]},#6] (t2) edge[bend right=#4] (#2);
                \draw[-{Stealth[scale=#3]},#7] (t1) edge[bend right=#4] (#1);
                \draw[-{Stealth[scale=#3]},#7] (t2) edge[bend right=#4] (t1);
                \draw[-{Stealth[scale=#3]},#7] (#2) edge[bend right=#4] (t2);
            }
            \begin{scope}[xshift=4cm,yscale=1.4,xscale=2.8]
                \node[fill] (z) at (2,-.5) {};
                \node[fill] (b) at (0,0) {};
                \node[fill] (a) at ($(b)+(90:1)$) {};
                \node[fill] (x) at ($(a)+(90:1)$) {};
                \node[fill] (f) at ($(b)+(155:.7)$) {};
                \node[fill] (e) at ($(f)+(90:1)$) {};
                \node[fill] (s) at ($(e)+(90:1)$) {};
                \node[fill] (c) at (4,0) {};
                \node[fill] (d) at ($(c)+(90:1)$) {};
                \node[fill] (y) at ($(d)+(90:1)$) {};
                \node[fill] (h) at ($(c)+(25:.7)$) {};
                \node[fill] (g) at ($(h)+(90:1)$) {};
                \node[fill] (t) at ($(g)+(90:1)$) {};
                \customarccolor{g}{t}{0.9}{10}{.33}{red}
                \customarccolor{h}{g}{0.9}{10}{.33}{red}
                \customarc{c}{h}{0.9}{20}{.33}
                \customedgecolor{d}{c}{0.9}{10}{.33}{blue}{black}
                \customedgecolor{y}{d}{0.9}{10}{.33}{blue}{black}
                \customedgecolor{a}{x}{0.9}{10}{.33}{black}{blue}
                \customedgecolor{b}{a}{0.9}{10}{.33}{black}{blue}
                \customarc{f}{b}{0.9}{20}{.33}
                \customarccolor{e}{f}{0.9}{10}{.33}{red}
                \customarccolor{s}{e}{0.9}{10}{.33}{red}
                \foreach \w in {1,...,9}{
                    \node[fill] (w\w) at ($(x)+(\w/2.5,0)$) {};
                }
                \customarc{w3}{w4}{0.9}{50}{.33}
                \customarc{w6}{w7}{0.9}{50}{.33}
                \foreach \w in {2,3,5,6,8,9}{
                    \pgfmathtruncatemacro\ww{\w-1}
                    \customedgecolor{w\ww}{w\w}{0.9}{50}{.33}{black}{red}
                }
                \customarc{x}{w1}{0.9}{50}{.33}
                \customarc{w9}{y}{0.9}{50}{.33}
                \customarccolor{w7}{h}{0.9}{6}{.42}{red}
                \customarccolor{f}{w3}{0.9}{6}{.42}{red}
                \customarccolor{z}{c}{0.9}{10}{.33}{blue}
                \customarccolor{b}{z}{0.9}{10}{.33}{blue}
                \node[fill=white] (t1) at ($(w1)!.33!(w6)+(0,0.4)$) {};
                \node[fill=white] (t2) at ($(w1)!.67!(w6)+(0,0.4)$) {};
                \draw (w1) edge[bend left=8, red] (t1);
                \draw (t1) edge[bend left=8, red] (t2);
                \draw (t2) edge[bend left=8, red] (w6);
                \begin{pgfonlayer}{bg}
                    \begin{scope}[red!30,line width=4pt]
                        \draw[line cap=rect,-] (w1) -- (t1);
                        \draw[line cap=rect,-] (t1) -- (t2);
                        \draw[line cap=rect,-] (t2) -- (w6);
                    \end{scope}
                \end{pgfonlayer}
                \node[fill=white] (t1) at ($(w4)!.33!(w9)+(0,0.4)$) {};
                \node[fill=white] (t2) at ($(w4)!.67!(w9)+(0,0.4)$) {};
                \draw (w4) edge[bend left=8, red] (t1);
                \draw (t1) edge[bend left=8, red] (t2);
                \draw (t2) edge[bend left=8, red] (w9);
                \begin{pgfonlayer}{bg}
                    \begin{scope}[red!30,line width=4pt]
                        \draw[line cap=rect,-] (w4) -- (t1);
                        \draw[line cap=rect,-] (t1) -- (t2);
                        \draw[line cap=rect,-] (t2) -- (w9);
                    \end{scope}
                \end{pgfonlayer}
                \node[hide,anchor=north] at (z.south) {$z$};
                \node[hide,anchor=north] at (b.south) {$b$};
                \node[hide,anchor=north] at (f.south) {$f$};
                \node[hide,anchor=north] at (c.south) {$c$};
                \node[hide,anchor=north] at (h.south) {$h$};
                \node[hide,anchor=west] at (g.east) {$g$};
                \node[hide,anchor=west] at (t.east) {$t_{\rm out}$};
                \node[hide,anchor=east] at (a.west) {$a$};
                \node[hide,anchor=east] at (s.west) {$s_{\rm in}$};
                \node[hide,anchor=west] at (e.east) {$e$};
                \node[hide,anchor=west] at (d.east) {$d$};
                \node[hide,anchor=south] at (y.north) {$y_{\rm out}$};
                \node[hide,anchor=south] at (x.north) {$x_{\rm in}$};
                \foreach \w in {1,...,9}{
                    \node[hide,below] at (w\w.south) {$w_\w$};
                }
                \begin{pgfonlayer}{bg}
                    \begin{scope}[red!30,line width=4pt]
                        \draw[line cap=rect,-] (s)  -- (f);
                        \draw[line cap=rect,-] (f)  -- (w3);
                        \draw[line cap=rect,-] (w3) -- (w1);
                        \draw[line cap=rect,-] (w6) -- (w4);
                        \draw[line cap=rect,-] (w9) -- (w7);
                        \draw[line cap=rect,-] (w7) -- (h);
                        \draw[line cap=rect,-] (h)  -- (t);
                    \end{scope}
                    \begin{scope}[blue!30,line width=4pt]
                        \draw[line cap=rect,-] (x)  -- (b);
                        \draw[line cap=rect,-] (b)  -- (z);
                        \draw[line cap=rect,-] (z)  -- (c);
                        \draw[line cap=rect,-] (c)  -- (y);
                    \end{scope}
                \end{pgfonlayer}
            \end{scope}
        \end{tikzpicture}
        \caption{
            The cover-$z$-gadget. Black and white vertices represent terminals and non-terminals, respectively. The path $R_2$ is colored red, while the path $P$ is colored blue.
        }%
        \label{fig:edge-gadget-adapted-b}
    \end{figure}

    \begin{claim}\label{claim:ltwogadget_variant}
        Let $G$ be a graph that contains the cover-$z$-gadget $X$.
        If $S$ is a minimum solution to \DWRPshort that traverses every terminal vertex of $X$ exactly once, then $S$ includes either the path $R_1$, or the paths $P$ and $R_2$ as continuous subpaths.
    \end{claim}

    We now present the construction of $(G',\WP,\wFn,\cFn,\budget)$, see~\Cref{fig:whard-construction-b}.
    Let $\mathcal S$ be a force-$(k+1)$-traversals-gadget that we call \emph{selection-gadget} and let $\mathcal A$ be a force-$(2|E|+|V|+1)$-traversals-gadget.
    We create an auxiliary non-terminal vertex $x$ and add an arc $(x,u^{\mathcal A}_{\rm in})$.
    Let $x$ together with $\mathcal A$ be called the \emph{auxiliary-gadget}.
    We connect the selection- and auxiliary-gadgets by creating terminal vertices $x_1,x_2$, non-terminal vertex $x_3$, and arcs $(u^{\mathcal A}_{\rm out},x_2)$, $(x_2,x_3)$, $(x_3,u^{\mathcal S}_{\rm in})$, $(u^{\mathcal S}_{\rm out}, x_1)$, $(x_1,x)$. These two gadgets compose the \emph{upper part} of $G'$.

    Finally, we encode the vertices and edges of $G$ into $G'=(V',E')$. For every vertex $u \in V$ we add vertices $z_u$, $c_u$, and $d_u$ to $V'$, then we create a cover-$z_u$-gadget denoted $\mathcal E^u$.
    We add arcs $(u^{\mathcal S}_{\rm out}, x^{\mathcal E^u}_{\rm in})$, $(y^{\mathcal E^u}_{\rm out}, c_u)$, $(c_u,d_u)$, $(c_u,u^{\mathcal S}_{\rm in})$, $(u^{\mathcal A}_{\rm out}, s^{\mathcal E^u}_{\rm in})$, $(t^{\mathcal E^u}_{\rm out}, x)$ to $E'$, and we set the capacity of $(c_u,d_u)$ to $c(u)$.
    For every arc $(u,v)$ such that $\{u,v\} \in E$, we create a cover-$z_v$-gadget denoted $\mathcal E^{u,v}$.
    Then we add arcs $(d_u,x^{\mathcal E^{u,v}}_{\rm in})$, $(y^{\mathcal E^{u,v}}_{\rm out}, c_u)$, $(u^{\mathcal A}_{\rm out}, s^{\mathcal E^{u,v}}_{\rm in})$, $(t^{\mathcal E^{u,v}}_{\rm out}, x)$.
    Note that in the above we created two gadgets, $\mathcal E^{u,v}$ and $\mathcal E^{v,u}$, for each $\{u,v\} \in E$.
    Finally, we set the budget $\budget=132|E|+69|V|+3k+12$. This concludes construction of the \DWRPshort instance.

    \begin{figure}[tbh]
        \centering
        \scalebox{0.9}{
        \newcommand{\convexpath}[2]{
	[
	create hullnodes/.code={
		\global\edef\namelist{#1}
		\foreach [count=\counter] \nodename in \namelist {
			\global\edef\numberofnodes{\counter}
			\node at (\nodename) [draw=none,name=hullnode\counter] {};
		}
		\node at (hullnode\numberofnodes) [name=hullnode0,draw=none] {};
		\pgfmathtruncatemacro\lastnumber{\numberofnodes+1}
		\node at (hullnode1) [name=hullnode\lastnumber,draw=none] {};
	},
	create hullnodes
	]
	($(hullnode1)!#2!-90:(hullnode0)$)
	\foreach [
	evaluate=\currentnode as \previousnode using \currentnode-1,
	evaluate=\currentnode as \nextnode using \currentnode+1
	] \currentnode in {1,...,\numberofnodes} {
		let
		\p1 = ($(hullnode\currentnode)!#2!-90:(hullnode\previousnode)$),
		\p2 = ($(hullnode\currentnode)!#2!90:(hullnode\nextnode)$),
		\p3 = ($(\p1) - (hullnode\currentnode)$),
		\n1 = {atan2(\y3,\x3)},
		\p4 = ($(\p2) - (hullnode\currentnode)$),
		\n2 = {atan2(\y4,\x4)},
		\n{delta} = {-Mod(\n1-\n2,360)}
		in
		{-- (\p1) arc[start angle=\n1, delta angle=\n{delta}, radius=#2] -- (\p2)}
	}
	-- cycle
}

\begin{tikzpicture}[def]
    \newcommand{\traversalgadget}[2]{
        \def\name{#1}
        \def\traversals{#2}
        \node[solid] (uin_\name) at (0,1) {};
        \node[solid] (w) at (2,1) {};
        \node[solid] (uout_\name) at (3,1) {};
        \foreach \y in {1,...,\traversals}{
            \node[solid,fill] (v\y) at (1,2*\y/\traversals-1/\traversals) {};
            \draw (uin_\name) -- (v\y);
            \draw (v\y) -- (w);
        }
        \draw (w) --node[hide,above]{$\traversals$} (uout_\name);
    }
    \traversalgadget{selection}{3}
    \node[hide] at (-.2,1.4) {$u^\mathcal S_{\rm in}$};
    \node[hide] at (3,1.4) {$u^\mathcal S_{\rm out}$};
    \begin{pgfonlayer}{bbg}
     \fill[black!8, rounded corners] (-1.2,0) rectangle (3.4,2);
    \end{pgfonlayer}
    \node[hide] at (-0.8,0.6) {$\mathcal S$};
    \begin{scope}[xshift=6cm,densely dotted]
        \node[solid,label=above:$x$] (xx) at (-1,1) {};
        \traversalgadget{auxiliary}{6}
        \node[hide] at (0,1.4) {$u^\mathcal A_{\rm in}$};
        \node[hide] at (3.2,1.4) {$u^\mathcal A_{\rm out}$};
    \end{scope}
    \begin{pgfonlayer}{bbg}
     \fill[black!8, rounded corners] (4.5,0) rectangle (10,2);
    \end{pgfonlayer}
    \node[hide] at (9.6,0.6) {$\mathcal A$};
    \node[fill,label=right:$x_2$] (x2) at (8.2,2.5) {};
    \node[    ,label=left: $x_3$] (x3) at (0.8,2.5) {};
    \node[fill,label=above:$x_1$] (x1) at (4,1.4) {};
    \draw (uout_selection) -- (x1);
    \draw (x1) -- (xx);
    \draw[densely dotted] (xx) -- (uin_auxiliary);
    \draw (uout_auxiliary) -- (x2);
    \draw (x2) -- (x3);
    \draw (x3) -- (uin_selection);
    \newcommand{\connection}[4]{
        \pgfmathsetmacro\hangle{#3+20}
        \pgfmathsetmacro\langle{#3-20}
        \node[fill] (s#1) at ($(#2)+(\hangle:1.5*#4)$) {};
        \node[fill] (t#1) at ($(#2)+(\langle:1.5*#4)$) {};
        \node[fill] (x#1) at ($(s#1)+(#3:1.*#4)$) {};
        \node[fill] (y#1) at ($(t#1)+(#3:1*#4)$) {};
        \begin{scope}[opacity=.15,transparency group]
            \fill[black] \convexpath{s#1,x#1,y#1,t#1,#2}{5pt};
        \end{scope}
        \begin{pgfonlayer}{bg}    
            \draw[densely dotted] (uout_auxiliary) to (s#1);
            \draw[densely dotted] (t#1) to (xx);
        \end{pgfonlayer}
    }
    \newcommand{\vertexgadget}[1]{
        \node[fill,label=below:$z_#1$] (z#1) at (2,0) {};
        \connection{#1}{z#1}{90}{1.0}
        \node (c#1) at ($(y#1)+(5:1)$) {};
        \node (d#1) at ($(c#1)+(-90:.8)$) {};
        \draw (uout_selection) to[out=-90,in=110] (x#1);
        \draw (y#1) -- (c#1);
        \draw (c#1) -- (d#1);
        \draw (c#1) to[out=140,in=-90] (uin_selection);
    }
    \begin{scope}[xshift=-2cm,yshift=-4cm]
        \vertexgadget{u}
        \node[hide,label=left:$x_{\rm in}$] at (xu){};
        \node[hide,label=left:$y_{\rm out}$] at (yu){};
        \node[hide,label=left:$s_{\rm in}$] at (su){};
        \node[hide,label=left:$t_{\rm out}$] at (tu){};
        \node[hide,label=above right:$c_u$] at (cu){};
        \node[hide,label=above right:$d_u$] at (du){};
        \node[hide] at ($(zu)+(90:.7)$){$\mathcal E^u$};
    \end{scope}
    \begin{scope}[xshift=3cm,yshift=-4cm]
        \vertexgadget{v}
        \node[hide] at ($(zv)+(90:1)$){$\mathcal E^v$};
    \end{scope}
    \begin{scope}[xshift=6cm,yshift=-4cm]
        \vertexgadget{w}
        \node[hide] at ($(zw)+(90:1)$){$\mathcal E^w$};
    \end{scope}
    \connection{vu}{zu}{22}{.7}
    \connection{uv}{zv}{158}{.7}
    \node[hide] at ($(zu)+(25:1.3)$){$\mathcal E^{vu}$};
    \node[hide] at ($(zv)+(155:1.3)$){$\mathcal E^{uv}$};
    \draw (du) -- (xuv);
    \draw (yuv) -- (cu);
    \draw (dv) -- (xvu);
    \draw (yvu) -- (cv);
\end{tikzpicture}
        }
        \caption{An example of a completely constructed \DWRPshort instance for graph $G=(\{u,v,w\}, \{\{u,v\}\})$ and $k=2$ (capacities omitted). The top half contains the upper part of~$G'$. The bottom half corresponds to the vertices and edges of $G$. Gray boxes hide the inner parts of the corresponding cover-gadgets, see \cref{fig:edge-gadget-adapted-b}. Dotted arrows signify arcs that belong to the auxiliary-gadget or connect cover-gadgets to the auxiliary-gadget. Black and white  vertices represent terminals and non-terminals, respectively.}%
        \label{fig:whard-construction-b}
    \end{figure}

    \proofsubparagraph{The instances are equivalent.} On a high level, we show that any solution $C$ of the instance $(G',\WP,\wFn,\cFn,\budget)$ obeying the above budget has a specific behavior. First, the budget is set so that each terminal can be visited at most once. Let $S\subseteq V$ be a solution of $I$. The walk $C$ is separated into two stages. In the first stage, for each $u\in S$ the walk $C$ starts from~$u^\mathcal{S}_{\rm in}$, goes through~$\mathcal{S}$ to~$u^\mathcal{S}_{\rm out}$, from there to $x^{\mathcal E^u}_{\rm in}$, using the $P$ path (visiting $z_u$) of the gadget $\mathcal E^u$ to $y^{\mathcal E^u}_{\rm out}$ and to $c_u$. 
    Then for each vertex $v$ that is dominated by $u$, it goes from $c_u$ to $d_u$, uses the $P$ path of gadget $\mathcal E^{u,v}$ to visit $z_v$ and return to~$c_u$.
    Crucially, this can be done at most $c(u)$ times.
    Finally it returns to~$u^\mathcal{S}_{\rm in}$.    

    Once all the vertices of $S$ are dealt with in this way the first stage is done, and $C$ proceeds to $\mathcal{A}$. For each cover-gadget $\mathcal E$ (which are equal in number to the terminals of $\mathcal{A}$ minus one), the walk $C$ goes from $x$ through $\mathcal A$ to $u^{\mathcal{A}}_{\rm out}$  and then through $\mathcal E$ to cover the ``remaining'' vertices of $\mathcal E$ (i.e., the vertices that were not visited during the first stage); this is done by employing either the $R_1$ or the $R_2$ path, according to whether $\mathcal E$ was traversed during the first step. We then show that $C$ exists and is of the correct budget if and only if $I$ is a yes-instance of CDS.

    \smallskip

    The rest of the proof focuses on proving that $G'$ has a modulator of size $7|M|+7$ to treedepth $86$, using that $G' \setminus M$ has treedepth $2$ (it is a disjoint union of stars).
\end{proof}

\toappendix{
\begin{proof}[Proof of \Cref{thm:dwp_td_hardness}.]
    Let $(G,c,k)$ be an instance of CDS where $G$ is a graph, $c \colon V(G) \to \mathbb{N}_0$ is the capacity function, and $k$ is a maximum size of the solution; let $M$ be the modulator to stars of minimum cardinality in~$G$.
    We build an instance $(G',\WP,\wFn,\cFn,\budget)$ of \DWRPshort that has modulator of size $7\cdot |M|+7$ to treedepth $86$.
    Let us first show a construction of gadgets one by one while noting intuition for how they work.
    Each gadget has dedicated connecting vertices which we mark with `in' (e.g. $u_{\rm in}$), `out', or `io'.
    In the final instance, within $G'$ every gadget will be connected to other gadgets by edges incident only to its connecting vertices; notably `in' vertices can be incident to arcs that are coming into the gadget, `out' vertices can be incident to outgoing arcs, and `io' can be incident to both types of arcs.
    If the above assertions about arcs incident to gadget's vertices in a graph is true then we say that the graph \emph{contains} the gadget.
    When a gadget $X$ contains in its construction a vertex $u$ we use $u^X$ to denote this vertex in the resulting graph.
    All the edges have weight $1$.
    Unless stated otherwise, every edge has capacity $n$; this is always sufficient by \Cref{obs:n_traversals}.
    Last before describing the gadgets, let us note that the crucial property of the construction is that to get from one terminal to another we have to traverse exactly two non-terminal vertices.
    As the final budget is set to exactly $3 \cdot (\text{number of terminals})$ it follows that every terminal is visited exactly once.

    To force $p \in \mathbb N^+$ traversals we define \emph{force-$p$-traversals-gadget}.
    It consists of three non-terminal vertices $u_{\rm in},u_{\rm out},w$ and terminal vertices $v_1,\dots,v_p$ joined with edges $(u_{\rm in},v_i),(v_i,w)$ for every $i \in [p]$ and $(w,u_{\rm out})$ of capacity $p$, see \Cref{fig:force-traversal-b}.

    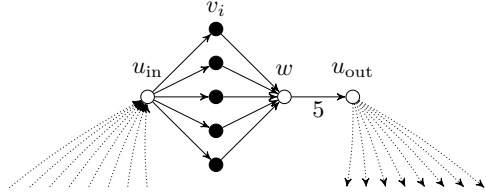
\begin{figure}[tbh]
        \centering
        \scalebox{0.9}{
        \begin{tikzpicture}[def]
            \node (uin) at (0,1) {};
            \node (w) at (2,1) {};
            \node (uout) at (3,1) {};
            \node[hide] at (0,1.4) {$u_{\rm in}$};
            \node[hide] at (1,2.3) {$v_i$};
            \node[hide] at (2,1.4) {$w$};
            \node[hide] at (3,1.4) {$u_{\rm out}$};
            \foreach \y in {1,...,5}{
                \node[fill] (v\y) at (1,.5*\y-.5) {};
                \draw (uin) -- (v\y);
                \draw (v\y) -- (w);
            }
            \draw (w) --node[hide,below]{$5$} (uout);
            \foreach \x in {1,...,8}{
                \node[hide] (q\x) at (-2.4+.3*\x,-0.4) {};
                \draw (q\x) edge[densely dotted,bend left=5] (uin);
                \node[hide] (e\x) at (2.6+.3*\x,-0.4) {};
                \draw (uout) edge[densely dotted,bend left=5] (e\x);
            }
        \end{tikzpicture}
        }
        \caption{Force-5-traversals-gadget}%
        \label{fig:force-traversal}
    \end{figure}

    \begin{claim}\label{claim:traversal_gadget_full}
        Let $G'$ be a graph that contains a force-$p$-traversals-gadget $X$.
        Then every solution to \DWRPshort in $G'$ traverses from $u^X_{\rm in}$ to $u^X_{\rm out}$ exactly $p$ times.
    \end{claim}
    \begin{claimproof}
        Every terminal vertex $v^X_i$ for $i\in[p]$ requires one (unique) traversal from $u^X_{\rm in}$ to $u^X_{\rm out}$.
        The capacity of $(w,u^X_{\rm out})$ upper bounds the number of traversals from $u^X_{\rm in}$ to $u^X_{\rm out}$ by~$p$.
    \end{claimproof}

    In order to define the second gadget, called cover-$z$-gadget, we first introduce an undirected raw-cover-$z$-gadget (see \Cref{fig:raw-edge-gadget}) that was used for \textsc{Hamiltonian Cycle} (gadget is called $L_2$ in \cite[Lemma 9]{Fomin2010}).
    For a fixed vertex $z$ let a \emph{raw-cover-$z$-gadget} be an undirected graph that consist of a path $R_1$ on terminals $s_{\rm in},e,f,b,a,x_{\rm in},w_1,w_2,\dots,w_9,y_{\rm out},d,c,h,g,t_{\rm out}$ (in this order) together with additional edges $\{f,w_3\},\{w_1,w_6\}, \{w_4,w_9\}$, and $\{w_7,h\}$.
    Moreover, the vertex $z$ is also considered to be part of the raw-cover-$z$-gadget; in this context we call it $z_{\rm io}$ and it is connected via edges $\{c,z_{\rm io}\}$ and $\{z_{\rm io},b\}$.
    Let $R_2$ be the path $s_{\rm in},e,f,w_3,w_2,w_1,w_6,w_5,w_4$, $w_9,w_8,w_7,h,g,t_{\rm out}$ and $P$ be the path $x_{\rm in},a,b,z_{\rm io},c,d,y_{\rm out}$.

    \begin{figure}[tbh]
        \centering
        \begin{tikzpicture}[def,scale=0.9,line cap=rect]
            \begin{scope}[xshift=4cm,yscale=.7]
                \node[fill] (z) at (2,-.5) {};
                \node[fill] (b) at (0,0) {};
                \node[fill] (a) at ($(b)+(90:1)$) {};
                \node[fill] (x) at ($(a)+(90:1)$) {};
                \node[fill] (f) at ($(b)+(160:1)$) {};
                \node[fill] (e) at ($(f)+(90:.65)$) {};
                \node[fill] (s) at ($(e)+(90:.65)$) {};
                \node[fill] (c) at (4,0) {};
                \node[fill] (d) at ($(c)+(90:1)$) {};
                \node[fill] (y) at ($(d)+(90:1)$) {};
                \node[fill] (h) at ($(c)+(30:1)$) {};
                \node[fill] (g) at ($(h)+(90:.65)$) {};
                \node[fill] (t) at ($(g)+(90:.65)$) {};
                \draw[-,ultra thick] (s) -- (e) -- (f) -- (b) -- (a) -- (x);
                \draw[-,ultra thick] (y) -- (d) -- (c) -- (h) -- (g) -- (t);
                \foreach \w in {1,...,9}{
                    \node[fill] (w\w) at ($(x)+(\w/2.5,0)$) {};
                }
                \draw[-,ultra thick] (x) -- (y);
                \draw[-] (b) -- (z) -- (c)
                    (f)  edge[bend right=50] (w3)
                    (w1) edge[bend left=50] (w6)
                    (w4) edge[bend left=50] (w9)
                    (w7) edge[bend right=50] (h)
                    ;
                \node[hide,anchor=north] at (z.south) {$z_{\rm io}$};
                \node[hide,anchor=north] at (b.south) {$b$};
                \node[hide,anchor=north] at (f.south) {$f$};
                \node[hide,anchor=north] at (c.south) {$c$};
                \node[hide,anchor=north] at (h.south) {$h$};
                \node[hide,anchor=west] at (g.east) {$g$};
                \node[hide,anchor=west] at (t.east) {$t_{\rm out}$};
                \node[hide,anchor=west] at (a.east) {$a$};
                \node[hide,anchor=east] at (s.west) {$s_{\rm in}$};
                \node[hide,anchor=east] at (e.west) {$e$};
                \node[hide,anchor=east] at (d.west) {$d$};
                \node[hide,anchor=south west] at (y.north west) {$y_{\rm out}$};
                \node[hide,anchor=south east] at (x.north east) {$x_{\rm in}$};
                \node[hide,anchor=north] at (w1.south) {$w_1$};
                \node[hide,anchor=north] at (w9.south) {$w_9$};
                \node[hide,anchor=north] at ($(w5.south)!.7!(w6.south)$) {$R_1$};
            \end{scope}
            \begin{scope}[xshift=12cm,yscale=.7]
                \node[fill] (z) at (2,-.5) {};
                \node[fill] (b) at (0,0) {};
                \node[fill] (a) at ($(b)+(90:1)$) {};
                \node[fill] (x) at ($(a)+(90:1)$) {};
                \node[fill] (f) at ($(b)+(160:1)$) {};
                \node[fill] (e) at ($(f)+(90:.65)$) {};
                \node[fill] (s) at ($(e)+(90:.65)$) {};
                \node[fill] (c) at (4,0) {};
                \node[fill] (d) at ($(c)+(90:1)$) {};
                \node[fill] (y) at ($(d)+(90:1)$) {};
                \node[fill] (h) at ($(c)+(30:1)$) {};
                \node[fill] (g) at ($(h)+(90:.65)$) {};
                \node[fill] (t) at ($(g)+(90:.65)$) {};
                \foreach \w in {1,...,9}{
                    \node[fill] (w\w) at ($(x)+(\w/2.5,0)$) {};
                }
                \draw[-,ultra thick] (s) -- (e) --
                    (f)  edge[bend right=40] (w3) (w3)--
                    (w1) edge[bend left=50] (w6) (w6) --
                    (w4) edge[bend left=50] (w9) (w9) --
                    (w7) edge[bend right=40] (h) (h)--
                    (g) -- (t);
                \draw[-] (b) -- (f) (c) -- (h) (x) -- (y);
                \draw[-,ultra thick] (x) -- (b) -- (z) -- (c) -- (y);
                \node[hide,anchor=north] at ($(z.south)!.5!(b.south)$) {$P$};
                \node[hide,anchor=north] at ($(w5.south)!.7!(w6.south)$) {$R_2$};
            \end{scope}
        \end{tikzpicture}
        \caption{
            A raw-cover-$z$-gadget \cite{Fomin2010} and its two partial solutions for \textsc{Hamiltonian Cycle} problem.
        }%
        \label{fig:raw-edge-gadget}
    \end{figure}
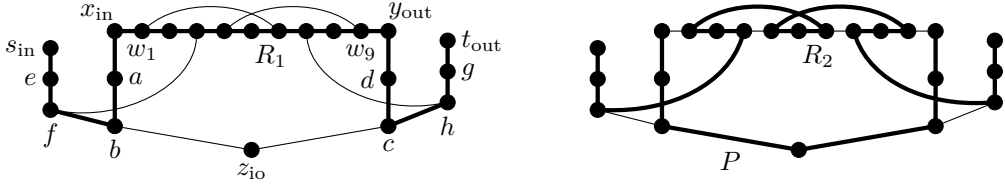

    For a fixed vertex $z$, we create \emph{cover-$z$-gadget} by first creating a raw-cover-$z$-gadget,
    then changing the following edges into arcs $(s_{\rm in},e)$, $(e,f)$, $(f,b)$, $(b,z)$, $(z,c)$, $(c,h)$, $(h,g)$, $(g,t_{\rm out})$, $(x_{\rm in}, w_1)$, $(w_9,y_{\rm out})$, $(w_1,w_6)$, $(w_4,w_9)$, $(w_3,w_4)$, and $(w_6,w_7)$,
    then subdividing every edge and arc with two non-terminal vertices,
    and last replacing every $\{u,v\}$ edge with arcs $uv$ and $vu$,
    see \Cref{fig:edge-gadget-adapted}.
    Let the paths $P$, $R_1$, and $R_2$ be directed paths in the cover-$z$-gadget that were created by subdivision and orientation (in the direction of order of vertices from their definition) of their analogues in the raw-cover-$z$-gadget.

    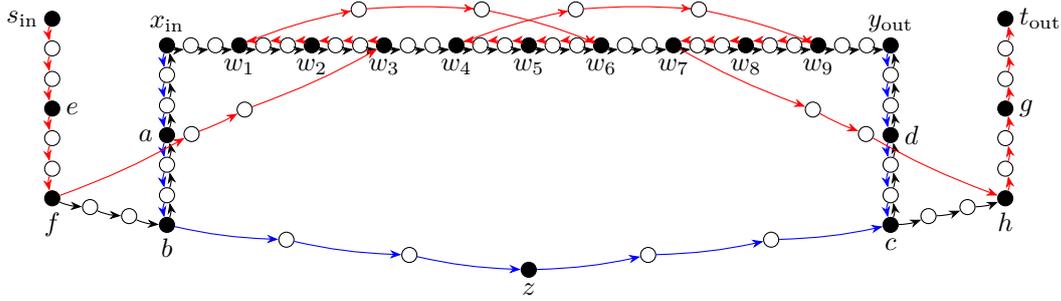
\begin{figure}[tbh]
        \centering
        \begin{tikzpicture}[def,scale=0.85]
            \newcommand{\customarc}[5]{
                \node[fill=white] (t1) at ($(#1)!#5!(#2)$) {};
                \node[fill=white] (t2) at ($(#1)!1-#5!(#2)$) {};
                \draw[-{Stealth[scale=#3]}] (#1) edge[bend right=#4] (t1);
                \draw[-{Stealth[scale=#3]}] (t1) edge[bend right=#4] (t2);
                \draw[-{Stealth[scale=#3]}] (t2) edge[bend right=#4] (#2);
            }
            \newcommand{\customarccolor}[6]{
                \node[fill=white] (t1) at ($(#1)!#5!(#2)$) {};
                \node[fill=white] (t2) at ($(#1)!1-#5!(#2)$) {};
                \draw[-{Stealth[scale=#3]},#6] (#1) edge[bend right=#4] (t1);
                \draw[-{Stealth[scale=#3]},#6] (t1) edge[bend right=#4] (t2);
                \draw[-{Stealth[scale=#3]},#6] (t2) edge[bend right=#4] (#2);
            }
            \newcommand{\customedgecolor}[7]{
                \node[fill=white] (t1) at ($(#1)!#5!(#2)$) {};
                \node[fill=white] (t2) at ($(#1)!1-#5!(#2)$) {};
                \draw[-{Stealth[scale=#3]},#6] (#1) edge[bend right=#4] (t1);
                \draw[-{Stealth[scale=#3]},#6] (t1) edge[bend right=#4] (t2);
                \draw[-{Stealth[scale=#3]},#6] (t2) edge[bend right=#4] (#2);
                \draw[-{Stealth[scale=#3]},#7] (t1) edge[bend right=#4] (#1);
                \draw[-{Stealth[scale=#3]},#7] (t2) edge[bend right=#4] (t1);
                \draw[-{Stealth[scale=#3]},#7] (#2) edge[bend right=#4] (t2);
            }
            \begin{scope}[xshift=4cm,yscale=1.4,xscale=2.8]
                \node[fill] (z) at (2,-.5) {};
                \node[fill] (b) at (0,0) {};
                \node[fill] (a) at ($(b)+(90:1)$) {};
                \node[fill] (x) at ($(a)+(90:1)$) {};
                \node[fill] (f) at ($(b)+(155:.7)$) {};
                \node[fill] (e) at ($(f)+(90:1)$) {};
                \node[fill] (s) at ($(e)+(90:1)$) {};
                \node[fill] (c) at (4,0) {};
                \node[fill] (d) at ($(c)+(90:1)$) {};
                \node[fill] (y) at ($(d)+(90:1)$) {};
                \node[fill] (h) at ($(c)+(25:.7)$) {};
                \node[fill] (g) at ($(h)+(90:1)$) {};
                \node[fill] (t) at ($(g)+(90:1)$) {};
                \customarccolor{g}{t}{0.9}{10}{.33}{red}
                \customarccolor{h}{g}{0.9}{10}{.33}{red}
                \customarc{c}{h}{0.9}{20}{.33}
                \customedgecolor{d}{c}{0.9}{10}{.33}{blue}{black}
                \customedgecolor{y}{d}{0.9}{10}{.33}{blue}{black}
                \customedgecolor{a}{x}{0.9}{10}{.33}{black}{blue}
                \customedgecolor{b}{a}{0.9}{10}{.33}{black}{blue}
                \customarc{f}{b}{0.9}{20}{.33}
                \customarccolor{e}{f}{0.9}{10}{.33}{red}
                \customarccolor{s}{e}{0.9}{10}{.33}{red}
                \foreach \w in {1,...,9}{
                    \node[fill] (w\w) at ($(x)+(\w/2.5,0)$) {};
                }
                \customarc{w3}{w4}{0.9}{50}{.33}
                \customarc{w6}{w7}{0.9}{50}{.33}
                \foreach \w in {2,3,5,6,8,9}{
                    \pgfmathtruncatemacro\ww{\w-1}
                    \customedgecolor{w\ww}{w\w}{0.9}{50}{.33}{black}{red}
                }
                \customarc{x}{w1}{0.9}{50}{.33}
                \customarc{w9}{y}{0.9}{50}{.33}
                \customarccolor{w7}{h}{0.9}{6}{.42}{red}
                \customarccolor{f}{w3}{0.9}{6}{.42}{red}
                \customarccolor{z}{c}{0.9}{10}{.33}{blue}
                \customarccolor{b}{z}{0.9}{10}{.33}{blue}
                \node[fill=white] (t1) at ($(w1)!.33!(w6)+(0,0.4)$) {};
                \node[fill=white] (t2) at ($(w1)!.67!(w6)+(0,0.4)$) {};
                \draw (w1) edge[bend left=8, red] (t1);
                \draw (t1) edge[bend left=8, red] (t2);
                \draw (t2) edge[bend left=8, red] (w6);
                \node[fill=white] (t1) at ($(w4)!.33!(w9)+(0,0.4)$) {};
                \node[fill=white] (t2) at ($(w4)!.67!(w9)+(0,0.4)$) {};
                \draw (w4) edge[bend left=8, red] (t1);
                \draw (t1) edge[bend left=8, red] (t2);
                \draw (t2) edge[bend left=8, red] (w9);
                \node[hide,anchor=north] at (z.south) {$z$};
                \node[hide,anchor=north] at (b.south) {$b$};
                \node[hide,anchor=north] at (f.south) {$f$};
                \node[hide,anchor=north] at (c.south) {$c$};
                \node[hide,anchor=north] at (h.south) {$h$};
                \node[hide,anchor=west] at (g.east) {$g$};
                \node[hide,anchor=west] at (t.east) {$t_{\rm out}$};
                \node[hide,anchor=east] at (a.west) {$a$};
                \node[hide,anchor=east] at (s.west) {$s_{\rm in}$};
                \node[hide,anchor=west] at (e.east) {$e$};
                \node[hide,anchor=west] at (d.east) {$d$};
                \node[hide,anchor=south] at (y.north) {$y_{\rm out}$};
                \node[hide,anchor=south] at (x.north) {$x_{\rm in}$};
                \foreach \w in {1,...,9}{
                    \node[hide,below] at (w\w.south) {$w_\w$};
                }
            \end{scope}
        \end{tikzpicture}
        \caption{
            Cover-$z$-gadget is created via selective orientation and subdivision of the raw-cover-$z$-gadget.
        }%
        \label{fig:edge-gadget-adapted}
    \end{figure}

    \begin{claim}\label{claim:ltwogadget_variant_full}
        Let $G$ be a graph that contains the cover-$z$-gadget $X$ on $G[X]$, i.e., $G[X]$ is isomorphic to the cover-$z$-gadget
        and all arcs in $E(G) \setminus E(G[X])$ that end in $X$ are incident with the vertices $x^X_{\rm in}$, $s^X_{\rm in}$, or $z^X_{\rm io}$ and the arcs in $E(G) \setminus E(G[X])$ that start in $X$ are incident with the vertices $y^X_{\rm out}$, $t^X_{\rm out}$, or $z^X_{\rm io}$ in $X$.
        If $S$ is a minimal solution to \DWRPshort that traverses every terminal vertex of $X$ exactly once, then the solution includes either the path $R_1$, or two paths $P$ and $R_2$ as continuous subpaths.
    \end{claim}
    \begin{claimproof}
        Let $S$ be any minimal solution that traverses the terminal vertices of the cover-$z$-gadget $X$ exactly once.
        Note that once $S$ enters a non-terminal vertex that was created by subdividing edges, $S$ has to follow through to the other terminal, as the initial terminal vertex cannot be visited again.
        Terminal vertices $e$, $a$, $w_2$, $w_5$, $w_8$, $d$, $g$ have only two terminal ``neighbors''.
        Hence, they must be visited immediately after their first terminal neighbor is visited.
        Solution traverses from $s_{\rm in}$ through $e$ to $f$ where it has two choices.
        If $S$ continues from $f$ to $b$, then it is forced through $a$ to $x_{\rm in}$, by the above argument, and then
        the only way to continue is to $w_1$ and through $w_2$ to $w_3$.
        After that the only option is $w_4$, forced to visit $w_5$ continues to $w_6$ and $w_7$.
        Then $S$ must visit $w_8$, continues to $w_9$ and $y_{\rm out}$.
        Again, $S$ must visit $d$ and continues to $c,h,g,t_{\rm out}$ as no other options remain.
        In the other case, where $S$ initially goes from $s_{\rm in}$ through $e$ to $f$ and continues to $w_3$ one observes that then it is forced through $w_2$ to $w_1$ where no other option remains than to go to $w_6$, by similar logic we have $w_5,w_4,w_9,w_8,w_7,h$ after which, again, the only option is to continue through $g$ to $t_{\rm out}$.
        This traversal leaves out $x_{\rm in},a,b,z,c,d,y_{\rm out}$ where the only way to visit these terminals by $S$ is exactly in this order.
    \end{claimproof}

    Now we present the construction, see \Cref{fig:whard-construction}; recall that we are reducing from an instance $(G,c,k)$ of CDS and we aim to create an instance $(G',\WP,\wFn,\cFn,\budget)$ of \DWRPshort.
    Let $\mathcal S$ be a force-$(k+1)$-traversals-gadget that we call \emph{selection-gadget} and let $\mathcal A$ be a force-$(2|E(G)|+|V(G)|+1)$-traversals-gadget.
    We create an auxiliary non-terminal vertex named $x$ and add an arc $(x,u^{\mathcal A}_{\rm in})$.
    Let $x$ together with $\mathcal A$ be called \emph{auxiliary-gadget}.
    We mutually connect the selection-gadget and the auxiliary-gadget by creating terminal vertices $x_1,x_2$, non-terminal vertex $x_3$, and arcs $(u^{\mathcal A}_{\rm out},x_2)$, $(x_2,x_3)$, $(x_3,u^{\mathcal S}_{\rm in})$, $(u^{\mathcal S}_{\rm out}, x_1)$, $(x_1,x)$.

    For every vertex $u \in V(G)$ we add vertices $z_u$, $c_u$, and $d_u$ to $V(G')$, then we create a cover-$z_u$-gadget denoted $\mathcal E^u$.
    We add arcs $(u^{\mathcal S}_{\rm out}, x^{\mathcal E^u}_{\rm in})$, $(y^{\mathcal E^u}_{\rm out}, c_u)$, $(c_u,d_u)$, $(c_u,u^{\mathcal S}_{\rm in})$, $(u^{\mathcal A}_{\rm out}, s^{\mathcal E^u}_{\rm in})$, $(t^{\mathcal E^u}_{\rm out}, x)$ to $E(G')$, and we set the capacity of $(c_u,d_u)$ to $c(u)$ (the capacity of $u$).
    For every arc $(u,v)$ such that $\{u,v\} \in E(G)$ we create a cover-$z_v$-gadget denoted $\mathcal E^{u,v}$.
    Then we add arcs $(d_u,x^{\mathcal E^{u,v}}_{\rm in})$, $(y^{\mathcal E^{u,v}}_{\rm out}, c_u)$, $(u^{\mathcal A}_{\rm out}, s^{\mathcal E^{u,v}}_{\rm in})$, $(t^{\mathcal E^{u,v}}_{\rm out}, x)$.
    Note that in the above we created two gadgets, $\mathcal E^{u,v}$ and $\mathcal E^{v,u}$, for each $\{u,v\} \in E(G)$.
    Finally, we set the budget $\budget=132|E(G)|+69|V(G)|+3k+12$ which concludes construction of the \DWRPshort instance.

    \begin{figure}[tbh]
       \centering
        \newcommand{\convexpath}[2]{
	[
	create hullnodes/.code={
		\global\edef\namelist{#1}
		\foreach [count=\counter] \nodename in \namelist {
			\global\edef\numberofnodes{\counter}
			\node at (\nodename) [draw=none,name=hullnode\counter] {};
		}
		\node at (hullnode\numberofnodes) [name=hullnode0,draw=none] {};
		\pgfmathtruncatemacro\lastnumber{\numberofnodes+1}
		\node at (hullnode1) [name=hullnode\lastnumber,draw=none] {};
	},
	create hullnodes
	]
	($(hullnode1)!#2!-90:(hullnode0)$)
	\foreach [
	evaluate=\currentnode as \previousnode using \currentnode-1,
	evaluate=\currentnode as \nextnode using \currentnode+1
	] \currentnode in {1,...,\numberofnodes} {
		let
		\p1 = ($(hullnode\currentnode)!#2!-90:(hullnode\previousnode)$),
		\p2 = ($(hullnode\currentnode)!#2!90:(hullnode\nextnode)$),
		\p3 = ($(\p1) - (hullnode\currentnode)$),
		\n1 = {atan2(\y3,\x3)},
		\p4 = ($(\p2) - (hullnode\currentnode)$),
		\n2 = {atan2(\y4,\x4)},
		\n{delta} = {-Mod(\n1-\n2,360)}
		in
		{-- (\p1) arc[start angle=\n1, delta angle=\n{delta}, radius=#2] -- (\p2)}
	}
	-- cycle
}

\begin{tikzpicture}[def]
    \newcommand{\traversalgadget}[2]{
        \def\name{#1}
        \def\traversals{#2}
        \node[solid] (uin_\name) at (0,1) {};
        \node[solid] (w) at (2,1) {};
        \node[solid] (uout_\name) at (3,1) {};
        \foreach \y in {1,...,\traversals}{
            \node[solid,fill] (v\y) at (1,2*\y/\traversals-1/\traversals) {};
            \draw (uin_\name) -- (v\y);
            \draw (v\y) -- (w);
        }
        \draw (w) --node[hide,above]{$\traversals$} (uout_\name);
    }
    \traversalgadget{selection}{3}
    \node[hide] at (-.2,1.4) {$u^\mathcal S_{\rm in}$};
    \node[hide] at (3,1.4) {$u^\mathcal S_{\rm out}$};
    \begin{pgfonlayer}{bbg}
     \fill[black!8, rounded corners] (-1.2,0) rectangle (3.4,2);
    \end{pgfonlayer}
    \node[hide] at (-0.8,0.6) {$\mathcal S$};
    \begin{scope}[xshift=6cm,densely dotted]
        \node[solid,label=above:$x$] (xx) at (-1,1) {};
        \traversalgadget{auxiliary}{6}
        \node[hide] at (0,1.4) {$u^\mathcal A_{\rm in}$};
        \node[hide] at (3.2,1.4) {$u^\mathcal A_{\rm out}$};
    \end{scope}
    \begin{pgfonlayer}{bbg}
     \fill[black!8, rounded corners] (4.5,0) rectangle (10,2);
    \end{pgfonlayer}
    \node[hide] at (9.6,0.6) {$\mathcal A$};
    \node[fill,label=right:$x_2$] (x2) at (8.2,2.5) {};
    \node[    ,label=left: $x_3$] (x3) at (0.8,2.5) {};
    \node[fill,label=above:$x_1$] (x1) at (4,1.4) {};
    \draw (uout_selection) -- (x1);
    \draw (x1) -- (xx);
    \draw[densely dotted] (xx) -- (uin_auxiliary);
    \draw (uout_auxiliary) -- (x2);
    \draw (x2) -- (x3);
    \draw (x3) -- (uin_selection);
    \newcommand{\connection}[4]{
        \pgfmathsetmacro\hangle{#3+20}
        \pgfmathsetmacro\langle{#3-20}
        \node[fill] (s#1) at ($(#2)+(\hangle:1.5*#4)$) {};
        \node[fill] (t#1) at ($(#2)+(\langle:1.5*#4)$) {};
        \node[fill] (x#1) at ($(s#1)+(#3:1.*#4)$) {};
        \node[fill] (y#1) at ($(t#1)+(#3:1*#4)$) {};
        \begin{scope}[opacity=.15,transparency group]
            \fill[black] \convexpath{s#1,x#1,y#1,t#1,#2}{5pt};
        \end{scope}
        \begin{pgfonlayer}{bg}    
            \draw[densely dotted] (uout_auxiliary) to (s#1);
            \draw[densely dotted] (t#1) to (xx);
        \end{pgfonlayer}
    }
    \newcommand{\vertexgadget}[1]{
        \node[fill,label=below:$z_#1$] (z#1) at (2,0) {};
        \connection{#1}{z#1}{90}{1.0}
        \node (c#1) at ($(y#1)+(5:1)$) {};
        \node (d#1) at ($(c#1)+(-90:.8)$) {};
        \draw (uout_selection) to[out=-90,in=110] (x#1);
        \draw (y#1) -- (c#1);
        \draw (c#1) -- (d#1);
        \draw (c#1) to[out=140,in=-90] (uin_selection);
    }
    \begin{scope}[xshift=-2cm,yshift=-4cm]
        \vertexgadget{u}
        \node[hide,label=left:$x_{\rm in}$] at (xu){};
        \node[hide,label=left:$y_{\rm out}$] at (yu){};
        \node[hide,label=left:$s_{\rm in}$] at (su){};
        \node[hide,label=left:$t_{\rm out}$] at (tu){};
        \node[hide,label=above right:$c_u$] at (cu){};
        \node[hide,label=above right:$d_u$] at (du){};
        \node[hide] at ($(zu)+(90:.7)$){$\mathcal E^u$};
    \end{scope}
    \begin{scope}[xshift=3cm,yshift=-4cm]
        \vertexgadget{v}
        \node[hide] at ($(zv)+(90:1)$){$\mathcal E^v$};
    \end{scope}
    \begin{scope}[xshift=6cm,yshift=-4cm]
        \vertexgadget{w}
        \node[hide] at ($(zw)+(90:1)$){$\mathcal E^w$};
    \end{scope}
    \connection{vu}{zu}{22}{.7}
    \connection{uv}{zv}{158}{.7}
    \node[hide] at ($(zu)+(25:1.3)$){$\mathcal E^{vu}$};
    \node[hide] at ($(zv)+(155:1.3)$){$\mathcal E^{uv}$};
    \draw (du) -- (xuv);
    \draw (yuv) -- (cu);
    \draw (dv) -- (xvu);
    \draw (yvu) -- (cv);
\end{tikzpicture}
        \caption{An example of a complete constructed \DWRPshort instance for graph $G=(\{u,v,w\}, \{\{u,v\}\})$, $k=2$, (capacities omitted). Top half contains the selection- and auxiliary-gadget along with their connections through $x_1,x_2,x_3$. Bottom half corresponds to the graph $G$. Gray boxes hide inner parts of the cover-gadgets, see \cref{fig:edge-gadget-adapted}. Dotted arrows signify arcs that make-up the auxiliary-gadget or connect cover-gadgets to the auxiliary-gadget.}%
        \label{fig:whard-construction}
    \end{figure}
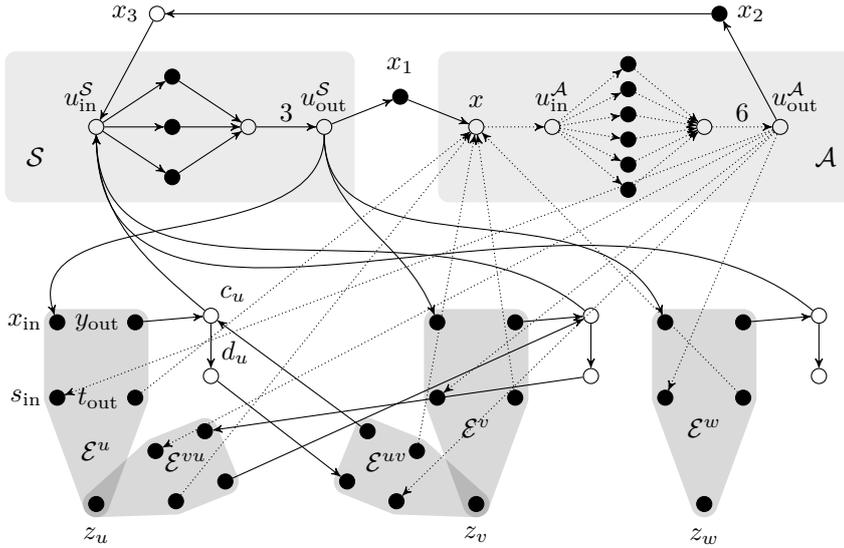

    \begin{claim}\label{claim:treedepth_forwards}
        If $(G,c,k)$ is a yes-instance of CDS, then the instance $(G',\WP,\wFn,\cFn,\budget)$ of \DWRP with budget $\budget = 132|E(G)|+69|V(G)|+3k+12$ has a solution.
    \end{claim}
    \begin{claimproof}
        Let $S$ ($|S| \le k$) be the set of vertices in the solution of the instance $(G,c,k)$ of CDS and $f \colon (V(G) \setminus S) \to S$ be the domination mapping.
        Solution to \textsc{Directed Waypoint Routing} is presented as two vertex-disjoint walks, one starting in the selection-gadget and one in the auxiliary-gadget.
        These two walks are then joined via the auxiliary vertices $x_1$, $x_2$, and $x_3$.

        First walk:
        For each $s \in S$ we traverse the selection-gadget once, covering one of its terminals, and then through the $P$ path of $\mathcal E^s$, covering $z_s$, then to $c_s$.
        For each $x \in f^{-1}(s)$ the walk continues from $c_s$ to $d_s$, through the $P$ path in $\mathcal{E}^{s,x}$ covering $z_x$, returning back to $c_s$.
        Finally, the walk jumps back to the selection-gadget.

        Second walk:
        Let $E$ be a cover-$z$-gadget within $G'$.
        The second walk goes through the auxiliary-gadget while covering one of its terminals, then it continues through a path of $E$.
        More precisely, if the path $P$ of $E$ was traversed by the first walk then the second walk goes through $R_2$, otherwise it traverses $R_1$.
        Last, the walk goes to $x$, returning back to the auxiliary-gadget.

        These two walks are joined into a single walk by adding cycle $u^{\mathcal S}_{\rm out}$, $x_1$, $x$, $u^{\mathcal A}_{\rm in}$, $v^{\mathcal A}_{\rm i}$, $w^{\mathcal A}$, $u^{\mathcal A}_{\rm out}$, $x_2$, $x_3$, $u^{\mathcal S}_{\rm in}$, $v^{\mathcal S}_{\rm j}$, $w^{\mathcal S}$, $u^{\mathcal S}_{\rm out}$ for some $i$ and $j$ so that uncovered terminals of the selection-gadget and the auxiliary-gadget are traversed by this cycle.

        To compute the walk weight, note that every terminal is traversed exactly once and that distance of every pair of consecutive terminals along the walk is 3, hence, we used the budget of $3 \cdot T$ where $T$ is the number of terminals.
        There are $k+1$ terminals in the selection gadget, $2|E(G)|+|V(G)|+1$ terminals in the auxiliary-gadget, terminals $x_1$ and $x_2$, then in every cover-$z$-gadget we have 21 terminals (note that $z$ is not counted here), and last we have terminals $z_u$ for every $u \in V(G)$.
        In total this is $3 \cdot (k+1+2|E(G)|+|V(G)|+1+2+21\cdot(2|E(G)|+|V(G)|)+|V(G)|)$ which simplifies to $132|E(G)|+69|V(G)|+3k+12$ which is equal to the budget $\budget$.
    \end{claimproof}

    \begin{claim}\label{claim:treedepth_backwards}
        If \DWRPshort instance $(G',\WP,\wFn,\cFn,\budget)$ with budget $\budget = 132|E(G)|+69|V(G)|+3k+12$ is a yes-instance, then the instance $(G,c,k)$ of CDS is a yes-instance.
    \end{claim}
    \begin{claimproof}
        Let us identify partial walks that must be present in any solution and argue that they compose exactly into a walk shown in the proof of \Cref{claim:treedepth_forwards}.
        Observe that the budget is exactly so that every terminal can be assigned budget of 3, and as each pair of terminals is at distance at least 3 we conclude that the solution traverses every terminal exactly once.

        First, by \Cref{claim:traversal_gadget_full} the selection-gadget is traversed exactly $k+1$ times and the auxiliary-gadget is traversed $2|E(G)|+|V(G)|+1$ times.
        As $x_1$ and $x_2$ need to be covered the paths $u^{\mathcal S}_{\rm out}$, $x_1$, $x$, $u^{\mathcal A}_{\rm in}$ and $u^{\mathcal A}_{\rm out}$, $x_2$, $x_3$, $u^{\mathcal S}_{\rm in}$ need to be traversed exactly once.
        These together with one traversal of the selection-gadget and one traversal of the auxiliary-gadget give us the cycle that joined the first and second walk.

        Next, from the vertex $u^{\mathcal S}_{\rm out}$ the remaining $k$ walks must continue, and as all its out-neighbors are terminals the solution continues into $k$ disjoint out-neighbors.
        These out-neighbors belong to some cover-$z_u$-gadgets for a set $u \in S$, consequently $|S|=k$.
        More precisely, these out-neighbors are $x^{\mathcal E^u}_{\rm in}$ for $u \in S$.
        Let $\mathcal E^S= \{ \mathcal E^u \mid u \in S \}$ and similarly for any vertex $w$ let $w^{\mathcal E^S} = \{w^{\mathcal E^u} \mid u \in S\}$.
        As vertices $x^{\mathcal E^S}$ are beginnings of the respective $P$ paths within $\mathcal E^S$ we know by \Cref{claim:ltwogadget_variant_full} that $\mathcal E^S$ are traversed by their $P$ and $R_2$ paths, covering $z_u$ for each $u \in S$ in the process.
        After the traversal of $P$ the walks continue to $c_u$ for each $u \in S$.
        From there each walk can either return back to the selection-gadget or go to $d_u$, but as $(c_u,d_u)$ has capacity $c(u)$ such traversal can happen at most $c(u)$ many times.
        From $d_u$ the walk continues through $\mathcal E^{u,w}$ for any $w \in N(u)$, covering $z_w$ in the process, and returning to $c_u$.

        For every other cover-gadget $\mathcal E^v$ where $v \in V(G) \setminus S$ we know the walks do not enter the gadgets from $u^{\mathcal S}_{\rm out}$ (the only in-neighbor of $x^{\mathcal E^v}_{\rm in}$) so these must be traversed by their $R_1$ paths.
        All of the $R_1$ and $R_2$ traversals of the cover-gadgets can only continue to the auxiliary gadget, and vice-versa, all walks from the auxiliary gadget (except for the connecting cycle) continue only continue to the cover-gadgets paths $R_1$ or $R_2$.

        This forced structure of the solution is identical to the one described in \Cref{claim:treedepth_forwards} and so there is a straight-forward translation back to the capacitated dominating set, namely, the set $S$ is a capacitated dominating set of $G$.
    \end{claimproof}

    \begin{claim}\label{claim:tdtd}
        If $G$ has modulator of size $k$ to treedepth $td$, then $G'$ has modulator of size $7\cdot k + 7$ to treedepth $7\cdot td+72$.
    \end{claim}
    \begin{claimproof}
        Recall that $G$ having treedepth bounded by $td$ is witnessed by an \emph{elimination forest} which is a family of rooted trees of height at most $td$ where the trees are over vertices $V(G)$ but their edges are unrelated to $E(G)$; however, they satisfy a single property -- for every edge $uv \in E(G)$ the vertices $u$ and $v$ are in the forest in an ancestor-descendant relationship; we call such edges \emph{satisfied}.
        For example, if we place a set of vertices $V_1$ on path rooted in one of its endpoints, then all edges in $V_1 \times V_1$ have their endpoints in an ancestor-descendant relationship, however, the height of such a tree is $|V_1|$.

        To show the claim we will proceed as follows.
        We take the set of stars created by removing the modulator $M$ from $G$, and create an elimination forest where $M$ forms a path from the root, and the stars are attached below -- this naturally satisfies all the edges of $G$.
        Then, we replace each vertex of the graph with a constant number of vertices of $G'$, placing the new set on a path in-place of the original vertices, and arriving at an elimination forest of $G'$.
        As the sets are paths, they satisfy all the edges within each path separately; then we argue that all edges of $G'$ are still in the ancestor-descendant relationship.
        We then observe that the vertices $M'$ that replaced the vertices of $M$ form a new path from the root and contains $7\cdot k+7$ vertices; and whose removal creates an elimination forest of constant depth.

        Consider $M$ to be the modulator of size $k$ and $\mathcal F_o$ to an elimination forest that consists of trees $T_1,\dots,T_\ell$ and witnesses that the treedepth of $G \setminus M$ is at most $td$.
        Let $\mathcal F$ be an elimination forest that contains a single tree with $M$ being its path from the root and with trees $T_1,\dots,T_\ell$ of $\mathcal F_o$ attached below $M$.
        Observe that the height of $\mathcal F$ is $k+td$.
        We aim to create an elimination forest $\mathcal F'$ of $G'$ such that when a path from the root $M'$ of small size is removed, we get an elimination forest $\mathcal F'_o$ of the desired size.

        The elimination forest $\mathcal F'$ consists of a single tree $T'$ (see~\Cref{fig:treedepth-hard-overview}).
        The tree $T'$ contains a path $R$ from the root containing vertices $x$, $u^{\mathcal A}_{\rm in}$, $u^{\mathcal A}_{\rm out}$, $w^{\mathcal A}$, $u^{\mathcal S}_{\rm in}$, $u^{\mathcal S}_{\rm out}$, $w^{\mathcal S}$, i.e., all non-terminal vertices of the selection-gadget and the auxiliary-gadget.
        Below $R$ we attach a path on vertices $x_1$, $x_2$, and $x_3$.
        We also attach below $R$ (independently) each terminal vertex of the selection-gadget and the auxiliary-gadget.
        Next, below $R$ we attach the tree $T_i \in \mathcal F$ while replacing each vertex $u \in T_i$ (where $u \in V(G)$) by a path $P_u$ containing vertices $z_u$, $c_u$, $d_u$, $s^{\mathcal E^u}_{\rm in}$, $x^{\mathcal E^u}_{\rm in}$, $y^{\mathcal E^u}_{\rm out}$, $t^{\mathcal E^u}_{\rm out}$.
        Below each $P_u$ we attach all vertices of $\mathcal E^u$ that are not in the decomposition yet.
        Recall that the elimination forest $\mathcal F$ has the property that for every edge $\{u,v\} \in E(G)$ either $u$ is the ancestor of $v$ or $v$ is the ancestor of $u$ in $\mathcal F$.
        So for every edge $\{u,v\} \in E(G)$ we can find a leaf $\ell_{u,v}$ of $\mathcal F$ such that the path from $\ell_{u,v}$ to the root contains both $u$ and $v$.
        As the last step, for every edge $\{u,v\} \in E(G)$ we attach to the bottommost vertex of $P_{\ell_{u,v}}$ a path containing all vertices of $\mathcal E^{u,v}$.

        To show that $\mathcal F'$ is an elimination forest we need to show that all the edges of $G'$ are in an ancestor-descendant relationship within $\mathcal F'$ -- let us call such edges \emph{resolved}.
        Note that if vertices are attached as a path in $\mathcal F'$ then all edges among these vertices are resolved.
        All edges incident to the 7 vertices in $R$ are resolved as these vertices are placed as ancestors of all the other vertices of $G'$.
        Placing $x_1,x_2,x_3$ below $R$ resolved all edges incident to them because they are adjacent only to $R$.
        Attaching all terminal vertices of the selection-gadget and the auxiliary gadget below $R$ resolves all their incident edges as they are also adjacent only to $R$.
        For any vertex $u \in V(G)$ the path $P_u$ contains all vertices that have edges connecting to $G \setminus (V(\mathcal E^u) \cup R)$.
        By attaching $\mathcal E^u$ as a path below $P_u$ we resolved all edges between two vertices of $P_u \cup \mathcal E^u$.
        Last, all the edges within or incident to $\mathcal E^{u,v}$ for any ${u,v} \in E(G)$ are resolved by placing the vertices of $\mathcal E^{u,v}$ below $P_{\ell_{u,v}}$ as $\mathcal E^{u,v}$ has edges only to $c_u$, $d_u$, and $z_v$ which are all ancestors of $\ell_{u,v}$ in $\mathcal F'$.

        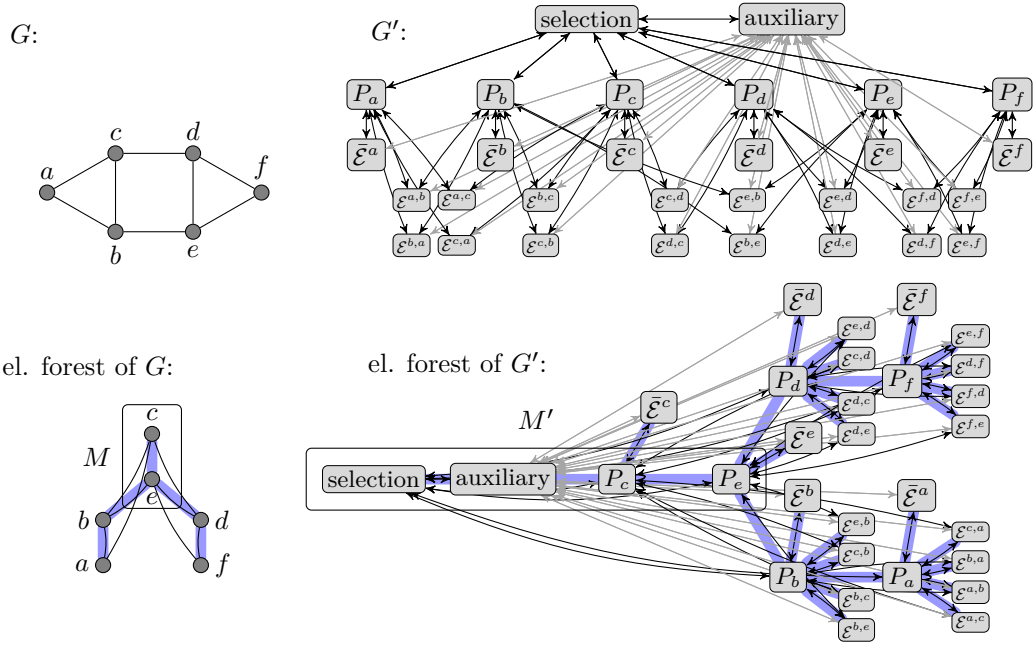
\begin{figure}[tbh]
            \centering
            \begin{tikzpicture}
    \begin{scope}[shift={(-6,2.2)},scale=.6,every node/.style={draw,fill=gray,circle,inner sep=2pt},every label/.style={rectangle,hide}]
        \node[hide] at (-3,3.5) {$G$:};
        \begin{scope}[shift={(-1.5,0)}]
            \node[label={$a$}] (a) at (180:1) {};
            \node[label=-90:{$b$}] (b) at (300:1) {};
            \node[label={$c$}] (c) at (420:1) {};
        \end{scope}
        \begin{scope}[shift={(+1.2,0)}]
            \node[label={$d$}] (d) at (120:1) {};
            \node[label=-90:{$e$}] (e) at (240:1) {};
            \node[label={$f$}] (f) at (360:1) {};
        \end{scope}
        \foreach \f/\t in {a/b,b/c,a/c,b/e,c/d,e/d,d/f,e/f}{
            \draw (\f) -- (\t);
        }
    \end{scope}
    \begin{scope}[shift={(-6,-2.2)},scale=.6,every node/.style={draw,fill=gray,circle,inner sep=2pt},every label/.style={rectangle,hide}]
        \node[hide] at (-1.6,3.5) {el. forest of $G$:};
        \begin{scope}[shift={(-.2,2)}]
            \node[label=90:{$c$}] (c) at (0,0) {};
            \node[label=-90:{$e$}] (e) at ($(c)+(270:1)$) {};
            \node[label=180:{$b$}] (b) at ($(e)+(220:1.4)$) {};
            \node[label=180:{$a$}] (a) at ($(b)+(270:1)$) {};
            \node[label=0:{$d$}] (d) at ($(e)+(320:1.4)$) {};
            \node[label=0:{$f$}] (f) at ($(d)+(270:1)$) {};
        \end{scope}
        \begin{pgfonlayer}{bg}
            \foreach \f/\t in {c/e,e/b,b/a,e/d,d/f}{
                \draw[line width=4pt,line cap=round,blue!40] (\f) -- (\t);
            }
        \end{pgfonlayer}
        \foreach \f/\t in {a/b,b/c,a/c,b/e,c/d,e/d,d/f,e/f}{
            \draw (\f) to[bend right=10] (\t);
        }
    \end{scope}
    \tikzstyle{cbox} = [fill=gray!30,draw,very thin,font=\small,rounded corners=2pt,inner sep=2pt,rectangle,opacity=1]
    \begin{scope}[def]
        \node[hide] at (-3,4.35) {$G'$:};
        \node[cbox] (selection) at (-0.4,4.5) {selection};
        \node[cbox] (auxiliary) at (2.3,4.5) {auxiliary};
        \draw[def,{stealth'}-{stealth'}] (selection) -- (auxiliary);
        \foreach[count=\i] \x in {a,b,c,d,e,f}{
            \node[cbox] (vertex\x) at (-5+1.7*\i,3.5) {$P_\x$};
            \node[cbox] (gadget\x) at (-5+1.7*\i,2.7) {$\bar{\mathcal E}^\x$};
            \draw[def] (selection) -- (vertex\x);
            \draw[def] (vertex\x) -- (selection);
            \begin{scope}[def,gray!70]
                \draw (auxiliary) to (gadget\x);
                \draw (gadget\x) to (auxiliary);
            \end{scope}
            \draw[def,{stealth'}-{stealth'}] (gadget\x) -- (vertex\x);
        }
        \foreach \f/\t in {b/a,c/b,c/a,d/c,b/e,d/e,d/f,e/f}{
            \node[cbox,scale=.7] (edge1) at ($(vertex\f)!.65!(vertex\t)+(0,-2.0)$) {$\mathcal E^{\f,\t}$};
            \node[cbox,scale=.7] (edge2) at ($(vertex\t)!.35!(vertex\f)+(0,-1.4)$) {$\mathcal E^{\t,\f}$};
            \begin{pgfonlayer}{bg}
                \begin{scope}[def,{stealth'}-{stealth'}]
                    \draw (vertex\f) edge[bend left=5] (edge1);
                    \draw (edge1) edge[bend left=5] (vertex\t);
                    \draw (vertex\t) edge[bend left=5] (edge2);
                    \draw (edge2) edge[bend left=5] (vertex\f);
                \end{scope}
                \begin{scope}[def,gray!70]
                    \draw (auxiliary) to (edge1);
                    \draw (edge1) to (auxiliary);
                    \draw (auxiliary) to (edge2);
                    \draw (edge2) to (auxiliary);
                \end{scope}
            \end{pgfonlayer}
        }
        \node[fit=(c)(e),draw,inner sep=8pt,rectangle,rounded corners=2pt,label={left:$M$}] {};
    \end{scope}
    \begin{scope}[def,shift={(0,-4.3)}]
        \node[hide] at (-2.1,4.25) {el. forest of $G'$:};
        \node[cbox] (selection) at (-3.2,2.7) {selection};
        \node[cbox] (auxiliary) at ($(selection)+(0:1.7)$) {auxiliary};
        \draw[def,{stealth'}-{stealth'}] (selection) -- (auxiliary);
        \foreach \x/\d/\p in {c/0/auxiliary,e/0/c,b/-60/e,a/0/b,d/60/e,f/0/d}{
            \node[cbox] (\x) at ($(\p)+(\d:1.5)$) {$P_\x$};
            \begin{pgfonlayer}{bg}
                \draw[def,{stealth'}-{stealth'}] (selection) to[bend right=10] (\x);
            \end{pgfonlayer}
        }
        \foreach \x/\d in {c/60,e/30,b/80,a/80,d/80,f/80}{
            \node[cbox] (gadget\x) at ($(\x)+(\d:1.1)$) {$\bar{\mathcal E}^\x$};
            \begin{pgfonlayer}{bg}
                \draw[def,{stealth'}-{stealth'}] (\x) to (gadget\x);
                \begin{scope}[def,gray!70]
                    \draw (auxiliary) to (gadget\x);
                    \draw (gadget\x) to (auxiliary);
                \end{scope}
            \end{pgfonlayer}
        }
        \begin{pgfonlayer}{bbg}
            \foreach \f/\t in {selection/auxiliary,auxiliary/c,c/e,e/d,d/f,e/b,b/a}{
                \draw[-,line width=4pt,line cap=round,blue!40] (\f) -- (\t);
            }
            \foreach \f in {c,e,d,f,b,a}{
                \draw[-,line width=4pt,line cap=round,blue!40] (\f) -- (gadget\f);
            }
        \end{pgfonlayer}
        \foreach \f/\t/\o in {b/a/2,c/b/3,c/a/6,c/d/3,e/b/7,e/d/7,d/f/2,e/f/6}{
            \node[cbox,scale=.7] (edge1) at ($(\t)+(.9,+\o*.1)$) {$\mathcal E^{\f,\t}$};
            \node[cbox,scale=.7] (edge2) at ($(\t)+(.9,-\o*.1)$) {$\mathcal E^{\t,\f}$};
            \begin{pgfonlayer}{bbg}
                \draw[-,line width=4pt,line cap=round,blue!40] (\t) -- (edge1);
                \draw[-,line width=4pt,line cap=round,blue!40] (\t) -- (edge2);
            \end{pgfonlayer}
            \begin{pgfonlayer}{bg}
                \begin{scope}[def,{stealth'}-{stealth'}]
                    \draw (\f) edge[bend left=5] (edge1);
                    \draw (edge1) edge[bend left=5] (\t);
                    \draw (\t) edge[bend left=5] (edge2);
                    \draw (edge2) edge[bend left=5] (\f);
                \end{scope}
                \begin{scope}[def,gray!70]
                    \draw (auxiliary) to (edge1);
                    \draw (edge1) to (auxiliary);
                    \draw (auxiliary) to (edge2);
                    \draw (edge2) to (auxiliary);
                \end{scope}
            \end{pgfonlayer}
        }
        \node[fit=(selection)(e),draw,inner sep=6pt,rounded corners=2pt,rectangle,label={above:$M'$}] {};
    \end{scope}
\end{tikzpicture}
            \caption{
                An example input graph $G$ with a respective result of the reduction $G'$ depicted as a set of simplified gadgets and paths of \Cref{claim:tdtd} together with their connections.
                We use $\bar{\mathcal E}^x$ for $\mathcal E^x \setminus P_x$.
                Bottom half: elimination forests of $G$ and $G'$; forest for $G'$ is rotated to preserve space.
            }%
            \label{fig:treedepth-hard-overview}
        \end{figure}

        Now, let $M'=R \cup \bigcup_{u \in M} P_u$ and let $\mathcal F'_o$ be $\mathcal F'$ with $M'$ removed.
        By definition $|M'|=|R|+\sum_{u \in M} |P_u|=7 + |M| \cdot 7$.
        As $M$ was a path from the root of the tree in $\mathcal F_o$ we observe that the construction replaced each vertex of $u \in M$ by $P_u$ which is a path, hence, their concatenation along with $R$ creates a path $M'$ from the root of $\mathcal F'_o$.
        Hence, removing $M'$ from $\mathcal F'$ results in a forest where each tree is rooted in its vertex that lied below $M'$ in $\mathcal F'$.
        We observe that the trees of $\mathcal F'_o$ consists of:
        \begin{itemize}
            \item $x_1,x_2,x_3$ of height 3,
            \item terminal vertices of the selection-gadget and the auxiliary-gadget of height 1,
            \item $\bar{\mathcal E}^u=\mathcal E^u \setminus P_u$ for each $u \in M$ of height $|\mathcal E^u|-|P_u|=72-7=65$,
            \item trees $T'_1,\dots,T'_{\ell}$ corresponding to $T_1,\dots,T_\ell$ in $\mathcal F_o$, height of which we argue below.
        \end{itemize}
        The trees $T_i$ have height at most $td$.
        Each vertex $u \in V(T_i)$ was replaced with $P_u$ to create $T'_i$, multiplying its height by 7.
        Moreover, for each $u \in V(T_i)$ we attached to $P_u$ vertices in $V(\mathcal E^u) \setminus P_u$, which is of size 65 (as computed above).
        Independently of that, for each edge $uv$ we attached entire $\mathcal E^{uv}$ below either $P_u$ or $P_v$, whichever was lower in the tree.
        As each $\mathcal E^{uv}$ is added independently, their addition increases the height of $T'_i$ at most once.
        As $|\mathcal E^{uv}|=72$ we conclude that each $T'_i$ has height at most $7\cdot td+72$.
    \end{claimproof}
    The above claim together with the fact that disjoint union of stars has treedepth 2 gives us the following corollary.
    \begin{corollary}\label{cor:consttreedepthmodulator}
        If $G$ has modulator of size $k$ to stars, then $G'$ has modulator of size $7\cdot k + 7$ to treedepth $86$.
    \end{corollary}

    As noted in \Cref{thm:cds_whard_stars} CDS is \Wh w.r.t.\ distance to stars.
    Assuming $G$ has a modulator to stars $M$ we have that $G'$ has a modulator of size at most $7\cdot |M|+7$ to treedepth $86$ by \Cref{cor:consttreedepthmodulator}.
    By \Cref{claim:treedepth_forwards,claim:treedepth_backwards} the constructed instance is equivalent to the initial one, therefore, \DWRP is \Wh w.r.t.\ distance to constant treedepth.
\end{proof}
}

\toappendix{
\section{Pathwidth Hardness for Unweighted Graphs}

Marx et al.~\cite{MarxSS16} showed the \Whness of \DTSPshort w.r.t.\ pathwidth by a chain of reductions.
In the last step, they ensured that some vertices are visited at most once by significantly increasing the weight of their incoming edges.
An obvious idea to show the \Whness also for unweighted graphs would be to subdivide the edges of large weight.
However, as this creates new vertices that might not be visited by the solution to the original instance, this only shows the \Whness for \DsTSPshort.
In this section we revisit the proof by Marx et al.~\cite{MarxSS16} with two goals.
First, we want to give precise bounds on the resulting pathwidth.
Second, we want to modify the last step in order to prove the \Whness for unweighted graphs also for \DTSPshort.

\begin{theorem}\label{thm:d_to_pw}
 \DTSPshort is \Wh w.r.t.
 \begin{enumerate}[a)]
 \item distance to pathwidth~2 and
 \item distance to pathwidth~3 in unweighted graphs. 
 \end{enumerate}
 Unless ETH fails, there is no algorithm for \DTSPshort with running time 
 \begin{enumerate}[i)]
  \item $n^{o(\sqrt{d_2})}$, where $d_2$ is the distance to pathwidth~2,
  \item $n^{o(d_3)}$, where $d_3$ is the distance to pathwidth~3,
  \item $n^{o(\sqrt{d_3})}$ in unweighted graphs, where $d_3$ is the distance to pathwidth~3, or
  \item $n^{o(d_4)}$ in unweighted graphs, where $d_4$ is the distance to pathwidth~4.
 \end{enumerate}
\end{theorem}

The chain of parameterized reductions given by Marx et al.~\cite{MarxSS16} starts with \textsc{Clique} and continues over \textsc{Multicolored Biclique}, through \textsc{Edge Balancing} and \textsc{Constrained Closed Walk} to \DTSPshort (see the figure \cite[page 46]{MarxSS16}).
While they do not explicitly state \Whness, it follows from the fact that \textsc{Clique} is \Wh and all the reductions are parameterized reductions.
To prove our theorem, we only revisit the last two steps of the chain, where we only analyze the structure obtained in the second to last step and only modify the last step. 
Let us introduce the problems.
\problemQuestion{\textsc{Edge Balancing}}
{A directed graph $D$ and a set $X_e$ of positive integers for every edge $e \in E(D)$.}
{Is it possible to select $\chi(e) \in X_e$ for every $e \in E(D)$ such that for every $v \in V(D)$ we have $\sum_{u \in N^{\text{in}}(v)} \chi((u,v)) = \sum_{u \in N^{\text{out}}(v)} \chi((v,u))$?}

\begin{lemma}[{Marx et al.~\cite[Lemma 86]{MarxSS16}}]
 Assuming ETH, \textsc{Edge Balancing} has no $f(k)n^{o(k)}$ time algorithm for any computable function $f$, where $k$ is the number of vertices of~$D$. Moreover, the problem is \Wh with respect to~$k$. 
\end{lemma}%
\problemQuestion{\textsc{Constrained Closed Walk}}
{An unweighted directed graph $G$ and a set $U \subseteq V(G)$.}
{Is there a closed walk (of any length) that visits every vertex at least once and each vertex in $U$ exactly once?}

The reduction from \textsc{Edge Balancing} to \textsc{Constrained Closed Walk} uses two gadgets that we now present.

\begin{figure}[ht!]
\begin{center}
 \begin{tikzpicture}[-{stealth'},every node/.style={draw,circle,inner sep=1.5pt}, every label/.style={hide},scale=.85]
 \node[fill, label=90:$a_{\text{in}}$] (ain) at (-1,3.5) {}; 
 \node[fill, label=90:$a_{\text{out}}$] (aout) at (15,3.5) {}; 
 
 \node[fill, label=160:$v^1_1$] (v11) at (0,0) {}; 
 \node[fill, label=180:$v^2_1$] (v21) at (0,1) {};
 \node[fill, label=180:$v^3_1$] (v31) at (0,2) {};
 \node[fill, label=110:$v^4_1$] (v41) at (1,2) {};
 \node[fill, label=180:$v^5_1$] (v51) at (1,1) {};
 \node[fill, label=270:$v^6_1$] (v61) at (1,0) {};
 \foreach \x/\y/\b in {1/2/20,2/1/20,2/3/20,3/2/20,3/4/0,4/5/20,5/4/20,5/6/20,6/5/20,1/6/0} {
    \draw (v\x1) to[bend left=\b] (v\y1);
    }
 \draw (ain) to (v31);
 \draw (v41) to (aout);
 
 \begin{scope}[xshift=2.5cm]
 \node[fill, label=270:$v^1_2$] (v12) at (0,0) {}; 
 \node[fill, label=180:$v^2_2$] (v22) at (0,1) {};
 \node[fill, label=200:$v^3_2$] (v32) at (0,2) {};
 \node[fill, label=340:$v^4_2$] (v42) at (1,2) {};
 \node[fill, label=180:$v^5_2$] (v52) at (1,1) {};
 \node[fill, label=270:$v^6_2$] (v62) at (1,0) {};
 \foreach \x/\y/\b in {1/2/20,2/1/20,2/3/20,3/2/20,3/4/0,4/5/20,5/4/20,5/6/20,6/5/20,1/6/0} {
    \draw (v\x2) to[bend left=\b] (v\y2);
    }
 \draw (ain) to (v32);
 \draw (v42) to (aout);
 \end{scope}
 
 \draw (v61) to (v12);
 
 \foreach \s in {3,4,5} {
 \begin{scope}[xshift=\s*2.5cm-2.5cm]
 \node[fill] (v1\s) at (0,0) {}; 
 \node[fill] (v2\s) at (0,1) {};
 \node[fill] (v3\s) at (0,2) {};
 \node[fill] (v4\s) at (1,2) {};
 \node[fill] (v5\s) at (1,1) {};
 \node[fill] (v6\s) at (1,0) {};
 \foreach \x/\y/\b in {1/2/20,2/1/20,2/3/20,3/2/20,3/4/0,4/5/20,5/4/20,5/6/20,6/5/20,1/6/0} {
    \draw (v\x\s) to[bend left=\b] (v\y\s);
    }
 \draw (ain) to (v3\s);
 \draw (v4\s) to (aout);
 \end{scope}
}
\draw (v62) to (v13);
\draw (v63) to (v14);
\draw (v64) to (v15);

 \begin{scope}[xshift=13cm]
 \node[fill, label=270:$v^1_s$] (v16) at (0,0) {}; 
 \node[fill, label=180:$v^2_s$] (v26) at (0,1) {};
 \node[fill, label=70:$v^3_s$] (v36) at (0,2) {};
 \node[fill, label=0:$v^4_s$] (v46) at (1,2) {};
 \node[fill, label=180:$v^5_s$] (v56) at (1,1) {};
 \node[fill, label=270:$v^6_s$] (v66) at (1,0) {};
 \foreach \x/\y/\b in {1/2/20,2/1/20,2/3/20,3/2/20,3/4/0,4/5/20,5/4/20,5/6/20,6/5/20,1/6/0} {
    \draw (v\x6) to[bend left=\b] (v\y6);
    }
 \draw (ain) to (v36);
 \draw (v46) to (aout);
 \end{scope}
 
 \draw[-,dotted] (v65) to (v16);
 
 \node[fill, label=270:$b_{\text{in}}$] (bin) at (-1,-.5) {}; 
 \node[fill, label=270:$b_{\text{out}}$] (bout) at (15,-.5) {}; 
 
 \draw (bin) to (v11);
 \draw (v66) to (bout);
 
 \end{tikzpicture}
\end{center}
\caption{Gadget $H_s$ used in the reduction from \textsc{Edge Balancing} to \textsc{Constrained Closed Walk}.}
\label{fig:gadget_Hs}
\end{figure}
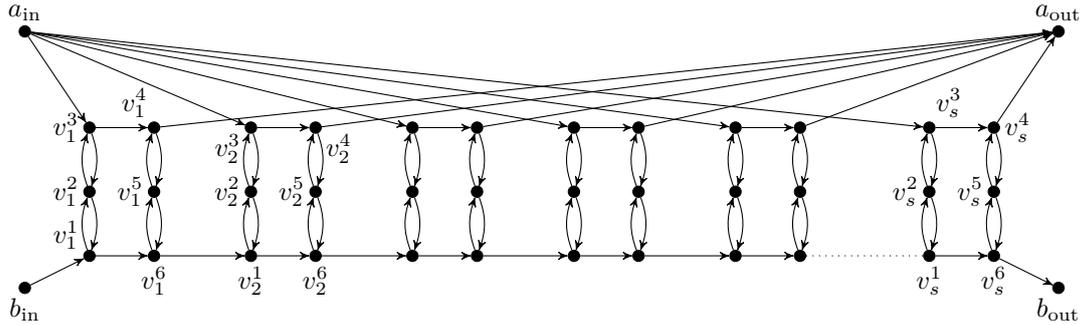

\begin{figure}[tbh]
\begin{center}
 \begin{tikzpicture}[-{stealth'},every node/.style={draw,circle,inner sep=1.5pt}, every label/.style={hide},scale=.85]
 \node[fill, label=90:$a_{\text{in}}$] (ain) at (-1,2) {}; 
 \node[fill, label=90:$a_{\text{out}}$] (aout) at (12,2) {}; 
 \node[fill, label=180:$v$] (v) at (-1,-1.5) {}; 
 \node[fill, label=270:$c_{\text{in}}$] (cin) at (-1.5,-2.1) {}; 
 \node[fill, label=270:$c_{\text{out}}$] (cout) at (12,-2.1) {}; 
 \draw (cin) to (v);
 \foreach \s in {1,2,3,4} {
  \node[rectangle,fill=black!20, minimum width=15mm, minimum height=10mm] (hx\s) at (3*\s-2,0) {$H_{x_\s}$};
  \draw (ain) to (hx\s.125);
  \draw (ain) to (hx\s.135);
  \draw (ain) to (hx\s.145);
  \draw (hx\s.35) to (aout);
  \draw (hx\s.45) to (aout);
  \draw (hx\s.55) to (aout);
  \draw (hx\s.south east) to (cout);
  \draw (v) to (hx\s.south west);
  } 
 \end{tikzpicture}
\end{center}
\caption{Gadget $H_X$ for the set $X=\{x_1,x_2,x_3,x_4\}$. The gray rectangles schematically represent the internal parts of the gadgets $H_{x_i}$, while the vertices $a_{\text{in}},a_{\text{out}}, v, c_{\text{out}}$ play the role of the vertices $a_{\text{in}},a_{\text{out}},b_{\text{in}},b_{\text{out}}$ for them, respectively.}
\label{fig:gadget_HX}
\end{figure}
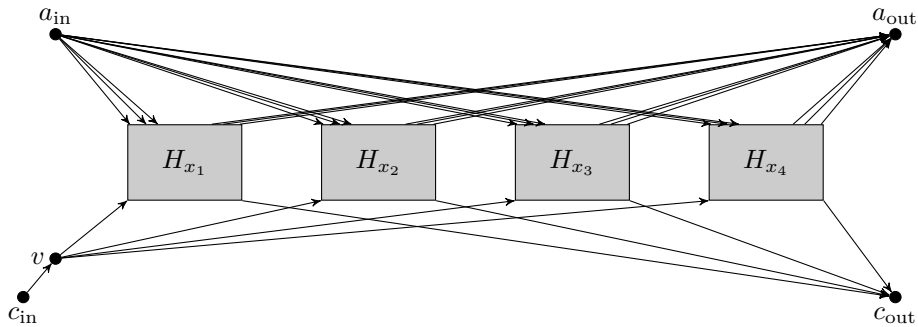

The gadgets use connecting vertices with subscripts 'in' or 'out' similarly as in \Cref{sec:d-to-td}, the other vertices are called \emph{internal}.

First, for a positive integer $s$ we construct a gadget $H_s$ depicted on \Cref{fig:gadget_Hs} with 4 connecting vertices: $a_{\text{in}}$, $a_{\text{out}}$, $b_{\text{in}}$, and $b_{\text{out}}$. 
As we rely on the original proof concerning the correctness, the only important property of the gadget for us is that the pathwidth of the graph $H_s \setminus \{a_{\text{in}},a_{\text{out}},b_{\text{in}},b_{\text{out}}\}$ is 2. 

Based on this gadget, for a set $X$ of positive integers, another gadget $H_X$ is constructed as follows (see \Cref{fig:gadget_HX}). 
First we introduce 4 connecting vertices $a_{\text{in}}$, $a_{\text{out}}$, $c_{\text{in}}$, and $c_{\text{out}}$, an internal vertex $v$, and an edge $\{c_{\text{in}},v\}$. 
Then for each $x \in X$ we introduce gadget $H_x$ and identify its connecting vertices $a_{\text{in}}$, $a_{\text{out}}$, $b_{\text{in}}$, and $b_{\text{out}}$ with $a_{\text{in}}$, $a_{\text{out}}$, $v$, and $c_{\text{out}}$, respectively.

Again the only important property of gadget $H_X$ for us is that the pathwidth of the graph $H_X \setminus \{a_{\text{in}},a_{\text{out}},c_{\text{in}},c_{\text{out}}\}$ is 3, while if we also remove vertex $v$, the pathwidth drops to~2.

Now we are ready to describe the reduction from \textsc{Edge Balancing} to \textsc{Constrained Closed Walk}.
Let $(D, (X_e)_{e \in E(D)})$ with $V(D) = \{w_1, \ldots, w_k\}$ be an instance of \textsc{Edge Balancing}.
We construct graph $D^*$ as follows.
First we introduce vertices of $D$ into $D^*$ together with two auxiliary vertices $c_{\text{in}}$ and $c_{\text{out}}$.
Let $Z = V(D) \cup \{c_{\text{in}},c_{\text{out}}\}$.
Then, for each edge $e=(w_i,w_j) \in E(D)$ we construct a copy of the gadget $H_{X_e}$, add it to $D^*$, and identify its vertices $a_{\text{in}}$, $a_{\text{out}}$, $c_{\text{in}}$, and $c_{\text{out}}$ with vertices $w_i$, $w_j$, $c_{\text{in}}$, and $c_{\text{out}}$, respectively.

Then we add some more paths into $D^*$.
Again, as we rely on the original paper concerning the correctness, the exact numbers of the paths are not important to us.
We add directed paths of length 2 (the internal vertices of them will be newly introduced) 
\begin{enumerate}
 \item from $c_{\text{in}}$ to $w_i$,
 \item from $w_i$ to $c_{\text{out}}$,
 \item and from $c_{\text{out}}$ to $c_{\text{in}}$.
\end{enumerate}

We let $U= V(D^*) \setminus Z$ be the set of vertices that should be traversed exactly once.
This finishes the construction of the instance $(D^*,U)$ of \textsc{Constrained Closed Walk}.
The correctness of the reduction follows from \cite[Lemma 90]{MarxSS16}.

Since each connected component of graph $D^* \setminus Z$ consist either of the internal vertices of a gadget $H_{X_e}$ or of a single internal vertex of the introduced path, the pathwidth of $D^* \setminus Z$ is 3. 
Hence, the distance of $D^*$ to pathwidth 3 is $O(k)$, where $k$ is the number of vertices of~$D$.

Moreover, let $Z'$ be the set $Z$ together with the vertex $v$ from every gadget $H_{X_e}$ introduced. 
Since we introduced gadget $H_{X_e}$ for every edge of $D$ and $D$ is a simple graph, we have that $|Z'| = O(k^2)$.
However, each connected component of graph $D^* \setminus Z'$ consist either of the internal vertices of a gadget $H_x$ or of a single internal vertex of the introduced path, and thus, the pathwidth of $D^* \setminus Z'$ is 2. 
Therefore, the distance of $D^*$ to pathwidth 2 is $O(k^2)$ and the reduction is parameterized reduction even to \textsc{Constrained Closed Walk} parameterized by distance to pathwidth 2.

Together we obtain the following.

\begin{lemma}
 \textsc{Constrained Closed Walk} is \Wh w.r.t. distance $d_2$ to pathwidth~2.
 Moreover, unless ETH fails, there is no algorithm for the problem with running time $f(d_2)n^{o(\sqrt{d_2})}$ or $f(d_3)n^{o(d_3)}$ for any computable function, where $d_3$ is the distance to pathwidth 3.
\end{lemma}

\begin{proof}[Proof of \Cref{thm:d_to_pw}]
To reduce \textsc{Constrained Closed Walk} to \DTSPshort Marx et al.~\cite{MarxSS16} suggested to use edge weights to ensure that all vertices in $U$ are visited at most once.
Note that if $n$ is the total number of vertices in the instance of \textsc{Constrained Closed Walk}, then any solution walk is of length at most $n^2$, as it consist of at most $n$ cycles, each of length at most~$n$.
Hence, if we give weight $2n^2$ to every edge with head in a vertex of $U$ and weight $1$ to every other edge and set the budget to $\budget=2n^2|U|+n^2$, then the obtained instance of \DTSPshort is equivalent to the instance of \textsc{Constrained Closed Walk}.
Indeed, any closed walk which visits each vertex of $D^*$ and each vertex of $U$ exactly once has weight at most $\budget$ and any closed walk that visits every vertex can visit each vertex of $U$ at most once, as otherwise it would exceed the budget.
This proves parts \itemstyle{a)}, \itemstyle{i)}, and \itemstyle{ii)} of the theorem.

\begin{figure}[tbh]
\begin{center}
 \begin{tikzpicture}[-{stealth'},every node/.style={draw,circle,inner sep=1.5pt}, every label/.style={hide},scale=.85]
 \node[fill, label=90:$a_{\text{in}}$] (ain) at (-1,3.5) {}; 
 \node[fill, label=90:$a_{\text{out}}$] (aout) at (15,3.5) {}; 
 
 \node[fill, label=270:$\widetilde{v^1_1}$] (v11i) at (0,0) {}; 
 \node[fill, label=180:$\widehat{v^2_1}$] (v21o) at (0,1) {};
 \node[fill, label=180:$\widetilde{v^3_1}$] (v31i) at (0,2) {};
 \node[fill, label=90:$\widehat{v^3_1}$] (v31o) at (1.5,2) {};
 \node[fill, label=0:$\widetilde{v^2_1}$] (v21i) at (1.5,1) {};
 \node[fill, label=270:$\widehat{v^1_1}$] (v11o) at (1.5,0) {};
 \node[fill, label=90:$\widetilde{v^4_1}$] (v41i) at (3,2) {};
 \node[fill, label=180:$\widehat{v^5_1}$] (v51o) at (3,1) {};
 \node[fill, label=270:$\widetilde{v^6_1}$] (v61i) at (3,0) {};
 \node[fill, label=90:$\widehat{v^4_1}$] (v41o) at (4.5,2) {};
 \node[fill, label=0:$\widetilde{v^5_1}$] (v51i) at (4.5,1) {};
 \node[fill, label=270:$\widehat{v^6_1}$] (v61o) at (4.5,0) {};
 
 \foreach \x in {1,...,6} {
    \draw[decorate, decoration={snake, segment length=1mm, amplitude=.5mm,post length=1mm}] (v\x1i) to (v\x1o);
 }
 
 \foreach \x/\y/\b in {1/2/20,2/1/20,2/3/20,3/2/20,3/4/0,4/5/20,5/4/20,5/6/20,6/5/20,1/6/0} {
    \draw (v\x1o) to (v\y1i);
    }
 \draw (ain) to (v31i);
 \draw (v41o) to (aout);
 
 \begin{scope}[xshift=6cm]
 \node[fill] (v12i) at (0,0) {}; 
 \node[fill] (v22o) at (0,1) {};
 \node[fill] (v32i) at (0,2) {};
 \node[fill] (v32o) at (1,2) {};
 \node[fill] (v22i) at (1,1) {};
 \node[fill] (v12o) at (1,0) {};
 \node[fill] (v42i) at (2,2) {};
 \node[fill] (v52o) at (2,1) {};
 \node[fill] (v62i) at (2,0) {};
 \node[fill] (v42o) at (3,2) {};
 \node[fill] (v52i) at (3,1) {};
 \node[fill] (v62o) at (3,0) {};
 
 \foreach \x in {1,...,6} {
    \draw[decorate, decoration={snake, segment length=1mm, amplitude=.5mm,post length=1mm}] (v\x2i) to (v\x2o);
 }
 
 \foreach \x/\y/\b in {1/2/20,2/1/20,2/3/20,3/2/20,3/4/0,4/5/20,5/4/20,5/6/20,6/5/20,1/6/0} {
    \draw (v\x2o) to (v\y2i);
    }
 \draw (ain) to (v32i);
 \draw (v42o) to (aout);
 \end{scope}
 
 \draw (v61o) to (v12i);
 
 \begin{scope}[xshift=11cm]
 \node[fill] (v13i) at (0,0) {}; 
 \node[fill] (v23o) at (0,1) {};
 \node[fill] (v33i) at (0,2) {};
 \node[fill] (v33o) at (1,2) {};
 \node[fill] (v23i) at (1,1) {};
 \node[fill] (v13o) at (1,0) {};
 \node[fill] (v43i) at (2,2) {};
 \node[fill] (v53o) at (2,1) {};
 \node[fill] (v63i) at (2,0) {};
 \node[fill] (v43o) at (3,2) {};
 \node[fill] (v53i) at (3,1) {};
 \node[fill] (v63o) at (3,0) {};
 
 \foreach \x in {1,...,6} {
    \draw[decorate, decoration={snake, segment length=1mm, amplitude=.5mm,post length=1mm}] (v\x3i) to (v\x3o);
 }
 
 \foreach \x/\y/\b in {1/2/20,2/1/20,2/3/20,3/2/20,3/4/0,4/5/20,5/4/20,5/6/20,6/5/20,1/6/0} {
    \draw (v\x3o) to (v\y3i);
    }
 \draw (ain) to (v33i);
 \draw (v43o) to (aout);
 \end{scope}
 
 \draw[-,dotted] (v62o) to (v13i);

 \node[fill, label=270:$b_{\text{in}}$] (bin) at (-1,-.5) {}; 
 \node[fill, label=270:$b_{\text{out}}$] (bout) at (15,-.5) {}; 
 
 \draw (bin) to (v11i);
 \draw (v63o) to (bout);
 
 \end{tikzpicture}
\end{center}
\caption{Modified gadget $H_s$. The ``snake'' paths represent directed paths of length $2n^2$ each.}
\label{fig:modified_gadget_Hs}
\end{figure}
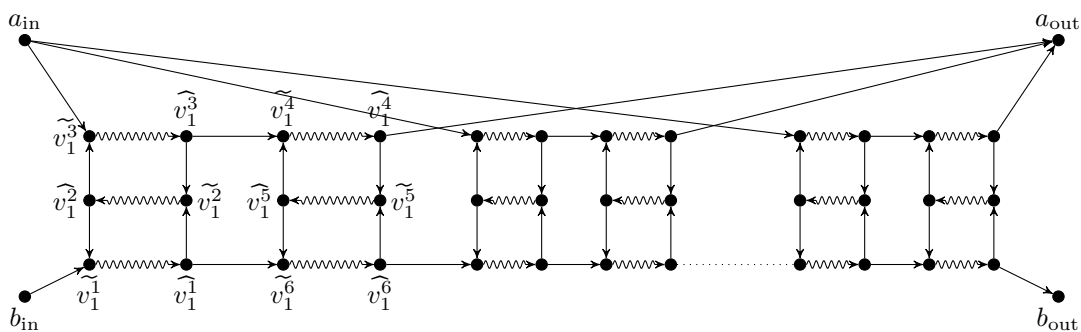

To prove the remaining parts of the theorem, that is, for unweighted graphs, we take a slightly different approach.
We divide every vertex $u \in U$ into two vertices $\widetilde{u}$ and $\widehat{u}$, redirect all the edges with head in $u$ to $\widetilde{u}$ and all the edges with tail in $u$ to $\widehat{u}$.
Then we connect $\widetilde{u}$ to $\widehat{u}$ by a path of length $2n^2$.
Denote the modified graph $\widehat{D}$ and set the budget again to $\budget=2n^2|U|+n^2$, to obtain an instance $(\widehat{D}, \mathbf{1},\budget)$ of \DTSPshort. 
\Cref{fig:modified_gadget_Hs} depicts the gadget $H_s$ from \Cref{fig:gadget_Hs} after this modification.

Let us now prove the equivalence of the instances $(D^*,U)$ of \textsc{Constrained Closed Walk} and $(\widehat{D}, \mathbf{1},\budget)$ of \DTSPshort.
If $C$ is a closed walk in $D^*$ that visit each vertex at least once and each vertex of $U$ exactly once, then by replacing each visit to a vertex $u \in U$ by the $\widetilde{u}$-$\widehat{u}$-path, we obtain a walk $\widehat{C}$ in $\widehat{D}$ of total length at most $|U|\cdot 2n^2+ n^2 = \budget$ visiting each vertex of $\widehat{D}$.
Conversely, any closed walk $\widehat{C}$ that visits each vertex in $\widehat{D}$ must visit vertex $\widetilde{u}$ for each $u \in U$. 
Whenever it visits it, it must continue along the path of length $2n^2$ to $\widetilde{u}$.
Hence, if it obeys the budget, it can visit each such vertex at most once. 
Thus, replacing each traverse of a $\widetilde{u}$-$\widehat{u}$-path with a visit to vertex $u$, we obtain a closed walk in $D^*$ which visits each vertex  and visit each vertex of $U$ exactly once.

It is easy to see that if the pathwidth of (part of) $D^*$ is $p$ then the pathwidth of (the corresponding part of) $\widehat{D}$ is at most $2p+3$.
We next prove more tight bounds.

Note that the modified gadget $H_s$ (\Cref{fig:modified_gadget_Hs}), after removal of connection vertices, has pathwidth 3, as it is a subgraph of a $3 \times \ell$ grid for a suitable $\ell$.
Thus, each modified gadget $H_{X_e}$ has pathwidth 4 after removal of the connection vertices, since we can put $\widehat{v}$ into every bag and the $\widetilde{v}$-$\widehat{v}$-path has pathwidth 1.

As $U = V(D^*) \setminus Z$ and we only modified vertices of $U$, we still have $Z \subseteq V(\widehat{D})$.
Since each connected component of graph $\widehat{D} \setminus Z$ consist either of the internal vertices of a modified gadget $H_{X_e}$ or of a single path resulting from a division of an internal vertex of the introduced path or of a vertex $v$, the pathwidth of $\widehat{D} \setminus Z$ is 4. 
Hence, the distance of $\widehat{D}$ to pathwidth 4 is $O(k)$, where $k$ is the number of vertices of $D$.
This implies part \itemstyle{iv)} of the theorem.

Moreover, apart from vertices in $Z$ only the vertices $\widehat{v}$ are adjacent to internal vertices of (modified) gadgets $H_s$.
Hence, we let $\widehat{Z}$ be the set $Z$ together with the vertex $\widehat{v}$ from every gadget $H_{X_e}$ introduced.
We again have that $|\widehat{Z}| = O(k^2)$.
Furthermore, each connected component of graph $\widehat{D} \setminus \widehat{Z}$ consist either of the internal vertices of a modified gadget $H_x$ or of a single path obtained from dividing the internal vertex of a path or of vertex $v$. 
Thus, the pathwidth of $\widehat{D} \setminus \widehat{Z}$ is 3. 
Therefore, the distance of $\widehat{D}$ to pathwidth 3 is $O(k^2)$.
This implies parts \itemstyle{b)} and \itemstyle{iii)} of the theorem.
\end{proof}
}%

\toappendix{
\section{Alternative Proof of W[1]-hardness with Respect to Distance to Constant Pathwidth}

In this section we present an independent alternative proof of the \Whness of the \DMTSP w.r.t.\ distance to pathwidth. Even though this proof may be somewhat redundant in view of~\cite{MarxSS16}, we still chose to include it here as it presents a different approach towards this result, which may be of interest to the potential reader.

\begin{theorem}\label{thm:pw_hardness}
    \DMTSP is \Wh with respect to distance to pathwidth 19 even if all weights are 1.
\end{theorem}

Gutin, Jones, and Wahlström~\cite{GJW16} have shown hardness of an auxiliary intermediate problem called \textsc{Properly Balanced Subgraph}.
We use the same problem to prove \Cref{thm:pw_hardness}.
We first introduce this problem.

 A directed multigraph is called {\em balanced} if the in-degree of each vertex equals its out-degree.
\problemQuestion{\textsc{Properly Balanced Subgraph (PBS)}}
{A directed multigraph $D = (V, A)$, a weight function $w: A \rightarrow \mathbb{Z}$,
a set $X = \{(a_1,a_1'), \dots, (a_r, a_r')\}$ %
          of disjoint pairs of arcs (called \emph{double arcs}), such
          that $a_i, a_i' \in A$ and $a_i,a_i'$ have the same start
          vertex and end vertex, for each $(a_i, a_i') \in X$;
         a set $Y = \{(b_1,b_1'), \dots, (b_s, b_s')\}$ %
         of disjoint pairs of arcs (called \emph{forbidden pairs}), such
         that $b_i,b_i' \in A$ and $b_i$ is the reverse of $b_i'$, for each
         $(b_i,b_i') \in Y$.
         Double arcs are disjoint from forbidden pairs.
         }
{Is there a balanced subgraph $D'$ of $D$ of negative weight such that
     $|A(D') \cap \{a_i,a_i'\}| \neq 1$ for each $(a_i, a_i') \in X$,
     and $|A(D') \cap \{b_i,b_i'\}| \neq 2$ for each $(b_i,b_i')\in Y$?}
Note that in this problem we are considering subgraphs, i.e., each arc can be contained in $D'$ at most as many times as it is contained in~$D$.

\begin{proposition}[{Gutin et al. \cite[Theorem 2]{GJW16}}]\label{prop:pbsHardness}
  \textsc{PBS} is W[1]-hard parameterized by \emph{distance to} pathwidth 16,
  even under the following restrictions:
  \begin{itemize}
   \item There exists a single arc $a^*$ of weight $-1$;
   \item $a^*$ is not part of a double arc;
   \item All other arcs have weight $0$;
   \item If there are at least two arcs between a pair of vertices, then they are exactly two and it is a double arc (in particular, there are no forbidden pairs).
  \end{itemize}
 \end{proposition}

 Gutin et al. prove the proposition by a reduction from \textsc{Multicolored Clique}. Although they formulate it for the parameterization by pathwidth, it is easy to observe (cf., ``$D^*$ has pathwidth at most 16'' \cite[Page 10--11]{GJW16}) that, in fact, the constructed graph is $2\ell$ vertices from being of pathwidth 16, where $\ell$ is the number of edges of the clique to be found.
 Also, it is easy to verify, that the reduction would also work starting from \textsc{Partitioned Subgraph Isomorphism}, giving an ETH lowerbound.

 As a next step, Gutin et al.~\cite{GJW16} provided a reduction from \textsc{PBS} (of the special form guaranteed by \Cref{prop:pbsHardness}) to the following problem, only increasing the pathwidth bound by 1.
 \problemQuestion{\textsc{Weighted Mixed Chinese Postman Problem (WMCPP)}}
{A  mixed multigraph $H = (V, A, E)$, where $A$ are arcs and $E$ are edges, a weight function $w: A \cup E \rightarrow \mathbb{N}$, a budget $b \in \mathbb{N}$}
{Is there a closed walk in $H$ of total weight at most $b$ which traverses each arc and each edge at least once?
     }

Since Gutin et al.\ do not prove the properties of this reduction separately but they include it as a part of a larger reduction, we repeat this part of the reduction here with some small modifications and simplifications for completeness.

\begin{lemma}\label{lem:wmcpp_hardness}
 \textsc{WMCPP} is W[1]-hard parameterized by distance to pathwidth 17,
  even if the resulting instance $(H=(V,A,E),w,b)$ satisfies the following restrictions:
  \begin{itemize}
   \item Each weight is either $2$, $3$, or $M$, where $M \le 3|A|+12$;
   \item weight of each edge is $M$;
   \item $b < M \cdot |w^{-1} (M)| +M -2$.
  \end{itemize}
\end{lemma}

\begin{proof}[Proof of \Cref{lem:wmcpp_hardness}]
We provide a reduction from \textsc{PBS} with the special properties as guaranteed by \Cref{prop:pbsHardness}.
We follow the approach of Gutin et al.~\cite{GJW16} with slight simplification and modification.
Let $(D, w, X, \emptyset)$ be an instance of \textsc{PBS} with the special properties as guaranteed by \Cref{prop:pbsHardness}.
Since there is only one weakly connected component containing the special arc of weight $-1$, we may ignore any parts of the solution in any other component.
I.e., we assume that $D$ is weakly connected.

We replace each arc, or double arc by a gadget which allows either \emph{a passive solution} (all arc and edges are traversed, but no imbalance is caused) or \emph{an active solution} (all vertices are traversed and the only disbalanced vertices are the original vertices---this corresponds to using the arc or the double arc).
The weights of these two solutions are the same for all gadgets except for the one replacing the special arc of negative weight, where the active solution is cheaper.

We present the gadgets in terms of normal arcs (weight 2), heavy arcs (weight $M$), and heavy edges (weight $M$).
We also introduce one arc of weight 3.
The value $M$ will be set later.
Due to the condition on $b$ any solution traverses each heavy arc and each heavy edge exactly once.
The gadgets are only connected by sharing the vertices of the original graph, i.e., all newly introduced vertices only have connections within their gadget.
Let us now present the gadgets.

\begin{figure}[tbh]
\centering
\begin{tikzpicture}[every node/.style={vertex},every label/.style={rectangle,hide},>={Stealth[length=4mm,width=3mm]}]
\node[draw=none, fill=none] at (1,-.7) {};
\node[label=90:$u$] (u) at (0,0) {};
\node[label=90:$v$] (v) at (2,0) {};
\draw[->, semithick] (u) to node[hide, above] {$0$} (v);
\end{tikzpicture}%
\hspace{2mm}%
\begin{tikzpicture}[every node/.style={vertex},every label/.style={rectangle,hide},>={Stealth[length=3mm,width=3mm]}]
\node[label=95:$u$] (u) at (0,0) {};
\node[label=below right:$z_{uv}$] (z) at (1.5,0) {};
\node[label=85:$v$] (v) at (3,0) {};
\draw[->, ultra thick, bend left=60] (u) to node[hide, above] {\footnotesize $M$} (z);
\draw[->, ultra thick, bend right=60] (u) to node[hide, below] {\footnotesize $M$} (z);
\draw[->, ultra thick, bend right=60] (v) to node[hide, above] {\footnotesize $M$} (z);
\draw[->, semithick] (z) to node[hide,above] {\footnotesize $2$} (u);
\draw[->, semithick] (z) to node[hide, above] {\footnotesize $2$} (v);
\end{tikzpicture}%
\hspace{1mm}%
\begin{tikzpicture}[every node/.style={vertex},every label/.style={rectangle,hide},>={Stealth[length=3mm,width=3mm]}]
\node[label=95:$u$] (u) at (0,0) {};
\node[label=below right:$z_{uv}$] (z) at (1.5,0) {};
\node[label=85:$v$] (v) at (3,0) {};
\draw[->, ultra thick, bend left=60] (u) to node[hide, above] {\footnotesize $M$} (z);
\draw[->, ultra thick, bend right=60] (u) to node[hide, below] {\footnotesize $M$} (z);
\draw[->, ultra thick, bend right=60] (v) to node[hide, above] {\footnotesize $M$} (z);
\draw[->, semithick, bend left=20] (z) to node[hide,above] {\footnotesize $2$} (u);
\draw[->, semithick, bend right=20] (z) to (u);
\draw[->, semithick] (z) to node[hide, above] {\footnotesize $2$} (v);
\end{tikzpicture}%
\hspace{1mm}%
\begin{tikzpicture}[every node/.style={vertex},every label/.style={rectangle,hide},>={Stealth[length=3mm,width=3mm]}]
\node[label=95:$u$] (u) at (0,0) {};
\node[label=below right:$z_{uv}$] (z) at (1.5,0) {};
\node[label=85:$v$] (v) at (3,0) {};
\draw[->, ultra thick, bend left=60] (u) to node[hide, above] {\footnotesize $M$} (z);
\draw[->, ultra thick, bend right=60] (u) to node[hide, below] {\footnotesize $M$} (z);
\draw[->, ultra thick, bend right=60] (v) to node[hide, above] {\footnotesize $M$} (z);
\draw[->, semithick, bend right=20] (z) to node[hide,above] {\footnotesize $2$} (v);
\draw[->, semithick, bend left=20] (z) to (v);
\draw[->, semithick] (z) to node[hide, above] {\footnotesize $2$} (u);
\end{tikzpicture}
\caption{Replacing a simple arc $(u,v)$ of weight 0. Left to right: original arc, the replacement gadget, passive solution, active solution.}
\label{fig:replacing-simple-arc}
\end{figure}
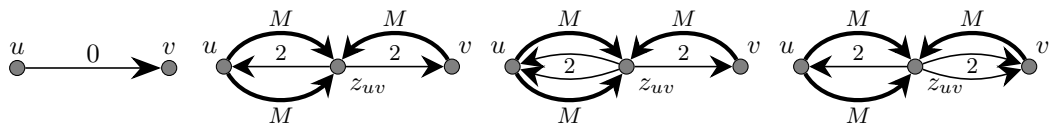

\begin{description}
\item[Arc $(u,v)$ of weight 0 that is not part of a double arc:]
To construct the gadget, we remove the original arc, add a new vertex $z_{uv}$, add normal arcs $(z_{uv}, u)$ and $(z_{uv},v)$, two heavy arcs $(u, z_{uv})$, and a heavy arc $(v, z_{uv})$ (see \Cref{fig:replacing-simple-arc}).

In any solution that traverses each heavy arc exactly once we have three incoming arcs to $z_{uv}$ so we have to use one of the outgoing normal arcs once and the other one twice to make $z_{uv}$ balanced.
\begin{itemize}
\item In the \emph{passive solution} we use $(z_{uv}, u)$ twice, which makes all vertices of the gadget balanced.
\item In the \emph{active solution} we use $(z_{uv}, v)$ twice. Then $u$ has one more outgoing arc than incoming, while $v$ has one more incoming than outgoing---which exactly corresponds to using
the arc $(u,v)$.
\end{itemize}
Both solutions use $3$ normal arcs and $3$ heavy ones, i.e., they have the same weight $3M+6$.

\item[The arc $(u,v)$ of weight -1:]
We use the above gadget, but make the arc $(z_{uv}, u)$ of weight~$3$.
Then the weight of the active solution is $3M+7$, whereas the weight of the passive one is $3M+8$.

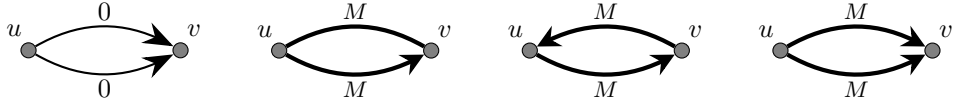
\begin{figure}[tbh]
\begin{center}
\begin{tikzpicture}[every node/.style={vertex},every label/.style={rectangle,hide},>={Stealth[length=4mm,width=3mm]}]
\node[label=95:$u$] (u) at (0,0) {};
\node[label=85:$v$] (v) at (2,0) {};
\draw[->, thick, bend left] (u) to node[hide, above] {$0$} (v);
\draw[->, thick, bend right] (u) to node[hide, below] {$0$} (v);
\end{tikzpicture}%
\hspace{6mm}%
\begin{tikzpicture}[every node/.style={vertex},every label/.style={rectangle,hide},>={Stealth[length=3mm,width=3mm]}]
\node[label=95:$u$] (u) at (0,0) {};
\node[label=85:$v$] (v) at (2,0) {};
\draw[ultra thick, bend left] (u) to node[hide, above] {\footnotesize $M$} (v);
\draw[->, ultra thick, bend right] (u) to node[hide, below] {\footnotesize $M$} (v);
\end{tikzpicture}%
\hspace{6mm}%
\begin{tikzpicture}[every node/.style={vertex},every label/.style={rectangle,hide},>={Stealth[length=3mm,width=3mm]}]
\node[label=95:$u$] (u) at (0,0) {};
\node[label=85:$v$] (v) at (2,0) {};
\draw[<-,ultra thick, bend left] (u) to node[hide, above] {\footnotesize $M$} (v);
\draw[->, ultra thick, bend right] (u) to node[hide, below] {\footnotesize $M$} (v);
\end{tikzpicture}%
\hspace{6mm}%
\begin{tikzpicture}[every node/.style={vertex},every label/.style={rectangle,hide},>={Stealth[length=3mm,width=3mm]}]
\node[label=95:$u$] (u) at (0,0) {};
\node[label=85:$v$] (v) at (2,0) {};
\draw[->, ultra thick, bend left] (u) to node[hide, above] {\footnotesize $M$} (v);
\draw[->, ultra thick, bend right] (u) to node[hide, below] {\footnotesize $M$} (v);
\end{tikzpicture}
\end{center}
\caption{Replacing a double arc $(u,v)$. Left to right: original arcs, the replacement gadget, passive solution, active solution.}
\label{fig:replacing-double-arc}
\end{figure}

\item[Double arc from $u$ to $v$:]
We remove the original arcs and introduce a heavy arc $(u,v)$ and a heavy edge $\{u,v\}$ (see \Cref{fig:replacing-double-arc}).
For a solution that traverses each of them exactly once, we only need to know, which direction the edge is traversed.
\begin{itemize}
\item In the \emph{passive solution} we traverse it from $v$ to $u$ which makes both vertices balanced.
\item In the \emph{active solution} we traverse it from $u$ to $v$, which creates two outgoing arcs from $u$ and two incoming arcs to $v$, corresponding to using the double arc.
\end{itemize}
Both solutions use one heavy arc and one heavy edge, i.e., they have the same weight~$2M$.
\end{description}

We replace each arc and double arc by the appropriate gadget and call the obtained mixed multigraph $H$ and the corresponding weight function $\widehat{w}$.
Suppose that $D$ had $m_1$ simple arcs of weight $0$, $m_2$ double arcs, and the one arc of weight $-1$.
Then, taking the passive solution for each gadget, we obtain a solution of total weight
\[m_1 \cdot (3M+6) + m_2 \cdot 2M +3M+8 = M \cdot (3m_1+2m_2+3)+ 6m_1+8.\]
We let $M=6m_1+12$ and $b = M \cdot (3m_1+2m_2+3)+ 6m_1+7$.
Note that there are $|\widehat{w}^{-1}(M)|=3m_1+2m_2+3$ heavy arcs and edges in $H$.
Therefore, we have \[b = M \cdot (3m_1+2m_2+3)+ 6m_1+7 < M \cdot (3m_1+2m_2+3)+ 6m_1+12 -2 = M \cdot |\widehat{w}^{-1}(M)| + M-2.\]
This implies that every solution of weight at most $b$ traverses each heavy arc and each heavy edge exactly once.

Let $Z \subseteq V(D)$ be such that $D \setminus Z$ has pathwidth at most 16.
We claim, that the pathwidth of $H \setminus Z$ is at most 17.
Indeed, given a path decomposition for $D\setminus Z$, to obtain a path decomposition of $H\setminus Z$, we only need to place the new vertices $z_{uv}$ into the decomposition.
To do so, for each arc $(u,v)$ we find a bag of the original decomposition which contains $\{u,v\} \setminus Z$, make a copy of this bag right after the original bag and place $z_{uv}$ into this new copy.
This way every arc and every edge appears in a bag and the width is only increased by $1$.

As the reduction can be clearly carried out in polynomial time, it remains to verify its correctness.
Let $D'$ be a balanced subgraph of $D$ of negative weight such that $|A(D') \cap \{a_i,a_i'\}| \neq 1$ for each $(a_i, a_i') \in X$.
We construct a solution $S$ for the constructed instance $(H,w,b)$ of \textsc{WMCPP} by taking the passive solution in each gadget replacing an arc or double arc not being part of $D'$ and the active solution for each arc or double arc included in $D'$.
Obviously, $S$ contains each edge and each arc at least once.
Since the disbalance caused by $S$ at each original vertex of $D$ is exactly the same as the disbalance caused by the arcs of $D'$, $S$ is a balanced subgraph.
Since $D$ was weakly connected and $S$ traverses all arcs and edges of $H$, $S$ is weakly connected.
Therefore, it corresponds to a walk traversing each arc and each edge at least once.
Moreover, since $D'$ is of negative weight, $a^* \in A(D')$ and the active solution is used in the gadget replacing arc $a^*$.
Thus, the total weight of the solution is
\[m_1 \cdot (3M+6) + m_2 \cdot 2M +3M+7 = M \cdot (3m_1+2m_2+3)+ 6m_1+7=b.\]
Hence, $(H,w,b)$ is a yes-instance of \textsc{WMCPP}.

Conversely, if there is a solution $S$ for $(H,w,b)$, then, due to the budget constraint, it traverses every heavy edge and every heavy arc exactly once.
As vertices $z_{uv}$ must be balanced with respect to $S$, it follows, that $S$ either takes the passive or the active solution in each gadget.
In particular, it must take the active one in the gadget replacing $a^*$, as otherwise it would exceed the budget.
Now taking $D'$ as the subgraph of $D$ formed by the arcs and double arcs whose gadgets use the active solution in $S$, we obtain a balanced subgraph of $D$, since the imbalance for each vertex is the same in $S$ as in $D'$.
The subgraph $D'$ also respects double arcs by construction.
As $a^* \in D'$, $D'$ has a negative weight.
Therefore, $(D, w, X, \emptyset)$ is a yes-instance, finishing the proof.
\end{proof}

\begin{proof}[Proof of \Cref{thm:pw_hardness}]
We provide a parameterized reduction from \textsc{WMCPP} parameterized by distance to pathwidth~17.
Let $(H,w,b)$ be an instance with the special properties guaranteed by \Cref{lem:wmcpp_hardness}.

To produce an instance of \DMTSPshort, we replace each arc of weight $h$ by a directed path of length $h$.
We replace a heavy edge of weight $M$ by a bidirected path of length $M$, that is, we take a directed path of length $M$ from one endpoint to the other and then add the opposite arc to every arc of the path.
Note that there are no edges of weight other than $M$ in $H$ by assumption.
We denote the resulting directed graph $G$.

The reduction can be clearly carried out in polynomial time, as $M \le 3|A|+12$.

Let $Z \subseteq V(H)$ be such that $H \setminus Z$ has pathwidth at most 17.
Given a path decomposition for $H \setminus Z$ of width $17$, we can obtain a path decomposition of~$G \setminus Z$ of width at most $19$ in the following way.
For each arc or edge $(u,v)$ of $H$ of weight~$h$ find a bag in which $\{u,v\}  \setminus Z$ appears.
If $h=2$, then introduce a copy of the bag right after the original bag, and add the inner vertex of the path replacing the arc $(u,v)$ into that copy.
Otherwise introduce $h-2$ copies of the bag, subsequently containing adjacent pairs of inner vertices of the replacement path.
It is easy to verify that we indeed obtain a path decomposition of $G \setminus Z$.
Therefore, the pathwidth of $G \setminus Z$ is at most $2$ more than the pathwidth of~$H \setminus Z$ and the reduction preserves the parameter.

Finally, we claim that the instance $(H,w,b)$ of \textsc{WMCPP} is equivalent to the instance $(G, \mathbf{1}, b)$ of \DMTSPshort, where $\mathbf{1}$ assigns weight $1$ to every arc of $G$.
If $S$ is a solution for $H$, then replacing each traversal of an edge or an arc by the traversal of the corresponding path in $G$ we obtain a walk of the same weight that visits every vertex at least once.
This is the case for the vertices of $H$, since each of them is incident with at least one (traversed) arc or edge.
For the inner vertices of the paths it follows from the fact that each path is traversed at least once (in at least one of the directions).

Conversely, suppose that we have a walk $S'$ in $G$ of weight at most $b$ that visits all vertices of $G$.
First, to visit an inner vertex of a directed path, the walk has to traverse the whole path at least once.
Now, consider a bidirected path of length $M$.
As each inner vertex has to be visited, $S'$ contains for each such vertex $w$ at least one arc with head in $w$ and one arc with tail in $w$.
To reach also other vertices, the walk must also contain an arc of the path with tail in one of the endpoints and an arc of the path with head in one of the endpoints.
It follows that altogether $S'$ contains at least $M$ arcs from each such path.
If neither of the two arcs between two consecutive vertices of the path would be contained in~$S'$, then both arcs would have to be contained for all other pairs.
Therefore, $S'$ would contain at least $2M-2$ arcs from that path, i.e., $M-2$ more than $M$, which cannot happen as $b < M \cdot |w^{-1} (M)| +M-2$.
Hence, $S'$ traverses the path in at least one direction completely.
As it cannot traverse it completely twice, we can omit any other visits to the path.

Now consider a walk $S$ in $H$ obtained from $S'$ by replacing each traversal of a path by a traversal of the corresponding arc or edge of $H$.
By the above argument $S$ visit every edge and every arc of $H$ and the weight of $S$ is at most that of $S'$, which is at most $b$.
Hence, $(H,w,b)$ is a yes-instance, finishing the proof.
\end{proof}

}%

\section{Conclusions and Open Problems}

We have shown that \DWRPshort is \FPT w.r.t.\ vertex integrity as well as w.r.t.\ feedback edge set number, while it is \Wh{} w.r.t.\ distance to constant treedepth even in unweighted graphs.
\DTSPshort is \Wh w.r.t.\ distance to pathwidth~3, also even in unweighted graphs and we conjecture that the above hardness w.r.t.\ distance to constant treedepth holds also for this problem.
An immediate open problem stemming from our research is to close the gap by considering distances to even smaller constant values of the parameters.
Namely, we would like to know whether the problems are \FPT w.r.t. feedback vertex set (distance to treewidth~1), distance to disjoint paths (pathwidth~1), or distance to stars (treedepth~2).

We have also shown, that the problems are \FPT parameterized by treewidth combined with the number of visits.
Given that all the hardness reductions leverage vertices of high degrees, we are wondering whether a combination of treewidth with the maximum degree of the input graph might lead to an fpt-algorithm.

In this paper, we focused on structural parameters of the underlying undirected graph.
Nevertheless, there are many width measure actually designated for directed graphs~\cite{GanianHK0ORS16} and the tractability w.r.t. them is open.
In fact, currently we do not even know whether \DTSPshort is polynomial time solvable or \NPh on directed graphs that become acyclic (DAGs) after deleting a single vertex or even a single edge.

\bibliography{main}

\clearpage

\appendix
\appendixText

\end{document}